\documentclass[12pt,reqno]{amsart}
\usepackage[utf8]{inputenc}
\usepackage[english]{babel}
\usepackage{amsfonts,amssymb,amsmath,mathtools,braket,amssymb,dsfont,empheq,stmaryrd}
\usepackage{enumerate,enumitem}
\usepackage{manfnt}
\usepackage[table, dvipsnames]{xcolor}
\usepackage{tikz}
\newcommand{\absenceSpectrum}{Red}
\newcommand{\existingSpectrum}{Black}
\newcommand{\possibleSpectrum}{Black}
\newlength\runit
\theoremstyle{plain}
\newtheorem{theorem}{Theorem}[section]
\newtheorem{lemma}[theorem]{Lemma}
\newtheorem{corollary}[theorem]{Corollary}
\newtheorem{proposition}[theorem]{Proposition}

\theoremstyle{definition}
\newtheorem{remark}[theorem]{Remark}
\newtheorem{assumption}[theorem]{Assumption}

\renewcommand{\Re}{\operatorname{Re}}
\renewcommand{\Im}{\operatorname{Im}}
\renewcommand{\leq}{\leqslant}
\renewcommand{\geq}{\geqslant}

\newcommand{\N}{\mathbb{N}} 
\newcommand{\C}{\mathbb{C}} 
\newcommand{\Z}{\mathbb{Z}} 
\newcommand{\R}{\mathbb{R}} 

\newcommand{\cF}{\mathcal{F}}
\newcommand{\cH}{\mathcal{H}}

\newcommand{\abs}[1]{\left\vert #1 \right\vert}

\newcommand{\pscal}[1]{\ensuremath{\left\langle #1 \right\rangle}} 
\newcommand{\pscalns}[1]{\ensuremath{\langle #1 \rangle}} 
\newcommand{\norm}[1]{\left\vert\kern-0.25ex\left\vert #1 \right\vert\kern-0.25ex\right\vert}
\newcommand{\normt}[1]{\left\vert\kern-0.25ex\left\vert\kern-0.25ex\left\vert #1 \right\vert\kern-0.25ex\right\vert\kern-0.25ex\right\vert}

\newcommand*{\transp}{^{\mkern-1.5mu\mathsf{T}}}
\newcommand{\prim}[1]{ #1^\prime} 
\newcommand{\di}{\,\mathrm{d}}

\DeclareMathOperator{\Id}{\mathds{1}}
\DeclareMathOperator{\sgn}{sign}
\DeclareMathOperator{\arctanh}{arctanh}
\DeclareMathOperator{\vect}{span}
\DeclareMathOperator{\ran}{ran}

\usepackage{float}
\usepackage{graphicx}
\usepackage{subcaption}
		
\usepackage{hyperref}

\title[On spectral stability of Soler standing waves]{Results on the spectral stability of standing wave solutions of the Soler model in 1-D}

\author[D. Aldunate]{Danko Aldunate}
\address{Instituto de F\'isica, Pontificia Universidad Cat\'olica de Chile, Vicu\~na Mackenna 4860, Santiago 7820436, Chile.}
\email{dmaldunate@uc.cl}

\author[J. Ricaud]{Julien Ricaud}
\address{Department of Mathematics, LMU Munich, Theresienstrasse 39, 80333 Munich, and Munich Center for Quantum Science and Technology, Schellingstr. 4, 80799 Munich, Germany}
\email{julien.ricaud@polytechnique.edu}

\author[E. Stockmeyer]{Edgardo Stockmeyer}
\address{Instituto de F\'isica, Pontificia Universidad Cat\'olica de Chile, Vicu\~na Mackenna 4860, Santiago 7820436, Chile.}
\email{stock@fis.puc.cl}

\author[H. Van Den Bosch]{Hanne Van Den Bosch}
\address{Departamento de Ingenier\'ia Matem\'atica and Centro
de Modelamiento Matem\'atico (CNRS IRL 2807), Universidad de Chile, Beauchef 851, Santiago, Chile}
\email{hvdbosch@dim.uchile.cl}

\date{\today}

\begin{document}

\begin{abstract}
    We study the spectral stability of the nonlinear Dirac operator in dimension~$1+1$, restricting our attention to nonlinearities of the form $f(\pscal{\psi,\beta \psi}_{\C^2}) \beta$. We obtain bounds on eigenvalues for the linearized operator around standing wave solutions of the form $e^{-i\omega t} \phi_0$. For the case of power nonlinearities $f(s)= s |s|^{p-1}$, $p>0$, we obtain a range of frequencies $\omega$ such that the linearized operator has no \emph{unstable eigenvalues} on the axes of the complex plane. As a crucial part of the proofs, we obtain a detailed description of the spectra of the self-adjoint blocks in the linearized operator. In particular, we show that the condition $\pscal{\phi_0,\beta \phi_0}_{\C^2} > 0$ characterizes \emph{groundstates} analogously to the Schr{\"o}\-dinger case. 
\end{abstract}

\maketitle
\tableofcontents

\section{Introduction}
\noindent
We consider the nonlinear Dirac equation in dimensions~$1+1$ with Soler-type nonlinearity $f:\R\to \R$ and  initial data $\phi \in H^1(\R, \C^2)$, given by
\begin{equation}\label{eq:time-dep}
    \begin{cases}
        i\partial_t \psi =  D_m \psi - f\!\left(\pscal{\psi,\beta \psi}_{\C^2}\right) \beta \psi\,, \\
        \psi(\cdot,0) = \phi\,,
    \end{cases} 
\end{equation} 
where $\psi = (\psi_1, \psi_2)\transp : \R\times \R \to \C^2$ and $D_m := -i\alpha \partial_x + m \beta$ is the one-dimensional Dirac operator with mass $m>0$.
In this paper, we chose the convention $(\alpha,\beta)= (- \sigma_2,\sigma_3)$, where the $\sigma_j$'s are the standard Pauli matrices
\[
    \sigma_1=
    \begin{pmatrix}
        0&1\\
        1&0
    \end{pmatrix},
    \qquad
    \sigma_2=
    \begin{pmatrix}
        0&-i\\
        i&0
    \end{pmatrix}
    \quad \textrm{ and } \quad
    \sigma_3=
    \begin{pmatrix}
        1&0\\
        0&-1
    \end{pmatrix}.
\]
Other choices for $\alpha\neq\beta$ among~$\pm\sigma_1$, $\pm\sigma_2$, and~$\pm\sigma_3$ in this equation lead to unitary equivalent linearized operators (Dirac--Pauli theorem, see, e.g., \cite{Dirac-28a,Pauli-36,Thaller-92,ComGuaGus-14}). The advantage of the present choice is that no complex coefficients appear in the r.h.s.\ of~\eqref{eq:time-dep}.

This type of models describes the  dynamics of spinors 
self-interacting through a mass-like (also called Lorentz-scalar) potential. They were introduced in~\cite{Ivanenko-38} and further developed in~\cite{Soler-70}. The case $f(s)=s$ corresponds to the massive Gross--Neveu model~\cite{GroNev-74}. 

We assume that the nonlinearity $f$ satisfies
\begin{assumption}\label{Assumption_general_nonlinearity}
    $f\in \mathcal{C}^1(\R\setminus\{0\}, \R) \cap \mathcal{C}^0(\R, \R)$ with $f(0)=0$, $\lim_{+\infty} f\geq m$, and $f'>0$ on~$(0,+\infty)$. Moreover,~$\lim_{s\to0^+} s f'(s) = 0$.
\end{assumption}

The first part of this assumption ensures that equation~\eqref{eq:time-dep} admits standing wave solutions of the form
\begin{equation}\label{eq:solitary-wave-def}
    \psi(x,t) = e^{-i\omega t} \phi_0(x)
\end{equation}
for all $\omega \in (0,m)$, where the initial condition $\phi_0\equiv(v, u)\transp \in H^1(\R,\C^2)$ solves
\begin{equation}\label{eq:phi_0}
    \left(D_m -\omega \Id\right)\phi_0 - f(\pscal{\phi_0, \sigma_3 \phi_0}_{\C^2})\sigma_3 \phi_0 = 0\,,
\end{equation}
decays exponentially with rate $\sqrt{m^2-\omega^2}$, and can be chosen real-valued with $v$ even, $v(0)>0$, and $u$ odd.
See~\cite[Lemma 3.2]{BerCom-12} and~\cite{CazVaz-86} for higher dimensions.
In the following, we always assume that $\phi_0$ is this solution. As we will prove in Proposition~\ref{prop:basic_groundstate_properties}, $\pscal{\phi_0,\sigma_3 \phi_0}_{\C^2} = v^2 - u^2 > 0$ and $\phi_0 \in \mathcal{C}^2(\R)$.
Notice that $\phi_0 \equiv \phi_0(\omega)$ depends on~$\omega$, which is in $(0,m)$ throughout this paper. If $f$ is an even function, then $\sigma_1 \phi_0(\omega)$ solves~\eqref{eq:phi_0} for $-\omega$ and we can treat negative frequencies as well in this case. 

In this paper we study the spectral stability of the solitary wave $\phi_0$ and obtain new spectral properties of the linearized operator of equation~\eqref{eq:time-dep} around $\phi_0$. We work in the Hilbert space~$L^2(\R,\C^2)$ with scalar-product
\[
	\pscal{f,g} = \int_\R \pscal{f(x),g(x)}_{\C^2} \di x\,.
\]
In order to set up the problem, let us define the family  of operators in~$L^2(\R,\C^2)$
\begin{equation}\label{Def_L_mu}
    L_\mu\equiv L_\mu (\omega) := D_m -\omega\Id - f\!\left(\pscal{\phi_0,\sigma_3\phi_0}_{\C^2}\right)\sigma_3 -\mu Q, \quad \textrm{ with domain } H^1(\R,\C^2)\,,
\end{equation}
parametrized by $\mu \in \R$, where $Q$ acts as the matrix-valued multiplication operator
\begin{equation}\label{Def_of_Q}
    Q := f'(\pscal{\phi_0,\sigma_3\phi_0}_{\C^2})  \left(\sigma_3\phi_0\right) \left(\sigma_3\phi_0\right)\transp.
\end{equation}
When needed for clarity, we will highlight the dependency of $Q$ in $\omega$ by writing $Q_\omega$.
We also define the family of operators
\begin{equation}\label{eq:def_H}
    H_\mu \equiv H_\mu (\omega) := \begin{pmatrix} 0 & L_0 \\ L_\mu & 0 \end{pmatrix}, \qquad \textrm{ with domain } H^1(\R, \C^4)\,.
\end{equation}
For any $\mu\in\R$, $H_\mu$ is a well-defined closed operator, see e.g.~\cite{kato}, since $H_0$ is self-adjoint and the operator $H_\mu - H_0$ is bounded.

Actually, only the operators $L_0$ and~$L_2$ appear in the linearization of the nonlinear equation~\eqref{eq:time-dep}. Indeed, looking for solutions to~\eqref{eq:time-dep} of the form 
\[
	\psi(x,t)= e^{-i\omega t}\left(\phi_0(x) + \rho(x,t)\right),
\]
the formal linearization of~\eqref{eq:time-dep} around the solitary wave solution $\phi_0$ takes the form
\[
    i\partial_t \rho = L_0 \rho - 2 f'(\pscal{\phi_0,\sigma_3\phi_0}_{\C^2}) \Re\left(\pscal{\phi_0,\sigma_3\rho}_{\C^2}\right)\sigma_3\phi_0\,.
\]
Using that $\pscal{\alpha,\beta}\alpha=\alpha\overline\alpha\transp\beta$ for any $\alpha,\beta\in\C^2$, and that $\phi_0$ is real-valued, this equation can be written as $i\partial_t \rho = i L_0 \Im \rho + L_2 \Re \rho$ or, equivalently,
\begin{equation}\label{Linearization_H}
    i\partial_t
    \begin{pmatrix}
        \Re \rho \\
        i \Im \rho
    \end{pmatrix}
    = H_2
    \begin{pmatrix}
        \Re \rho \\
        i \Im \rho
    \end{pmatrix}.
\end{equation}

A novelty in our approach towards spectral stability is to extend the study to the spectral properties of the analytic operator families $\mu\mapsto L_\mu$ and $\mu\mapsto H_\mu$ for $\mu\geq0$.

Note that in the linearized equation~\eqref{Linearization_H}, we take the convention to keep the complex number $i$ multiplying the left hand side,
which seems at odds with the most usual convention in the PDE literature. As a consequence, $H_0$ is self-adjoint, and $H_\mu$ has essential spectrum on the \emph{real} axis for all $\mu\in\C$. With this convention, \emph{spectral stability} corresponds to the absence of eigenvalues of~$H_2$ with \emph{positive imaginary part} (compare, for instance, with the definition given in~\cite{BouCom-16}). Since, as we will see, the spectrum of~$H_2$ is symmetric with respect to the real axis, spectral stability amounts to all eigenvalues of~$H_2$ being real.

In the case of the nonlinear Schr{\"o}\-dinger and Klein--Gordon operators, the spectral and orbital stability of solitary waves are well-understood since major breakthroughs in the~`80~\cite{GriShaStr-87,Weinstein-86}. However, this is not the case for the Dirac analogues.
The main difficulties are related to the lack of positivity of the Dirac operator. A notable exception is~\cite{PelShi-14}, where the authors prove orbital stability in the one-dimensional massive Thirring model, which is completely integrable.

In the Schr{\"o}\-dinger case, spectral stability is crucial to characterize the orbital stability and, together with some additional assumptions, also implies asymptotic stability (see for instance~\cite{Cuccagna-11}). 
For the Dirac equation, this connection is not clear. Nevertheless, asymptotic stability of small amplitude solitary waves (that is, $\omega$ close to~$m$) in dimension~$3$ is shown in~\cite{BouCuc-12} to follow from spectral stability, under several technical assumptions.

Spectral stability in the Soler model with nonlinearity $f(s)= s |s|^{p-1}$ for dimensions~$1$, $2$, and $3$ is studied in~\cite{BouCom-19, ComGuaGus-14}. In these works, results are obtained in the \emph{non-relativistic limit} $\omega$ going to $m$, using the convergence to the corresponding nonlinear Schr{\"o}\-dinger equation.  In~\cite{BouCom-19}, an interval in $\omega$ of spectral stability is shown to exist for powers $1< p \leq 2$. With similar methods, the model is shown in~\cite{ComGuaGus-14} to be spectrally unstable in the non-relativistic limit for $p>2$ in dimension~$1$.
We are not aware of analytical results for the case $p \leq 1$, but spectral stability of solitary waves for the one-dimensional Soler model with $p=1$ (massive Gross--Neveu model) has been studied numerically in~\cite{BerCom-12,Lakoba-18}.
The general conjecture seems to be that this model is spectrally stable, although there has been some controversy~\cite{Lakoba-18} for the case of small frequencies~$\omega$. 

Still in dimension~$1$, but when translation invariance is broken by a potential added to the Dirac operator, spectral and asymptotic stability for large $\omega$ are proved in~\cite{PelSte-12}. On the other hand, when translation invariance is broken by switching off the nonlinearity away from the origin, a recent paper~\cite{BouCacCarComNojPos-20} shows spectral stability and instability by explicit computations. Interestingly, as long as the nonlinearity preserves a parity symmetry, all eigenvalues are real or purely imaginary.
For a more complete account on the spectral stability of Dirac equations, we refer the reader to the recent monograph~\cite{Book-BouCom-19} by Boussa\"{i}d and Comech, and the references therein.  

Before describing our results, we recall some well-known properties of the operators~$L_\mu$ and~$H_\mu$, see e.g.~\cite{Book-BouCom-19}. For the convenience of the reader, we give a proof of those in Section~\ref{Section_prelim}. See also Figures~\ref{fig:spectra} and~\ref{fig:spectrum_H_2}, where the spectra are sketched.

\begin{proposition}\label{WellKnownProperties}
    Let $f$ satisfy Assumption~\ref{Assumption_general_nonlinearity}, $\omega\in(0,m)$, $\mu\in\R$, $L_\mu$ be defined in~\eqref{Def_L_mu}, and $H_\mu$ be defined in~\eqref{eq:def_H}. Then,
    \begin{enumerate}[label=(\roman*)] 
        \item\label{WellKnownProperties_Lmu_ess_spec} $\sigma_{\rm ess}(L_\mu) = (-\infty, -m-\omega] \cup [m-\omega, +\infty)$. 
        \item\label{WellKnownProperties_L0_EVs} The spectrum of~$L_0$ is symmetric with respect to $-\omega$. $-2\omega$ and~$0$ are simple eigenvalues of~$L_0$ with respective eigenfunctions $\phi_{-2\omega} := \sigma_1 \phi_0$ and~$\phi_0$.
        \item\label{WellKnownProperties_L2_EVs} $-2\omega$ and~$0$ are simple eigenvalues of~$L_2$ with eigenfunctions $\phi_{-2\omega}$ and~$\partial_x \phi_0$.
        \item\label{WellKnownProperties_Hmu_ess_spec} $\sigma_{\rm ess}(H_\mu) =(-\infty, -m +\omega] \cup [m-\omega, +\infty)$.
        \item\label{WellKnownProperties_Hmu_symm_spect_axis} The spectrum of~$H_\mu$ is symmetric with respect to the real and imaginary axes.
        \item\label{WellKnownProperties_H2_EVs} $\pm 2\omega$ are eigenvalues of~$H_2$ and~$0$ is a double eigenvalue. 
    \end{enumerate}
\end{proposition}

\subsection{Main results}\label{Section_Main_results}
We start with a qualitative description of our results. Note that the operator family $\mu \mapsto H_\mu$, $\mu\in[0,2]$, has two types of eigenvalue branches $\mu \mapsto z(\mu)$: (a)~those that are on the axes of the complex plane, i.e, $z(\mu)^2 \in \R$ for all values of $\mu$ where the branch is well-defined, and (b)~those that leave the axes for some values of $\mu$. Remarkably, in the Schr{\"o}\-dinger case, the latter are excluded as a consequence of the semi-boundedness of the blocks in the corresponding  linearization operator. Still in the Schr{\"o}\-dinger case, linear stability is characterized by the~\emph{Vakhitov--Kolokolov criterion}~\cite{Kolokolov-73,VakKol-73}. We emphasize that the possible existence of eigenvalues off the axes is one of the main differences between Dirac and Schr{\"o}\-dinger cases.

In this paper, we provide
\begin{enumerate}[label=\arabic*.]
	\item detailed information on the spectra of $L_0$ and $L_2$;
	\item a generalization of the Vakhitov--Kolokolov criterion to the Dirac context, as a condition for eigenvalue branches of type~(a) to stay on the real axis;
	\item bounds on eigenvalue branches of type~(b) implying that they can only occur close to the outer~\emph{thresholds}~$\pm (m+\omega)$ when $\omega$ is sufficiently close to~$m$.
\end{enumerate}
The missing step for a proof of spectral stability is to completely exclude eigenvalue branches of type~(b). Still, we believe that this work is an important step forward. Indeed, for the one-dimensional case, eigenvalue branches of type~(b) have not been observed in the numerical results available in the literature.\footnote{According to numerical work in~\cite{CueBouComLanKevSax-18}, this type of eigenvalue branches do occur in higher dimensions for non-radial perturbations.}

As an application of these results, we obtain precise conditions on $\phi_0$ such that the corresponding $H_2$ has no eigenvalues on the imaginary axis.

We now give precise statements of the main theorems.
For point 1., we start by showing in Section~\ref{Section_prelim} that the solitary wave solutions are \emph{groundstates} in the sense that they correspond to the smallest eigenvalue (in absolute value) of the operator $L_0 + \omega\Id$.
\begin{theorem}\label{groundstate}
    Let $f$ satisfy Assumption~\ref{Assumption_general_nonlinearity}, $\omega\in(0,m)$, and $L_0$ be as in~\eqref{Def_L_mu}. Then, $L_0$ has no eigenvalues in~$(-2\omega, 0)$. 
\end{theorem}

Through a perturbation argument (see Lemma~\ref{A_priori_bounds_on_first_ev_of_Lmu}), this result implies that $L_2$ has at least one eigenvalue in $(-2\omega, 0)$. See Figure~\ref{fig:spectra} for an illustration.
\begin{figure}[ht]
	\centering
	\begin{tikzpicture}[scale=1.5,
	    eje/.style = {-latex },
	    known/.style = {ultra thick,\existingSpectrum},
	new/.style = { thick, violet}]

	    \shade[right color=\existingSpectrum!60,left color=\existingSpectrum!10] (-1,0) -- (-2,1) --(-2.2, 1) -- (-2.2,0)--cycle;
	    \shade[left color=\existingSpectrum!60,right color=\existingSpectrum!10] (1,0) -- (0,1) --(1.9, 1) -- (1.9,0)--cycle;
	    \fill[color=\absenceSpectrum!15] (0,0) -- (-2,1) --(0,1)--cycle; 
	
	    \draw[eje](-2.25, 0)--(2,0) node [below] {\footnotesize $\sigma(L_0(\omega))$};
	    \draw[eje](0,-0.2) --(0,1.2) node[right]{\footnotesize $\omega$} ;
	    
	    \draw[known](0,0) --(0,1) ;
	    \draw[known](0,0) --(-2,1) ;
	    \draw[shift={(1,0)}] (0pt,0pt) -- (0pt,-2pt) node[below] {\tiny$+m$};
	    \draw[shift={(-1,0)}] (0pt,0pt) -- (0pt,-2pt) node[below] {\tiny$-m$};
	    \draw[shift={(-4pt,-4pt)}] (0pt,0pt) -- (0pt,-0pt) node[] {\tiny$0$};
	    \draw[shift={(0,1)}] (2pt,0) -- (-2pt,0) node[xshift=-4pt, yshift=4pt] {\tiny$m$};  
	\end{tikzpicture} 
	%
	\begin{tikzpicture}[scale=1.5,
	    eje/.style = {-latex },
	    known/.style = {ultra thick,\existingSpectrum},
	new/.style = { thick, violet}]
	
	    \shade[right color=\existingSpectrum!60,left color=\existingSpectrum!10] (-1,0) -- (-2,1) --(-2.2, 1) -- (-2.2,0)--cycle;
	    \shade[left color=\existingSpectrum!60,right color=\existingSpectrum!10] (1,0) -- (0,1) --(1.9, 1) -- (1.9,0)--cycle;
	    \fill[color=\absenceSpectrum!15] (0,1)..controls(-3/8,0.66)..(0,0) -- (-2,1) --cycle;   
	    \draw[eje](-2.25, 0)--(2,0) node [below] {\footnotesize $\sigma(L_2(\omega))$};
	    \draw[eje](0,-0.2) --(0,1.2) node[right]{\footnotesize $\omega$} ;
	    \draw[known](0,0) --(0,1) ;
	    \draw[known](0,0) --(-2,1) ;
	    \draw[known] (0,1)..controls(-3/8,0.66)..(0,0);
	    \draw[shift={(1,0)}] (0pt,0pt) -- (0pt,-2pt) node[below] {\tiny$+m$};
	    \draw[shift={(-1,0)}] (0pt,0pt) -- (0pt,-2pt) node[below] {\tiny$-m$};
	    \draw[shift={(-4pt,-4pt)}] (0pt,0pt) -- (0pt,-0pt) node[] {\tiny$0$};
	    \draw[shift={(0,1)}] (2pt,0) -- (-2pt,0) node[xshift=-4pt, yshift=4pt] {\tiny$m$};  
\end{tikzpicture}

	\captionsetup{width=.9\textwidth}
	\caption{Spectra of the families of self-adjoint operators~$\{L_0(\omega)\}_\omega$ and~$\{L_2(\omega)\}_\omega$. The spectrum of one operator is the intersection of the sketch with an horizontal line. Shading gray: essential spectrum; solid black: known eigenvalues~$\{-2\omega, 0\}$ and, for~$L_2$, the eigenvalue in~$(-2\omega,0)$ we prove the existence~of; salmon: region where we prove the absence of eigenvalues.}
		\label{fig:spectra}
\end{figure}
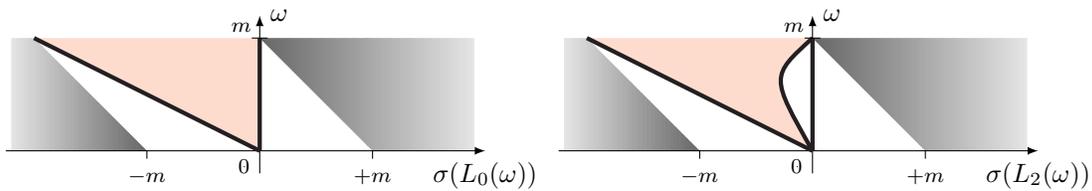
For power nonlinearities, we show in Section~\ref{Section_minmax_principle_and_Lmu} that $L_2$ has exactly one eigenvalue in $(-2\omega, 0)$ through a novel implementation of the minmax  principle for operators with a gap in the essential spectrum from~\cite{DolEstSer-00,GriSei-99, DolEstSer-06,EstLewSer-19,SchSolTok-19}. This minmax principle allows us to give the precise asymptotic behaviour of the eigenvalues of $L_\mu(\omega)$ non-relativistic limit $\omega$ going to $m$, for any power $p>0$. See Theorem~\ref{negative_spectrum_of_Lmu_nonrelativistic_limit} and the remarks below.

Point 2. crucially depends on the exact number of eigenvalues of $L_2$ in~$(-2\omega, 0)$.
\begin{theorem}\label{thm_VK_intro}
    Let $f$ satisfy Assumption~\ref{Assumption_general_nonlinearity}. Assume further that $\omega \in (0,m)$ and~$f$ are such that 
    \begin{enumerate}[label=\roman*),leftmargin=5\parindent]
        \item \label{it_L_2_hyp} $L_2$ has a single eigenvalue in~$(-2\omega, 0)$\,,
        \item\label{it_VK_cond} $\partial_\omega \norm{\phi_0(\omega)}^2_{L^2} \leq 0 $\,.
    \end{enumerate}
    Let $\mu \in (0,2)$ and~$H_\mu$ be the associated linearized operator defined in~\eqref{eq:def_H}.   
    Then, the algebraic multiplicity of zero, as an eigenvalue of~$H_\mu$, equals $2$.
\end{theorem}
We recall that the algebraic multiplicity of an eigenvalue is the dimension of the generalized null space. See Section~\ref{Section_bif_from_origin} for details. Note that the differentiability of~$\norm{\phi_0(\omega)}^2_{L^2}$ is shown in Proposition~\ref{prop_differentiability}.

Conditions~\emph{\ref{it_L_2_hyp}} and~\emph{\ref{it_VK_cond}} are precisely the generalizations to the Dirac case of the classical hypotheses in~\cite{GriShaStr-87, Weinstein-86}. The results there apply to groundstates of the corresponding  Schr{\"o}\-dinger operator and to many other models with positive-definite kinetic term.

Condition~\emph{\ref{it_VK_cond}} is known as the Vakhitov--Kolokolov criterion. In~\cite{BerComSuk-15}, the authors study the meaning of the equality case in the Dirac context. They show that eigenvalue branches can \emph{pass through zero} as $\omega$ varies when equality holds in~\emph{\ref{it_VK_cond}}. To the best of our knowledge, our result is the first generalization of the \emph{inequality} condition to nonlinear models of Dirac type.

In Section~\ref{Section_Lower_Bound_Re_z2}, we address our point 3. We give bounds of the form $\Re z^2 \geq E^2$ for eigenvalue branches of type~(b). See Theorems~\ref{Bound_on_Re_z2_and_Im_z_general_statement} and~\ref{Thm_ineq_Re_z2} for precise statements. For $\omega$ not too far from $m$, a bound of the form $\Re (z^2) \ge E^2 >0$ holds, as sketched (red dashed lines) in Figure~\ref{fig:spectrum_H_2}.
\begin{figure}[ht]
    \centering

	\begin{tikzpicture}[scale=2.5,
	    eje/.style = {-latex },
	    known/.style = {very thick,\existingSpectrum},
	new/.style = { thick, violet}]
	
	    \fill[\absenceSpectrum!15, shift={(0, 0)}] (-1.5,1)..controls(-0.9,0.6) and (-0.9, 0.2) ..(-0.9,0.02)--(-0.02,0.02)--(-0.02,1);
	    \fill[\absenceSpectrum!15] (-2.5,1) -- (-2.5,0.8) --(2.5,0.8)--(2.5,1)--cycle;
	    \draw[\absenceSpectrum!50, thick, dashed, shift={(0, 0)}] (-1.5,1)..controls(-0.9,0.6) and (-0.9, 0.2) ..(-0.9,0.02);
	\begin{scope}[xscale=-1, yscale= 1]
	    \fill[\absenceSpectrum!15, shift={(0, 0)}] (-1.5,1)..controls(-0.9,0.6) and (-0.9, 0.2) ..(-0.9,0.02)--(-0.02,0.02)--(-0.02,1);
	    \draw[\absenceSpectrum!50, thick, dashed, shift={(0, 0)}] (-1.5,1)..controls(-0.9,0.6) and (-0.9, 0.2) ..(-0.9,0.02);
	\end{scope}
	\begin{scope}[xscale=1, yscale= -1]
	    \fill[\absenceSpectrum!15, shift={(0, 0)}] (-1.5,1)..controls(-0.9,0.6) and (-0.9, 0.2) ..(-0.9,0.02)--(-0.02,0.02)--(-0.02,1);
	    \fill[\absenceSpectrum!15] (-2.5,1) -- (-2.5,0.8) --(2.5,0.8)--(2.5,1)--cycle;
	    \draw[\absenceSpectrum!50, thick, dashed, shift={(0, 0)}] (-1.5,1)..controls(-0.9,0.6) and (-0.9, 0.2) ..(-0.9,0.02);
	\end{scope}
	\begin{scope}[xscale=-1, yscale= -1]
	    \fill[\absenceSpectrum!15, shift={(0, 0)}] (-1.5,1)..controls(-0.9,0.6) and (-0.9, 0.2) ..(-0.9,0.02)--(-0.02,0.02)--(-0.02,1);
	    \draw[\absenceSpectrum!50, thick, dashed, shift={(0, 0)}] (-1.5,1)..controls(-0.9,0.6) and (-0.9, 0.2) ..(-0.9,0.02);
	\end{scope}
	
	    \shade[right color=\existingSpectrum!35,left color=\existingSpectrum!10, shift={(0, 0)}] (-1.6,-0.06) rectangle (-2.5, 0.06);
	    \shade[left color=\existingSpectrum!35,right color=\existingSpectrum!10, shift={(0, 0)}] (1.6,-0.06) rectangle (2.5, 0.06);
	    \shade[right color=\existingSpectrum!60,left color=\existingSpectrum!10, shift={(0, 0)}] (-0.4,-0.03) rectangle (-2.5, 0.03);
	    \shade[left color=\existingSpectrum!60,right color=\existingSpectrum!10, shift={(0, 0)}] (0.4,-0.03) rectangle (2.5, 0.03);
	    \draw[eje](-2.55, 0)--(2.55,0) node[] {};
	    \draw[eje](0,-1.05) --(0,1.05) node[above]{} ;
	    \fill[known] (-1.2,0) circle (1.4pt) ;
	    \fill[known] (1.2,0) circle (1.5pt) ;
	    \fill[known] (0.04,0) circle (1.5pt) ;
	    \fill[known] (-0.04,0) circle (1.5pt) ;
	    \draw[shift={(1.6,-1pt)}] node[below] {\tiny$m+\omega$};
	    \draw[shift={(-1.6,-1pt)}] node[below] {\tiny$-m-\omega\phantom{-}$};
	    \draw[shift={(0.4,-1pt)}] node[below] {\tiny$m-\omega$};
	    \draw[shift={(-0.4,-1pt)}] node[below] {\tiny$-m+\omega\phantom{-}$};
	    \draw[shift={(-1.2,-1pt)}] node[below] {\tiny$-2 \omega$};
	    \draw[shift={(-2pt,-1pt)}] node[below] {\tiny$0$};
	    \draw[shift={(1.2,-1pt)}] node[below] {\tiny$+2 \omega$};
	    \node at (2,0.5){ $\sigma(H_2) \subset \C$}; 
	\end{tikzpicture}
	
	\captionsetup{width=.9\textwidth}
	\caption{
		Sketch of the spectrum of $H_2$, at a fixed $\omega$, as a subset of $\C$. \\
		Shading gray: essential spectrum (real); black disks: known eigenvalues; salmon: regions where we prove absence of eigenvalues; dashed red lines: hyperbola $\Re z^2 = E^2$. \\
		The lower bound $\Re z^2 \geq E^2$ in Theorem~\ref{Bound_on_Re_z2_and_Im_z_general_statement} means that no eigenvalues occur off the axes in the curved region, while Proposition~\ref{Prop_bound_on_Im_z} bounds $|\!\Im(z)|$. \\
		Sketch drawn with $\omega/m \sim 0.6$ and $E/m \sim 0.9$.
		}
	\label{fig:spectrum_H_2}
\end{figure}
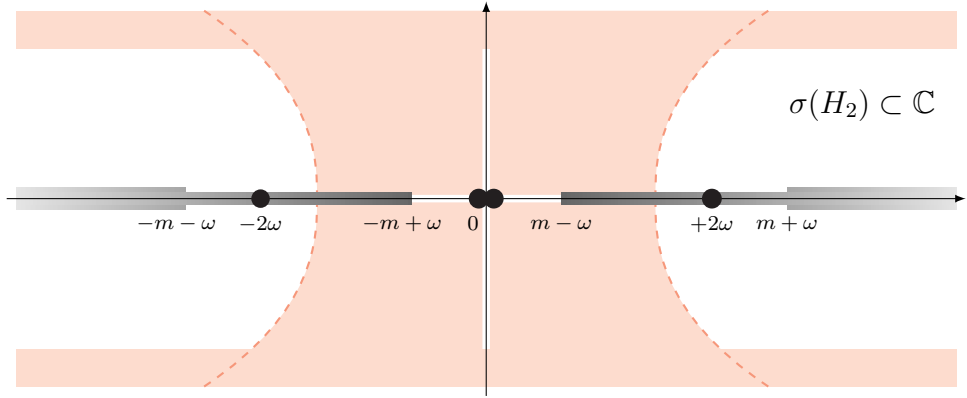

Combining Theorem~\ref{thm_VK_intro} with the bounds of Section~\ref{Section_Lower_Bound_Re_z2} allows to exclude eigenvalues of $H_2$ on the imaginary axis. From now on we denote by $\normt{Q}$ the operator norm of~$Q$.
\begin{corollary}\label{Corollary_no_imaginary_evs}
	Under the conditions of Theorem~\ref{thm_VK_intro}, $H_2$ has no non-zero eigenvalues on the imaginary axis if
	\begin{equation} \label{eq:simplified_corollary}
    		\normt{Q} \le \frac{3\sqrt{3}}{2} \omega\,.
	\end{equation}
\end{corollary}
We emphasize that condition~\eqref{eq:simplified_corollary} is a simplified version of the more elaborate bounds available in Section~\ref{Section_Lower_Bound_Re_z2}. Moreover, it is always satisfied in the non-relativistic limit $\omega$ going to $m$ since $\normt{Q}$ goes to $0$ in that limit. See Proposition~\ref{prop:basic_groundstate_properties}.

In section~\ref{Section_Power_nonlinearity_generalities}, we specify to power-like nonlinearities $f(s)= s |s|^{p-1}$, $p>0$, and to the Gross--Neveu model $p=1$ in Section~\ref{Section:Gross_Neveu}. We show that, for these nonlinearities, the Vakhitov--Kolokolov criterion holds on one hand for all $\omega \in (0,m)$ when $0< p \leq 2$, and on the other hand for $0<\omega<\omega_c(p)< m$ when $p > 2$.
We also compute $\normt{Q}$ as a function of $(p, \omega)$ and find a region in the $(p,\omega)$-plane where \eqref{eq:simplified_corollary} is satisfied. In the intersection of these regions, plotted in Figure~\ref{Power_nonlin_admissible_ranges_Re_z2_and_VK}, there are no eigenvalues on the imaginary axis.
\begin{figure}[ht]
    \captionsetup[subfigure]{labelformat=empty}
    \begin{subfigure}{.45\columnwidth}
        \centering
        \includegraphics[width=\columnwidth]{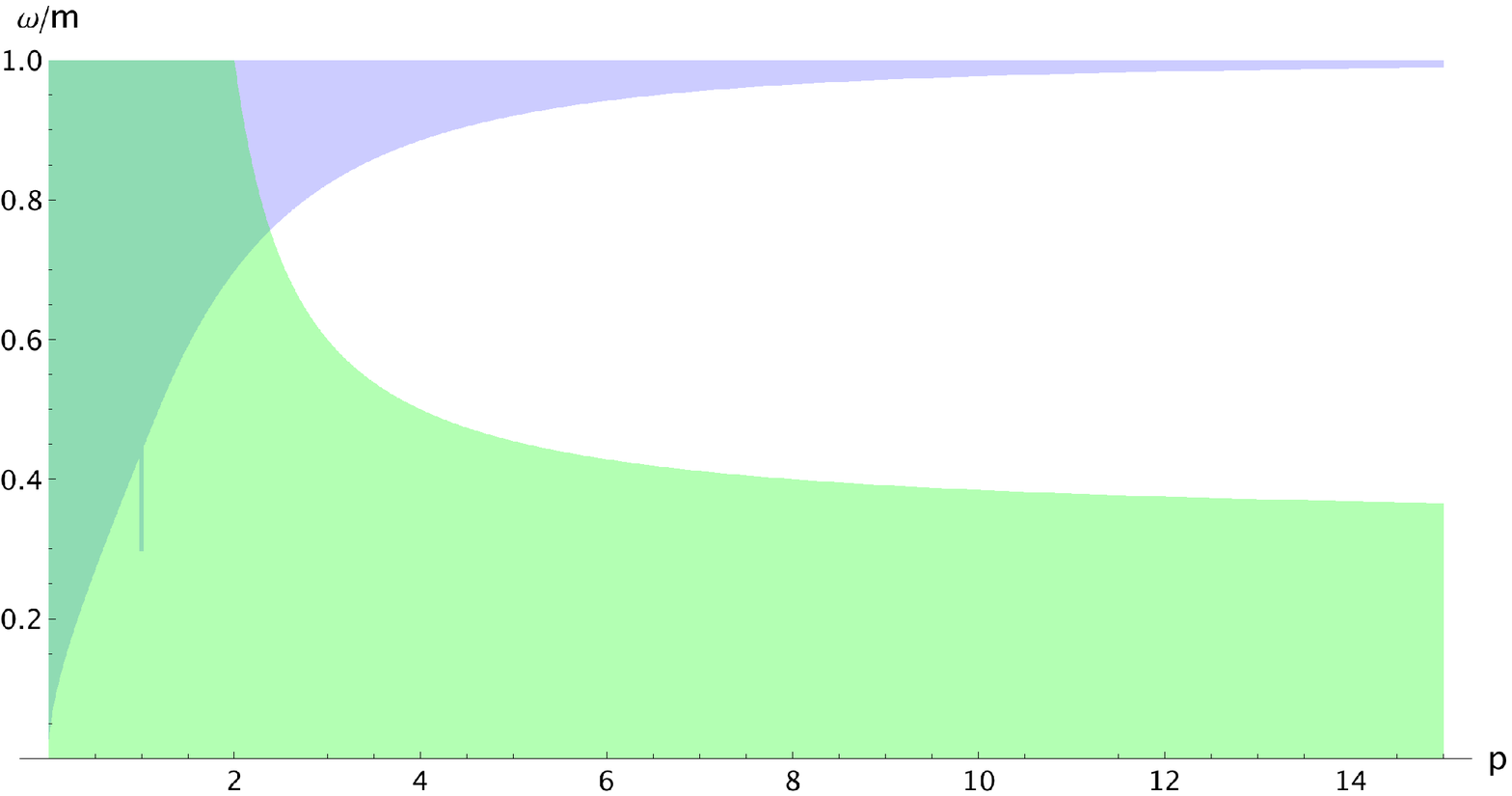}%
        \caption{$p \in (0,15)$}%
    \end{subfigure}%
    \begin{subfigure}{.45\columnwidth}
        \centering
        \includegraphics[width=\columnwidth]{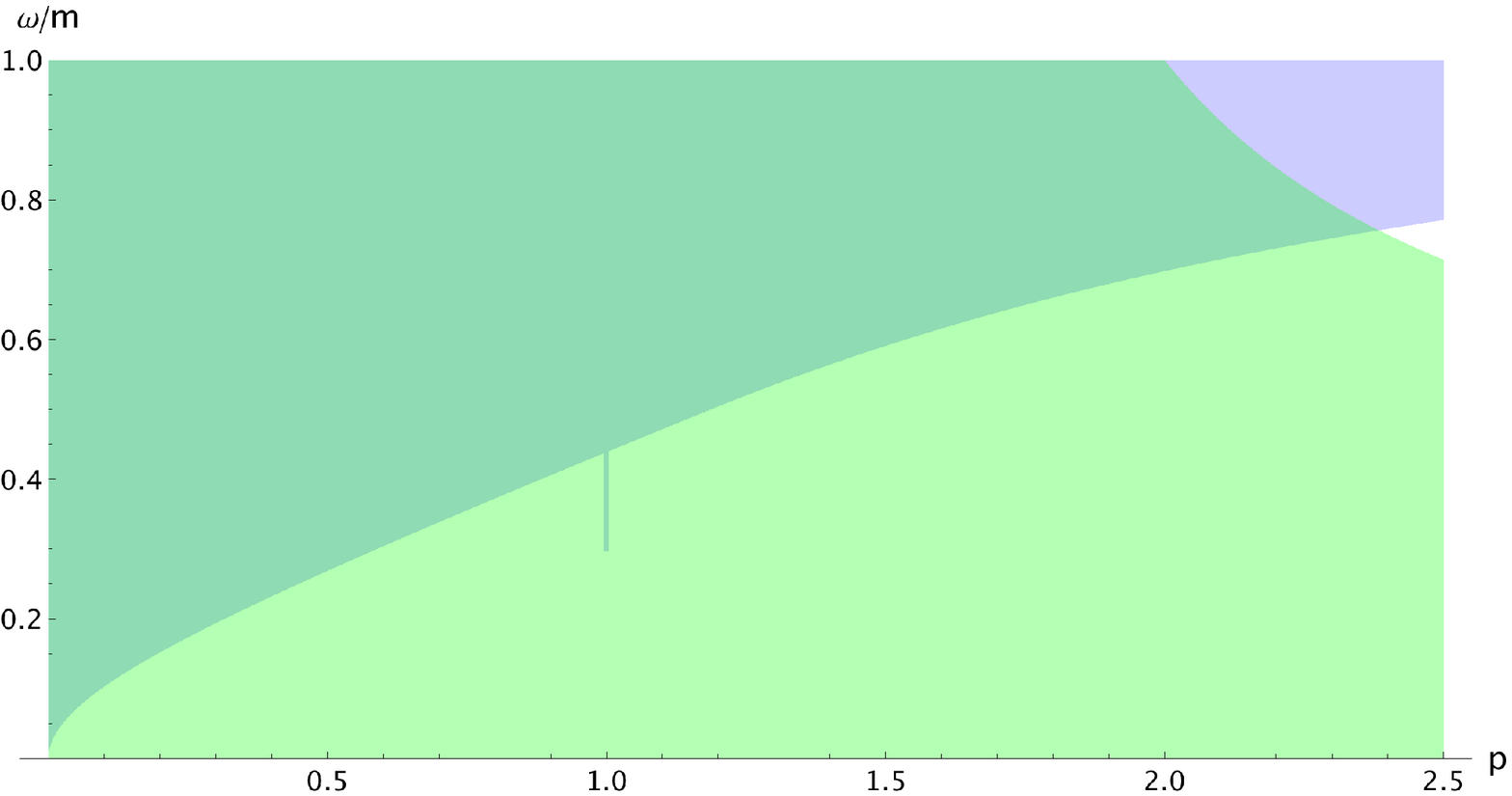}%
        \caption{Zoom onto~$p \in (0,2.5)$}%
    \end{subfigure}
    
	\captionsetup{width=.9\textwidth}
    \caption{
        Blue region (upper left corner): parameters $(p,\omega/m)$ for which we show that the $\Re z^2 \geq 0$ holds. Note that at $p=1$, we obtain an improved range in $\omega$. \\
        Green region (lower left corner): parameters for which we show that the Vakhitov--Kolokolov criterion holds. \\
        Intersection of these regions: $H_2$ has no non-zero purely imaginary eigenvalues.
	}
	\label{Power_nonlin_admissible_ranges_Re_z2_and_VK}
\end{figure}

\subsection{Outline of the paper} 
We start with basic properties of the groundstates and the operators in Section~\ref{Section_prelim}.
We then move on to the key ingredients of the paper, that hold for general nonlinearities. In Section~\ref{Section_bif_from_origin}, after recalling some basic facts about analytic perturbation theory for non-selfadjoint operators, we prove Theorem~\ref{thm_VK_intro} and show how it implies the absence of non-zero purely imaginary eigenvalues.
In Section~\ref{Section_Lower_Bound_Re_z2}, we obtain bounds depending on $\normt{Q}$ for eigenvalues lying off the axes of the complex plane and prove Corollary~\ref{Corollary_no_imaginary_evs}.

In the remainder of the paper, we specify to the power case $f(s) = s |s|^{p-1}$, $p>0$. In Section~\ref{Section_Power_nonlinearity_generalities}, we compute $\partial_\omega \norm{\phi_0(\omega)}_{L^2}^2$ to check the validity of the Vakhitov--Kolokolov criterion. We also compute $\normt{Q}$ and plug it into the bounds from the previous section.
We then specify to the case $p=1$ in Section~\ref{Section:Gross_Neveu}, where we show that $L_0$ has only two eigenvalues, and how this allows to exclude non-zero eigenvalues on the imaginary axis for all $\omega \in (0.297 m,m)$.

In the final section, we prove that $L_2$ has exactly one eigenvalue in $(-2 \omega, 0)$ for power nonlinearities $f(s) = s |s|^{p-1}$, $p>0$. 

\medskip

\noindent{\bf Acknowledgments.} We thank S\'ebastien Breteaux, J\'er\'emy Mougel, and Phan Th\`anh Nam for helpful discussions, and Matías Moreno for checking some of the computations. We thank Andrew Comech for his remarks on the preprint version. 
E.S. thanks Thomas S{\o}rensen and the Center for Advanced Studies at LMU for their hospitality during his stay in Munich where part of this work took place. We thank the referees for many suggestions that improved the manuscript.
The research visits leading to this work where partially funded by ANID (Chile) project REDI--170157.
J.R. received funding from the Deutsche Forschungsgemeinschaft (DFG, German Research Foundation) under Germany's Excellence Strategy (EXC-2111-390814868).
E.S. and H.VDB. have been partially funded by ANID (Chile) through Fondecyt  project
\#118--0355.
H.VDB. acknowledges support from ANID through Fondecyt projects \#318--0059 and \#11220194 and from CMM through ANID PIA AFB17000, \#ACE210010 and project France-Chile MathAmSud EEQUADDII 20-MATH-04.

\section{Preliminaries}\label{Section_prelim}

This section groups a number of results about~$\phi_0$ and the linearization operators. We start by establishing properties of~$\phi_0$. Then we identify the well-known eigenvalues and establish the symmetry properties of the spectra, both recalled in Proposition~\ref{WellKnownProperties}. After that, we move to the proof of Theorem~\ref{groundstate}, followed by the simplicity of the eigenvalues of $L_\mu$ that will be heavily used in the remainder of the paper, before completing the proof of Proposition~\ref{WellKnownProperties}. Next, we prove Proposition~\ref{Prop_bound_on_Im_z} that gives a simple bound on the imaginary part of eigenvalues of $H_\mu$. 
Finally, we use the spectral properties of $L_2$ to prove differentiability of $\omega \mapsto \norm{\phi_0(\omega)}^2$, necessary to state Theorem~\ref{thm_VK_intro}.

Throughout this section, we assume that $f$ satisfies Assumption~\ref{Assumption_general_nonlinearity} and we denote $\norm{\psi}_{\C^2}:=\sqrt{|\psi_1|^2+|\psi_2|^2}$ for any $\psi=(\psi_1,\psi_2)\transp$.

For~$\phi_0$, we have the following results.
\begin{proposition} \label{prop:basic_groundstate_properties}
    Let $f$ satisfy Assumption~\ref{Assumption_general_nonlinearity}.
    Then there exists a unique family $\{\phi_0(\omega)\}_{\omega \in (0,m)}$ of $H^1(\R,\R^2)$ solutions $\phi_0(\omega) = (v_\omega, u_\omega)\transp$ to~\eqref{eq:phi_0}, such that $v_\omega$ is even with $v_\omega(0) >0$ and $u_\omega$ is odd. These solutions $\phi_0(\omega,\cdot)$ are in $\mathcal{C}^2(\R)$ and satisfy
    \begin{enumerate}[label=(\roman*)]
        \item \label{it:positivity}$\pscal{\phi_0(\omega, \cdot), \sigma_3 \phi_0(\omega, \cdot) }_{\C^2}=v_\omega^2-u_\omega^2>0$ on~$\R$;
        \item \label{it:exponential_decay} $\norm{\phi_0(\omega, x)}_{\C^2}$ decays exponentially in $x$;
          \item \label{it:decay_of_Q} $\norm{Q(x)}_{\C^2\to \C^2}$ is bounded and decays exponentially in $x$;
         \item \label{it:decay_in_Nrel_limit} $\lim_{\omega \to m} \norm{\phi_0(\omega, \cdot)}_{L^\infty} = \lim_{\omega \to m}\normt{Q_\omega} = 0$.
    \end{enumerate}
\end{proposition}
Most of these results are well-known from the literature, though not necessarily formulated with our hypotheses in Assumption~\ref{Assumption_general_nonlinearity}. Also, for the reader only interested in power nonlinearities, the properties can be easily checked from the explicit expressions given in Section~\ref{Section_Power_nonlinearity_generalities}. We give references and proofs in Appendix~\ref{Appendix_ODE}.
 
The next statement groups known results about the eigenvalues of $L_\mu$ and $H_\mu$.

\begin{proposition} \label{Prop:eigenvalues_symmetry}
    Let $\mu\in\R$ and $\omega\in(0,m)$.
    \begin{enumerate}[label=(\roman*)]
        \item\label{it:sym_L_0} The spectrum of~$L_0$ is symmetric with respect to $-\omega$.
        \item\label{it:eigenvalues_L_0} $-2\omega$ and~$0$ are eigenvalues of~$L_0$ with eigenfunctions $\phi_{-2\omega} := \sigma_1 \phi_0$ and~$\phi_0$.
        \item\label{it:2omega} $Q\phi_{-2\omega}= 0$ and~$-2\omega$ is an eigenvalue of~$L_\mu$ with eigenfunction $\phi_{-2\omega}$.
        \item\label{it:eigenvalues_H_2} $-2\omega$ and~$0$ are eigenvalues of~$L_2$ with eigenfunctions $\phi_{-2\omega}$ and~$\partial_x \phi_0$. As a consequence, $0$ and~$\pm2\omega$ are eigenvalues of~$H_2$, with respective eigenfunctions
            \[
                \begin{pmatrix}
                    0  \\
                    \phi_0
                \end{pmatrix},\,
                \begin{pmatrix}
                    \partial_x \phi_0 \\
                    0 
                \end{pmatrix}
                \textrm{ and }
                \begin{pmatrix}
                    \phi_{-2\omega} \\
                    \mp \phi_{-2\omega}
                \end{pmatrix}.
            \]
        \item The spectrum of~$H_\mu$ is symmetric with respect to the real and imaginary axis.
    \end{enumerate}
\end{proposition}
\begin{proof}
    The equation $L_0 \phi_0 = 0$ is just a rewriting of the nonlinear ODE~\eqref{eq:phi_0}. Next, we use the anticommutators
    \[
        \{L_0 + \omega\Id, \sigma_1  \} = \{D_m, \sigma_1  \} - f(\pscal{\phi_0,\sigma_3\phi_0}_{\C^2}) \{\sigma_3, \sigma_1  \} =0
   \]
   to conclude that $\sigma_1 \phi_0$ is an eigenfunction associated to $-2\omega$, which concludes \emph{\ref{it:eigenvalues_L_0}}, and that the spectrum of~$L_0 + \omega\Id$ is symmetric, which is \emph{\ref{it:sym_L_0}}.
   
    The definition of~$Q$ in~\eqref{Def_of_Q} gives~\emph{\ref{it:2omega}} because, by properties of Pauli matrices and since the components of~$\phi_0$ are real, we have
    \[
        \left(\sigma_3\phi_0\right)\transp \phi_{-2\omega} = \pscal{\sigma_1 \phi_0, \sigma_3\phi_0}_{\C^2} = 0\,.
    \]
   
   For~\emph{\ref{it:eigenvalues_H_2}}, we compute
   \[
       H_2
        \begin{pmatrix}
            0 \\
            \phi_0
        \end{pmatrix}
        =
        \begin{pmatrix}
            L_0\phi_0 \\
            0
        \end{pmatrix}
        =0 
        \quad \textrm{ and } \quad 
       H_2
        \begin{pmatrix}
            \phi_{-2\omega} \\
            \mp \phi_{-2\omega}
        \end{pmatrix}
        =
        \begin{pmatrix}
            \mp L_0\phi_{-2\omega} \\
            L_0 \phi_{-2\omega}
        \end{pmatrix}
        =\pm 2\omega
        \begin{pmatrix}
            \phi_{-2\omega} \\
            \mp \phi_{-2\omega}
        \end{pmatrix},
   \]
    where we used \emph{\ref{it:2omega}} in the latter.
    The second eigenfunction associated to $0$ is a consequence of~$\phi_0$ being real-valued and the definition of~$Q$ in~\eqref{Def_of_Q}. Indeed, note that
    \[
        \Re(\pscal{\partial_x\phi_0,\sigma_3\phi_0}_{\C^2}) = \pscal{\partial_x\phi_0,\sigma_3\phi_0}_{\C^2} = (\sigma_3\phi_0)\transp \partial_x\phi_0\,.
    \]
    So we obtain
    \[
        L_0\partial_x\phi_0 = \partial_x\left(L_0\phi_0\right) + 2 f'(\pscal{\phi_0,\sigma_3\phi_0}_{\C^2}) \Re(\pscal{\partial_x\phi_0,\sigma_3\phi_0}_{\C^2})\sigma_3\phi_0 = 2Q\partial_x\phi_0\,,
    \]
   and conclude that $\partial_x\phi_0$ is an eigenfunction of~$L_2$ associated to the eigenvalue $0$, i.e.,
   \[
        L_2\partial_x\phi_0 = (L_0 - 2Q)\partial_x\phi_0 = 0\,.
   \]
   By definition of $H_2$, we have
   \[
       H_2
        \begin{pmatrix}
            \partial_x \phi_0 \\
            0 
        \end{pmatrix}
        =
        \begin{pmatrix}
            0 \\
            L_2\partial_x \phi_0
        \end{pmatrix}
        =0\,.
   \]

    The last statement of this proposition is an immediate consequence of the fact that $H_\mu$~is equal to its complex conjugate operator and that
    \[
        \begin{pmatrix}
            \Id_{2\times2}&0\\
            0&-\Id_{2\times2}
        \end{pmatrix}
        H_\mu = - H_\mu
        \begin{pmatrix}
            \Id_{2\times2}&0\\
            0&-\Id_{2\times2}
        \end{pmatrix}. \qedhere
    \]
\end{proof}

Notice that the operator $L_0 + \omega\Id$ may be interpreted as a Dirac operator with the effective mass
\begin{equation}\label{Def_M}
    M := m - f\!\left(v^2 - u^2\right)\,.
\end{equation}
This Dirac operator is connected to the Schr{\"o}\-dinger operators $-\partial_x^2 + M^2 \mp\prim{M}$. We use this connection to identify $\phi_0$ and $\phi_{-2\omega}$ as \emph{groundstates}. That is, they correspond to the smallest eigenvalue (in absolute value) of the operator $L_0 + \omega\Id$.
\begin{proof}[Proof of Theorem~\ref{groundstate}]
	We already know (Proposition~\ref{Prop:eigenvalues_symmetry})  that $\{-2\omega, 0\}\subset\sigma(L_0)$, so we have to prove that $L_0$ has no eigenvalues in~$(-2\omega,0)$. Defining $A := L_0 + \omega\Id$, we need to show that $A$ has no spectrum in~$(-\omega, +\omega)$.
	
	For this, it is convenient to use the \emph{change of basis}
	\[
	    U = \frac{1}{\sqrt{2}}\begin{pmatrix} 1 & 1 \\ 1& -1 \end{pmatrix} = \frac{1}{\sqrt{2}}(\sigma_1 + \sigma_3)
	\]
	in order to obtain the unitary equivalence, using the properties of the Pauli matrices,
	\[
	    U A U = -i\partial_x\sigma_2+M\sigma_1=\begin{pmatrix} 0 & - \partial_x + M \\ \partial_x + M  & 0\end{pmatrix},
	\]
	where $M$ is defined in~\eqref{Def_M}. From this, we find that $A^2$, with domain $H^2(\R)$, is unitarily equivalent to the following block diagonal operator
	\begin{align*}
	    (U  A U)^2 &=
	        \begin{pmatrix}
                (- \partial_x + M )( \partial_x + M ) &0 \\
                0 &(\partial_x + M )(- \partial_x + M )
            \end{pmatrix}\\
        &=
            \begin{pmatrix}
                - \partial_x^2 + M^2 - \prim{M}& 0 \\
                0& -\partial_x^2 + M^2 +\prim{M}
            \end{pmatrix},
	\end{align*}
	with, on the diagonal, two Schr{\"o}\-dinger operators with essential spectrum $[m^2, + \infty)$. We compute $(UAU)^2U\phi_0 = UA^2\phi_0 =\omega^2 U \phi_0$, since $A\phi_0=\omega\phi_0$ and
	\[
	    \sqrt2 U\phi_0 = \begin{pmatrix} v+u \\ v-u\end{pmatrix}.
	\]
	Given that $(v+u)(v-u)=v^2-u^2>0$ by Proposition~\ref{prop:basic_groundstate_properties}, this means that the functions $v \pm u\in\mathcal{C}^0(\R)$ do not change sign and are eigenfunctions of~$- \partial_x^2 + M^2 \mp \prim{M}$, respectively, associated to the same eigenvalue $\omega^2$. By Sturm's oscillation theorem (see e.g.~\cite{BerShu-91}), they are therefore the respective groundstates and~$A^2$ has no eigenvalue below~$\omega^2$. We conclude that $A= L_0 + \omega\Id$ has no eigenvalues in~$(-\omega,+\omega)$.
\end{proof}

We now prove the simplicity of eigenvalues of~$L_\mu$, which will be needed several times in the paper.
\begin{lemma}\label{L_mu_eigenvalues_simple_and_analytic_in_mu}
    Let $f$ satisfy Assumption~\ref{Assumption_general_nonlinearity}, $\omega\in(0,m)$, and $\mu\in\R$. Then, the eigenvalues of~$L_\mu$ are simple.
\end{lemma}
\begin{proof}
    Let $\lambda\in\R$ be an eigenvalue of~$L_\mu$ and assume $\phi_1 = (f_1,g_1)\transp$ and~$\phi_2 = (f_2,g_2)\transp$ to be eigenfunctions associated to $\lambda$.
    The equation 
    $L_\mu \phi_j = \lambda \phi_j$ can be written as
    \begin{equation*}
         \partial_x \phi_j = -i\lambda \sigma_2 \phi_j - M(x) \sigma_1 \phi_j - \mu \sigma_2 Q(x) \phi_j,
    \end{equation*}
    with $M$ defined in~\eqref{Def_M}. Furthermore, we can decompose
    \[
        \mu Q(x) =  q_0(x) \Id_{\C^2} + q_1(x)\sigma_1 + q_2(x) \sigma_3\,,
    \]
    for some functions $q_0$, $q_1$ and~$q_3$ whose explicit expressions are not necessary to complete the proof.
    Using the identity 
    \[
        \sigma_m \sigma_k = i \sum_{l=1}^3\epsilon_{mkl} \sigma_l 
    \]
   where $\epsilon_{mkl}$ is the completely antisymmetrix tensor such that $\epsilon_{123} = 1$,
    we finally rewrite the eigenvalue equation as
    \[
        \partial_x \phi_j = (-i\lambda + \mu q_0) \sigma_2 \phi_j - (M(x) + q_2(x) ) \sigma_1 \phi_j - \mu  q_1(x)\sigma_3 \phi_j.
    \]
    
    Now, define the determinant $W(x) := \det \begin{pmatrix}\phi_1(x) | \phi_2(x)\end{pmatrix}$ and compute
    \begin{align*}
        W' &= 
        \det \begin{pmatrix}\phi_1' |\phi_2\end{pmatrix} + 
        \det \begin{pmatrix}\phi_1 |\phi_2'\end{pmatrix} \\
        &=
        \begin{multlined}[t]
            \left( -i\lambda + \mu q_0\right) \left( \det \begin{pmatrix}\sigma_2\phi_1 |\phi_2\end{pmatrix} +  \det \begin{pmatrix}\phi_1 |\sigma_2 \phi_2\end{pmatrix}  \right)\\
            + (M(x) + q_2(x) ) \left( \det \begin{pmatrix}\sigma_1\phi_1 |\phi_2\end{pmatrix} +  \det \begin{pmatrix}\phi_1 |\sigma_1 \phi_2\end{pmatrix}  \right) \\
            -\mu q_1(x)\left( \det \begin{pmatrix}\sigma_3\phi_1 |\phi_2\end{pmatrix} +  \det\begin{pmatrix}\phi_1 |\sigma_3 \phi_2\end{pmatrix}  \right).
        \end{multlined}
    \end{align*}
    We conclude that $W'\equiv0$ since, for $k\in\{1,2,3\}$, we have
    \[
        \det \begin{pmatrix}\sigma_k \phi_1 |\phi_2\end{pmatrix} 
        = \det \begin{pmatrix}\sigma_k\phi_1 |\sigma_k^2\phi_2\end{pmatrix}
        = \det (\sigma_k) \det \begin{pmatrix}\phi_1 |\sigma_k \phi_2\end{pmatrix}
        = -\det \begin{pmatrix}\phi_1 |\sigma_k \phi_2\end{pmatrix}.
    \]
    Because $\phi_i\in L^2(\R)$,
    this implies that $W\equiv0$.
    Then, for each $x \in \R$, $\phi_1(x)$ and $\phi_2(x)$ are colinear vectors in $\C^2$. If $\phi_2(0) \neq 0$, there exists $\alpha \in \C$ such that
    $\phi_1(0) = \alpha \phi_2(0)$. 
    By linearity and uniqueness of the solution to the Cauchy problem
    \[
        \left\{
            \begin{aligned}
                L_\mu\Psi &= \lambda\Psi\\
                \Psi(0) &= \phi_1(0)= \alpha(0)\phi_2(0),
            \end{aligned}
        \right.
    \]
    this implies that  $\phi_1(x) = \alpha \phi_2(x)$ for all $x \in \R$.
    If $\phi_2(0) = 0$, the same uniqueness result implies that $\phi_2 \equiv 0$, contradicting that $\phi_2$ is an eigenfunction.
\end{proof}

We can now prove the basic spectral properties from the introduction.
\begin{proof} [Proof of Proposition~\ref{WellKnownProperties}]
The presence of eigenvalues stated in~\emph{\ref{WellKnownProperties_L0_EVs}}, \emph{\ref{WellKnownProperties_L2_EVs}}, \emph{\ref{WellKnownProperties_Hmu_symm_spect_axis}}, and~\emph{\ref{WellKnownProperties_H2_EVs}} are proved in Proposition~\ref{Prop:eigenvalues_symmetry} and their simplicity is Lemma~\ref{L_mu_eigenvalues_simple_and_analytic_in_mu}. 
 We are left with points~\emph{\ref{WellKnownProperties_Lmu_ess_spec}} and~\emph{\ref{WellKnownProperties_Hmu_ess_spec}}, the identification of the essential spectra.  
 The free Dirac operator~$D_m - \omega \Id$ has essential spectrum $\R \setminus (-m-\omega, m-\omega)$, see e.g.~\cite{Thaller-92}.
 Next, $f(v^2 - u^2)\sigma_3$ and $Q$ are bounded and symmetric operators, hence $L_\mu$ is self-adjoint with domain $H^1(\R)$. Finally, $f(v^2 - u^2)$ and $Q$ decay at infinity, so $L_\mu$ is a relatively compact perturbation (see e.g. \cite[Section~XII.4]{ReeSim4}) of~$D_m - \omega \Id$ and
    \[
        \sigma_{\rm ess}(L_\mu) = (-\infty, -m-\omega] \cup [m-\omega, +\infty)\,.
    \]
    Moreover the discrete spectrum of~$L_\mu$ consists of eigenvalues of finite multiplicity in the \emph{gap} $(-m-\omega, m-\omega)$.
    
  Similarly, for \emph{(v)}, $H_\mu$ is a bounded relatively compact perturbation of
    \[
        \mathrm{D}:=\begin{pmatrix}
            0 & D_m - \omega\Id \\ D_m - \omega\Id &0
        \end{pmatrix}
    \]
    and its essential spectrum is 
    \[
        \sigma_{\rm ess}(H_\mu) = (-\infty, -m + \omega] \cup [m-\omega, +\infty)\,. \qedhere
    \]
\end{proof}

The next result is our first bit of new information on the spectrum of~$H_\mu$.
\begin{proposition} \label{Prop_bound_on_Im_z}
    Let $f$ satisfy Assumption~\ref{Assumption_general_nonlinearity}, $\omega\in(0,m)$, and~$\mu\geq0$. If~$z \in \C$ is an eigenvalue of~$H_\mu$, then
	\[
	    \left|\Im z \right| \leq \frac{\mu}{2} \normt{Q}.
	\]
\end{proposition} 
\begin{proof}
    The eigenvalue equation for $z$, with associated eigenvector $\varphi = (\varphi_1, \varphi_2)\transp$, reads
    \begin{subequations}
        \begin{empheq}[left=\empheqlbrace]{align}
            L_0 \varphi_2 &= z \varphi_1\,, \label{eq:eigen_of_H_split_a} \\
            L_\mu \varphi_1 &= z \varphi_2\,. \label{eq:eigen_of_H_split_b}
        \end{empheq}
    \end{subequations}
    Summing the inner products of~\eqref{eq:eigen_of_H_split_a} and~\eqref{eq:eigen_of_H_split_b}, respectively with $\varphi_1$ and~$\varphi_2$, gives
    \[
        \pscal{\varphi_1,L_0 \varphi_2}_{L^2} + \pscal{\varphi_2,L_0 \varphi_1}_{L^2} - \pscal{\varphi_2,\mu Q \varphi_1}_{L^2} = z\norm{\varphi}_{L^2}^2.
    \]
    Since $L_0$ is self-adjoint, $\pscal{\varphi_1,L_0 \varphi_2}_{L^2} + \pscal{\varphi_2,L_0 \varphi_1}_{L^2} \in \R$ and we have
    \[
        \Im z \norm{\varphi}_{L^2}^2=\mu\Im\pscal{Q \varphi_1,\varphi_2}_{L^2}.
    \]
    By Cauchy--Schwarz inequality, it yields
    \[
        \left|\Im z \right| \leq \mu\normt{Q} \frac{\norm{\varphi_1}_{L^2}\norm{\varphi_2}_{L^2}}{\norm{\varphi_1}_{L^2}^2 + \norm{\varphi_2}_{L^2}^2} \leq \frac{\mu}{2}\normt{Q}. \qedhere
    \]
\end{proof}

In order to state the last result of this section, we need to separate $L^2(\R,\C^2)$ in \emph{even} and \emph{odd} subspaces.
Following the convention in~\cite{BerCom-12}, we say that a function $\psi:\R \mapsto \C^2$ is \emph{even}, if its first component is an even function and its second component is odd. These functions are eigenfunctions with eigenvalue $+1$ of the parity operator $\sigma_3 \mathcal{P}$, where $\mathcal{P} u(x) := u(-x)$. This operator commutes with $D_m$, and with $L_\mu$ for all $\mu$.
Since the eigenvalues of $L_2$ are simple and the eigenfunction $\partial_x \phi_0$ associated to $0$ is \emph{odd}, $L_2$ is an invertible operator from $H^{1, \rm even}(\R)$ to $L^{2, \rm even}(\R)$.
With this in place, we can state the following proposition, saying that the Vakhitov--Kolokolov criterion \emph{makes sense}.
\begin{proposition} \label{prop_differentiability}
    Let $f$ satisfy Assumption~\ref{Assumption_general_nonlinearity}, $\{\phi_0( \omega)\}_{\omega\in(0,m)}$ be the associated solitary wave solutions and $\{L_2( \omega)\}_{\omega\in(0,m)}$ be as in~\eqref{Def_L_mu}. Then, $L_2( \omega)$ restricted to the even subspace is invertible for all $\omega\in(0,m)$ and $\omega \mapsto \phi_0(\omega) \in\mathcal{C}^1\left((0,m), H^1(\R)\right)$ with
     \[
        \partial_\omega\phi_0(\omega) = \left(L_2( \omega)\right)^{-1} \phi_0(\omega) \quad \textrm{ and } \quad \partial_\omega\norm{\phi_0(\omega)}^2_{L^2} = 2 \pscal{\phi_0(\omega), \left(L_2( \omega)\right)^{-1} \phi_0(\omega)}.
     \]
\end{proposition}
The differentiability can be checked directly for power nonlinearities, and obtained from the implicit function theorem if we assume $f\in\mathcal{C}^1(\R)$. We defer the proof, under only Assumption~\ref{Assumption_general_nonlinearity}, to Appendix~\ref{Appendix_ODE}.

\section{Bifurcations from the origin}\label{Section_bif_from_origin}
The goal of this section is to prove Theorem~\ref{thm_VK_intro}. We first recall some basic facts about perturbation theory for closed operators. Proofs can be found in standard references such as~\cite{HislopSigal,kato}.  
The algebraic multiplicity $m_a(\lambda)$ of an isolated point $\lambda$ of the spectrum of a closed operator $T$ is defined as the dimension of the range of the Riesz projector
\[
    P_\gamma (T) = \frac{1}{2\pi i} \int_\gamma (T - z)^{-1} \di z,
\]
where $\gamma$ is any closed, simple contour such that $\gamma \subset \rho(T)$ and~$\lambda$ is the only point of the spectrum of~$T$ inside $\gamma$.
This multiplicity coincides, see~\cite[Chapter 6]{HislopSigal}, with the dimension of the generalized eigenspace
\[
    \cup_{n\in N} \ker(T-\lambda)^n.
\]
The geometric multiplicity $m_g(\lambda)$ is the dimension of~$\ker(T-\lambda)$. From this, it follows that $m_a(\lambda) \geq m_g(\lambda)$ and, if $m_a(\lambda)> 0$, then $m_g(\lambda) > 0$. 
If $T$ is self-adjoint, the geometric and algebraic multiplicities of any eigenvalue coincide. The algebraic multiplicity is the \emph{correct} multiplicity for analytic perturbation theory for closed operators.

Recall that a family of operators $\{T_\mu\}_\mu$ is said to be a type-A analytic family in an open set $\mathcal{S}\subset \C$ if for all $\mu\in \mathcal{S}$ the operator $T_\mu$ is closed on a domain~$\mathcal{D}$ independent of~$\mu$, and if for each $\varphi\in \mathcal{D}$ the map $\mu \mapsto T_\mu	 \varphi$ is strongly analytic. For shortness, we call such $\{T_\mu\}_\mu$ an \emph{analytic family}.

If $\{T_\mu\}_\mu$ is an analytic family of operators and~$\gamma$ is a contour such that $\gamma \subset \rho(T_\mu)$ for all $\mu \in B(0,r)$, 
then the sum of algebraic multiplicities of eigenvalues in the interior of~$\gamma$ is constant for $\mu \in B(0,r)$, see~\cite[Chapter 15]{HislopSigal}.

In our particular case, the families~$\{L_\mu\}_{\mu\in\C}$ and~$\{H_\mu\}_{\mu\in\C}$ are linear in~$\mu$ and therefore analytic. They have isolated eigenvalues of finite algebraic multiplicity away from their essential spectra.
By analytic perturbation theory, the eigenvalue branches $\lambda(\mu)$ of~$L_\mu$ are analytic functions of~$\mu$, and those of~$H_\mu$ are continuous.\footnote{Analytic functions of~$\mu$ as long as eigenvalues are simple, and branches of a multivalued function analytic in~$(\mu -\mu_0)^{1/k}$ if at least $k \in \N$ branches intersect for $\mu_0$.}

We start this section with a lemma that characterizes Condition 
\emph{\ref{it_L_2_hyp}} from Theorem~\ref{thm_VK_intro} in terms of the invertibility of~$L_\mu$.
\begin{lemma}\label{Lemma_equiv_number_evs_L2_and_Lmu}
    Let $f$ satisfy Assumption~\ref{Assumption_general_nonlinearity} and $\omega\in(0,m)$. The following statements are equivalent 
    \begin{enumerate}[label=\roman*),leftmargin=5\parindent]
	    \item \label{it_L_2_hyp-1} $L_2$ has a single eigenvalue in~$(-2\omega, 0)$,
	    \item $0\notin \sigma(L_\mu)$ for all $\mu\in (0,2)$.
    \end{enumerate}
\end{lemma}
\begin{proof}
\begin{figure}[ht]
    \centering
	\begin{tikzpicture}[x=4cm, y=1cm,
	    eje/.style = {-latex },
	    known/.style = {ultra thick,\existingSpectrum},
	new/.style = { ultra thick, \possibleSpectrum!60, dashed},
	newCircle/.style = { ultra thick, \possibleSpectrum!60}]
	
	    \shade[left color=\existingSpectrum!60,right color=\existingSpectrum!10, shading angle =180] (0, 1) -- (2,1) --(2, 2.5) --(0,2.5) --cycle;
	    \shade[left color=\existingSpectrum!60,right color=\existingSpectrum!30, shading angle =0] (0,-3) -- (2,-3) --(2,-3.5) -- (0,-3.5)--cycle;
	    
	    \draw[eje](0,-3.5)--(0,2.5) node [below left] {\footnotesize $\sigma(L_\mu)$};
	    \draw[eje](0,0) --(2.2,0) node[below]{\footnotesize $\mu$} ;
	    
	    \draw[known](0,0)  .. controls (0.8,-0.9) ..(2,-1.2) ;
	    \draw[known](0,-2) --(2,-2) ;
	   \draw[known](1.3,1)  .. controls (1.5, 0.4) ..(2,0) ;
	   
	 \draw[new](0.6,1)  .. controls (1.5, -0.2) ..(2,-0.4) ;
	 \draw[new](0,0.3)  .. controls (1.2, -0.4) ..(2,-0.9) ;
	 \draw[new](1.5 , 1)  .. controls (1.7, 0.6) ..(2,0.5) ;
	\draw[new](0 ,-2.3)  .. controls (0.5, -2.8) ..(1, -3) ;
	    \draw[shift={(-4pt,0)}] (0,0) -- (0pt,-0pt) node[] {\tiny$0$};
	    \draw[shift={(0,1)}] (0,0) -- (0,0) node[left] {\tiny$m-\omega$};  
	    \draw[shift={(0,-3)}] (0,0) -- (0,0) node[left] {\tiny$-m-\omega$};
	    \draw[shift={(0,-2)}] (0,0) -- (0,0) node[left] {\tiny$-2\omega$};
	     \draw[shift={(2,0)}] (0,2pt) -- (0,-2pt) node[xshift=+6pt, yshift=-4pt] {\tiny$2$};
	     \node at (0.4,-0.8) {\footnotesize$\lambda_1(\mu)$};
	    \fill[newCircle] (2,-0.4) circle (4pt) ;
	    \fill[newCircle] (2,-0.9) circle (4pt) ;
	    \draw[newCircle] (0.55,0) circle (4pt) ;
	    \draw[newCircle] (1.4,0) circle (4pt) ;
	\end{tikzpicture} 
	
	\captionsetup{width=.9\textwidth}
	\caption{Sketch for the proof of Lemma~\ref{Lemma_equiv_number_evs_L2_and_Lmu} representing vertically the spectrum of $L_\mu$ as $\mu$ varies from $0$ to $2$. Solid curves: eigenvalue branches always present; dashed curves: scenarios for other possible eigenvalue branches.}
		\label{fig:sketch_lemma}
\end{figure}
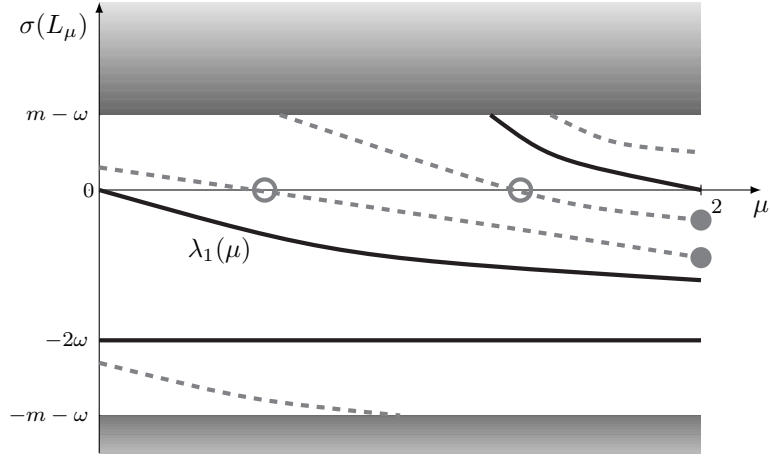

    Recall that eigenvalues of~$L_\mu$ are analytic in~$\mu$.
    Let $\lambda(\mu)$ be an eigenvalue of~$L_\mu$ and~$\varphi(\mu)$ be a corresponding normalized eigenfunction, then first order perturbation theory (also known as the Feynman--Hellmann theorem), together with the nonnegativity of~$Q$, yields
    \begin{equation}\label{Feynman_Hellmann_type_ineq}
        \partial_\mu\lambda(\mu) = -\pscal{\varphi(\mu), Q \varphi(\mu)}\leq 0\,.
    \end{equation}
    Denote now by $\lambda_1(\mu)$ the branch with $\lambda_1(0)=0$ and observe that
    \[
        \partial_\mu\lambda_1(0) = -  \pscal{v^2-u^2,f'(v^2-u^2)\left(v^2-u^2\right)} < 0\,,
    \]
    by Assumption~\ref{Assumption_general_nonlinearity}, the definition of~$Q$, and since $v^2-u^2>0$ by Proposition~\ref{prop:basic_groundstate_properties}.
    Thus, $\lambda_1(\mu)<0$ for any $\mu>0$. Moreover, $\lambda_1(\mu)>-2\omega$ since $-2\omega$ is an eigenvalues of~$L_\mu$ for any $\mu>0$ by Proposition~\ref{Prop:eigenvalues_symmetry}, and eigenvalues are simple by Lemma~\ref{L_mu_eigenvalues_simple_and_analytic_in_mu}. Hence,
    \[
        \partial_\mu\lambda_1(\mu) \leq 0\quad \mbox{and}\quad -2\omega<\lambda_1(\mu)<0, \qquad \textrm{ for } \mu>0\,.
    \]
    Now, refer to Figure~\ref{fig:sketch_lemma}. Eigenvalues of $L_2$ in $(\lambda_1(2), 0)$, represented by the disks, originate from eigenvalues of $L_0$ in $(0, m-\omega)$ or from the positive essential spectrum. Since the eigenvalues are analytic and monotonic, each curve intersects the axis at a single point, represented by the circles. Therefore, the number of eigenvalues of~$L_2$ in~$(-2\omega, 0)$ equals the number of~$\mu$'s in~$[0,2)$ for which $0$ is an eigenvalue of~$L_\mu$.
\end{proof}

Now we can prove our first main result.
\begin{proof}[Proof of Theorem~\ref{thm_VK_intro}]
    Let $m_a(\lambda, H_\mu)$ be the algebraic multiplicity of~$\lambda$ as an eigenvalue of~$H_\mu$. For $\mu =0$, the operator $H_0$ is self-adjoint and we have $m_a(0, H_0)=2$ since
    \[
        \ker (H_0^n) = \ker (H_0) = \vect \left\{\begin{pmatrix} \phi_0 \\0
    \end{pmatrix}, \begin{pmatrix} 0 \\ \phi_0 
    \end{pmatrix}\right\}.
    \]
    In view of Lemma~\ref{Lemma_equiv_number_evs_L2_and_Lmu}, the function
    \[
        l(\mu):= \pscal{\phi_0, L_\mu^{-1} \phi_0},
    \]
    where $L_\mu$ is identified with its restriction to the even subspace, is well-defined on~$(0,2]$ and real analytic. Indeed, referring once more to Figure~\ref{fig:sketch_lemma}, $l(\mu)$ is well-defined and continuous until the value of $\mu > 2$ where the first \emph{even} eigenvalue crosses the axis.
By Proposition~\ref{prop_differentiability}, we identify
    \[
        l(2) = \pscal{\phi_0, L_2^{-1} \phi_0} = \frac{1}{2}\partial_\omega \norm{\phi_0}^2.
    \]
    
    For $\mu \in (0,2)$, we claim that $m_a(0, H_\mu) \geq 3$ if and only if $l(\mu)= 0$. Indeed,
    \[
        \ker(H_\mu^2 ) = \ker (L_0 L_\mu) \times \ker (L_\mu L_0) = \vect \left\{\begin{pmatrix} L_\mu^{-1}\phi_0 \\0 \end{pmatrix}, \begin{pmatrix} 0 \\ \phi_0 \end{pmatrix}\right\},
    \]
    so $m_a(0,H_\mu)\geq 2$. Now fix $\mu  \in (0,2)$, and assume that $(\psi_1, \psi_2)\transp \in \ker(H_\mu^3)$. The corresponding equations are
    \[
        \left\{
        \begin{aligned}
            L_\mu L_0 L_\mu \psi_1 &= 0\,, \\
            L_0 L_\mu L_0 \psi_2 &= 0\,.
        \end{aligned}
        \right.
    \]
    Since $L_\mu$ is invertible, the first equation yields $\psi_1 \in \ker (L_0 L_\mu)$, and since $\ker(L_0)= \vect \{\phi_0\}$, the second equation implies that there exists $\alpha \in \C$ such that
    \[
        L_\mu L_0 \psi_2 = \alpha \phi_0\,.
    \]
    
    If $\alpha=0$,then  $\psi_2 \in \ker(L_\mu L_0)$ and $\ker(H_\mu^3) \subset \ker(H_\mu^2)$. By induction on~$n$, this implies that $\ker(H_\mu^n) \subset \ker (H_\mu^2)$ for all $n \in \N$, and~$m_a(0,H_\mu) = \dim (\ker(H_\mu^2))= 2$.
    
    Thus, if $m_a(0,H_\mu) \geq 3$, we need $\alpha \neq 0$ and, taking the inner product of both sides of
    \[
        L_0 \psi_2 = \alpha L_\mu^{-1} \phi_0\,
    \]
    with $\phi_0$ yields
    \[
        0 = \alpha \pscal{\phi_0,L_\mu^{-1} \phi_0 } = \alpha \, l(\mu)\,, 
    \]
    hence~$l(\mu) = 0$ as claimed. For the converse, if $l(\mu)= 0$, then
    \[
        L_\mu^{-1} \phi_0 \in \ker (L_0)^{\perp} = \ran(L_0)
    \]
    and thus there exists $\tilde \phi \notin \ker(L_0)$ such that 
    \[
        L_0 L_\mu L_0 \tilde \phi =0\,.
    \]
   This implies $m_a(0,H_\mu) \geq 3$ since, again,
    \[
        \ker(H_\mu^3) = \ker (L_\mu L_0 L_\mu) \times \ker (L_0 L_\mu L_0)\,.
    \]

    We now show that if
    \begin{equation}\label{eq:VK-inequality}
        \partial_\omega \norm{\phi_0(\omega)}^2_{L^2} \leq 0\,,
    \end{equation}
    then $l(\mu) < 0$ for all $\mu \in (0,2)$.
    First, $l$ is non-decreasing, since
    \[
        l'(\mu) = \lim_{h\searrow 0} h^{-1} \pscal{\phi_0, (L_{\mu +h}^{-1} - L_\mu^{-1}) \phi_0} = \lim_{h\searrow 0} \pscal{\phi_0, L_{\mu +h}^{-1} Q L_\mu^{-1} \phi_0} = \pscal{ L_{\mu}^{-1}\phi_0, Q L_\mu^{-1} \phi_0} \geq 0.
    \]

    If~\eqref{eq:VK-inequality} holds with a strict inequality, this is sufficient to conclude $l(\mu) \leq l(2) <0 $ for all $\mu \in (0,2)$, and therefore $m_a(0,H_\mu)= 2$.

    If the equality holds in~\eqref{eq:VK-inequality}, we have to show that $l$ is not constant on any interval $[\mu_0,2]$. Since $l$ is real analytic, this can only happen if $l(\mu)=0$ for all $\mu \in (0,2)$. But then $m_a(0,H_\mu) \geq 3$ for all $\mu \in (0,2)$.
    We now show that this is impossible for small $\mu$, since $m_a(0, H_0)=2$ and~$\norm{H_0 - H_\mu} = \mu \normt{Q}$. 
    Let $\delta = \operatorname{dist}(0, \sigma(H_0)\setminus\{0\}) > 0$. Then the circle $\partial B(0,\delta/2)$ is contained in the resolvent set $\rho(H_\mu)$ for all $\mu $ with $\abs{\mu} < \delta/(2 \normt{Q})$.
    Indeed, for all $\phi \in H^1(\R, \C^4)$ and~$z \in \partial B(0,\delta/2) $, we have
    \[
        \norm{(H_\mu - z)\phi} \geq \norm{(H_0 - z)\phi} - \norm{(H_0 - H_\mu) \phi} \geq\left( \frac{\delta}{2} - \mu \normt{Q}\right) \norm{\phi}.
    \]
    Thus, by~\cite{HislopSigal}, the sums of algebraic multiplicities of the eigenvalues of~$H_\mu$ within $B(0, \delta/2)$ is constant for $\mu \in [0, \delta/(3 \normt{Q}) )$. Therefore $m_a(0, H_\mu) =2$ for those~$\mu$'s.
\end{proof}

We end this section with a by-product of the proof of Lemma~\ref{Lemma_equiv_number_evs_L2_and_Lmu}, needed for the proof of Theorem~\ref{negative_spectrum_of_Lmu_nonrelativistic_limit}, about the smallest eigenvalue of~$L_\mu$ (strictly) above $-2\omega$.
\begin{lemma}\label{A_priori_bounds_on_first_ev_of_Lmu}
    Let $f$ satisfy Assumption~\ref{Assumption_general_nonlinearity}, $\omega\in(0,m)$, and~$\mu>0$. Then the first eigenvalue $\lambda_1$ of~$L_\mu$ in~$(-2\omega,m-\omega)$ exists and verifies
    \[
        -\mu\normt{Q} \leq \lambda_1(\mu) < 0\,.
    \]
    
    Moreover, if $-2\omega$ and~$0$ are the only eigenvalues of~$L_0$, then for any $\mu\geq0$, $L_\mu$ has no eigenvalues in the interval $(-m-\omega, -2\omega)$.
\end{lemma}
\begin{proof}
    The argument in the proof of Lemma~\ref{Lemma_equiv_number_evs_L2_and_Lmu} gives the existence and the negativity. Moreover, integrating~\eqref{Feynman_Hellmann_type_ineq} over $\mu$ for $\lambda_1$ gives
    \[
        \lambda_1 (\mu) = 0 - \int_0^\mu \pscal{\phi_\alpha, Q \phi_\alpha}\di\alpha \geq -\mu \normt{Q},
    \]
    where $\phi_\alpha$ is the normalized eigenvector associated to $\lambda_1(\alpha)$.
    
    The second result is obtained by the same Feynman--Hellmann type argument as in the proof of Lemma~\ref{Lemma_equiv_number_evs_L2_and_Lmu} and we refer again to Figure~\ref{fig:sketch_lemma}.
\end{proof}

\section{Lower bound on \texorpdfstring{$\Re z^2$}{Re z2}}\label{Section_Lower_Bound_Re_z2}
In this section, our goal is to establish bounds for eigenvalue branches of~$H_\mu$ of type~(b), those that are away from the real and imaginary axes. 
These bounds are non-trivial if the operator norm of $Q$ is sufficiently small. This is always the case in the non-relativistic limit, so a concise statement is given in the following theorem, while we refer to Lemma~\ref{technical_lemma} for the explicit bounds.

\begin{theorem}\label{Bound_on_Re_z2_and_Im_z_general_statement}
    Let $f$ satisfy Assumption~\ref{Assumption_general_nonlinearity}. For any $E \in [0, 2m)$, there exists $\omega_E \in [E/2, m)$ such that, for all $\omega \in [\omega_E, m)$ and~$\mu \in (0,2]$, if $z \in \C \setminus(\R \cup i \R)$ is an eigenvalue of~$H_\mu(\omega)$, then
	\begin{equation}\label{Bound_Re_z2_general_statement}
	    \Re z^2 \geq E^2\,.
	\end{equation}
\end{theorem}

Throughout this section we assume that $z \in \C $ is an eigenvalue of~$H_\mu$ with associated eigenvector $\varphi = (\varphi_1, \varphi_2)\transp$, where~$\varphi_1$ and~$\varphi_2$ are two-component spinors. We introduce the notations
\begin{equation*}
    P_- := E_{(-\infty, -2\omega)}(L_0) \quad \textrm{ and } \quad P_+ := E_{(0, \infty)}(L_0)
\end{equation*}
for the spectral projectors of~$L_0$ on the corresponding intervals. Note that the eigenfunctions $\phi_{-2\omega}$ and $\phi_0$, respectively associated to $-2\omega$ and~$0$, are excluded from the range of these projectors.
For a concise notation in the proofs below, we also define 
\[
    Q_{++} = P_{+} Q P_{+}, \quad Q_{--} = P_{-} Q P_{-}, \quad \textrm{ and } \quad Q_{+-} = Q_{-+}^* =  P_{+} Q P_{-} \,.
\]

We first establish an identity for eigenvalues.
\begin{lemma}\label{Identity_Re_z2}
    Let $f$ satisfy Assumption~\ref{Assumption_general_nonlinearity}, $\omega\in(0,m)$, and $\mu\geq0$. Assume that $z \in \C \setminus(\R \cup i \R)$ is an eigenvalue of~$H_\mu$ with eigenfunction $(\varphi_1, \varphi_2)\transp$.
    Then $\varphi_1 = (P_+ + P_-)\varphi_1$, $\pscal{\varphi_1, \varphi_2}= 0$, and
    \begin{align}\label{eq:bound_real_z^2}
        \Re z^2 = \frac{ \pscal{\varphi_1, \big(P_+ L_\mu P_+ - P_- L_\mu P_-\big) \varphi_1} }{ \pscal{\varphi_1, \abs{L_0}^{-1} \varphi_1} }\,,
    \end{align}
    where $L_0^{-1} \varphi_1$ is well-defined since $\varphi_1 \in \ker(L_0)^{\perp}$.
\end{lemma}
\begin{proof}
    Since there are no eigenvalues of~$L_0$ in~$(-2\omega, 0)$ by Theorem~\ref{groundstate}, we only 
    have to check that $\varphi_1$ is orthogonal to the eigenfunctions $\phi_0$ and~$\phi_{-2\omega}$ (associated to $0$ and~$-2 \omega$, respectively) in order to prove $\varphi_1 = (P_+ + P_-)\varphi_1$. 
    Taking the inner product of~\eqref{eq:eigen_of_H_split_a} with $\phi_0$ gives $z \pscal{\phi_0, \varphi_1} =0$. Since $z \neq 0$, this gives the first ortogonality condition. 
    The eigenvalue equation for $H_\mu^2$ reads
    \begin{equation}\label{eq:EV_H_mu^2}
        \begin{cases}
            L_\mu L_0 \varphi_2 = z^2 \varphi_2 \\
           L_0 L_\mu \varphi_1 = z^2 \varphi_1.
        \end{cases}
    \end{equation}
    Since $L_\mu \phi_{-2\omega} = L_0 \phi_{-2\omega}$, taking the inner product of the second line with $\phi_{-2\omega}$ gives 
    \[
        z^2 \pscal{\phi_{-2\omega}, \varphi_1} = 4 \omega^2  \pscal{\phi_{-2\omega}, \varphi_1}.
    \]
    Given that we assume $z^2 \notin \R$, this implies the orthogonality of~$\varphi_1$ to $\phi_{-2\omega}$.
    
    To prove $\pscal{\varphi_1, \varphi_2}= 0$, we start by taking the scalar product of~\eqref{eq:eigen_of_H_split_a} and~\eqref{eq:eigen_of_H_split_b} respectively with $\varphi_2$ and~$\varphi_1$.
    We obtain
    \[
        \left\{
        \begin{aligned}
            &\pscal{\varphi_2, L_0 \varphi_2} = z \pscal{\varphi_2, \varphi_1} \\
            &\pscal{\varphi_1, L_\mu \varphi_1} = z \pscal{\varphi_1, \varphi_2}.  
        \end{aligned}
        \right.
    \]
    By selfadjointness of~$L_0$ and~$L_\mu$, the left hand side of both equations is real.
    Taking their sum and difference gives
    \[
        \left\{
        \begin{aligned}
           2 &z \Re \pscal{\varphi_2, \varphi_1} =  \pscal{\varphi_2,  L_0 \varphi_2}+  \pscal{\varphi_1,  L_\mu \varphi_1} \in\R \\
           2 i &z \Im \pscal{\varphi_2, \varphi_1} =  \pscal{\varphi_2,  L_0 \varphi_2} -  \pscal{\varphi_1,  L_\mu \varphi_1} \in\R\,.  
        \end{aligned}
        \right.
    \]
    The first line yields $\Re\pscal{\varphi_2, \varphi_1} = 0 $, since $z\not\in\R$, and the second gives $\Im \pscal{\varphi_2, \varphi_1} =0$, since $z\not\in i\R$.
    
    We now turn to the proof of~\eqref{eq:bound_real_z^2}. We apply $L_0^{-1}$ to the second line of~\eqref{eq:EV_H_mu^2} and take the scalar product with $P_+ \varphi_1$ and~$P_- \varphi_1$. This gives
    \[
        \left\{
        \begin{aligned}
            &\pscal{P_+ \varphi_1, L_\mu \varphi_1} = z^2 \pscal{P_+ \varphi_1, L_0^{-1} \varphi_1} \\
            &\pscal{P_- \varphi_1, L_\mu \varphi_1} = z^2  \pscal{P_- \varphi_1, L_0^{-1} \varphi_1}. 
        \end{aligned}
        \right.
    \]
    We insert $\varphi_1 = (P_+ + P_-)\varphi_1$, and use the identity
    \[
        P_+ L_{\mu} P_- = P_+ L_{0} P_- - \mu P_+ Q P_- = - \mu Q_{+-}\,.
    \]
  We are left with
    \[
        \left\{
        \begin{aligned}
            &\pscal{ \varphi_1, P_+L_\mu P_+\varphi_1} - \mu \pscal{ \varphi_1, Q_{+-}\varphi_1} = z^2  \pscal{ \varphi_1, P_+ L_0^{-1}P_+ \varphi_1} \\
            &\pscal{ \varphi_1, P_- L_\mu P_-\varphi_1} - \mu  \pscal{ \varphi_1, Q_{-+}\varphi_1} = z^2  \pscal{ \varphi_1, P_- L_0^{-1}P_- \varphi_1}. 
        \end{aligned}
        \right.
    \]
    Subtracting both identities yields
    \[
        \pscal{ \varphi_1, \left( P_+ L_\mu P_+ - P_- L_\mu P_- \right)\varphi_1} - \mu  \pscal{ \varphi_1, \left(Q_{+-}- Q_{-+}\right)\varphi_1} = z^2  \pscal{ \varphi_1, \abs{L_0}^{-1} \varphi_1}. 
    \]
    Taking the real part eliminates the second term in the l.h.s.\ and gives identity~\eqref{eq:bound_real_z^2}.
\end{proof}

We exploit this identity to derive a lower bound on~$\Re z^2$. This bound depends on
\begin{equation}\label{Def_t}
    t := \min\{\lambda \,|\, \lambda \in \sigma(L_0 + \omega\Id) \cap (\omega,m] \}\,,
\end{equation}
which is either the smallest eigenvalue of $L_0 + \omega\Id$ above $\omega$, or $m$ (the bottom of its positive essential spectrum).

\begin{theorem}\label{Thm_ineq_Re_z2}
    Let $f$ satisfy Assumption~\ref{Assumption_general_nonlinearity}, $\omega\in(0,m)$, $\mu\geq0$, $\alpha\in[0,1]$, and~$t$ be as in~\eqref{Def_t}. Assume that $z \in \C \setminus(\R \cup i \R)$ is an eigenfunction of~$H_\mu$ with eigenfunction $(\varphi_1, \varphi_2)\transp$, and define $\eta>0$ as
    \[
        \eta := \frac{\norm{P_+ \varphi_1}^2}{\norm{P_- \varphi_1}^2}\,.
    \]
	
	If
	\begin{equation}\label{Lemma_ineq_Re_z2_restriction}
	    \left(1-\alpha^2\right)\eta^2 + (1+\alpha)^2 \geq \mu\eta  \frac{\normt{Q}}{t+\omega}\,,
	\end{equation}
	then
	\begin{equation}\label{Lemma_ineq_Re_z2_inequality}
	    \Re z^2 \geq \frac{(t+\omega)^2}{2} \left(\left(1-\alpha\right)\eta^2 + (1+\alpha) - \frac{\eta}{1+\alpha} \mu \frac{\normt{Q}}{t+\omega}\right) \geq 0\,.
	\end{equation}
\end{theorem}
\begin{proof}
    Combining the eigenvalue equations~\eqref{eq:eigen_of_H_split_a} and~\eqref{eq:eigen_of_H_split_b} with $\pscal{\varphi_1, \varphi_2}= 0$, we obtain
    \begin{equation}\label{Orthogonality_identities}
        \pscal{\varphi_1, L_\mu \varphi_1} = 0\,,  \quad \textrm{ and } \quad \pscal{\varphi_1, L_0^{-1} \varphi_1} = 0\,.
    \end{equation}
     For shortness, we define $\varphi_+ := P_+ \varphi_1$ and~$\varphi_- := P_- \varphi_1$. The above identities give
    \begin{equation}\label{eq:0-exp-value-2}
        \pscal{\varphi_+, L_\mu \varphi_+} + \pscal{\varphi_-, L_\mu \varphi_-} - \mu \pscal{\varphi_1,(Q_{+-} + Q_{-+}) \varphi_1} = 0
    \end{equation}
    and
    \begin{equation}\label{eq:0-exp-value}
        \pscal{\varphi_+, \abs{L_0}^{-1} \varphi_+} - \pscal{\varphi_-, \abs{L_0}^{-1}\varphi_-} = 0\,.
    \end{equation}
    As a first consequence of these identities, $\varphi_+ \neq 0 \neq \varphi_-$ hence $\eta>0$ is well-defined. Indeed, if $\varphi_+=0$ or $\varphi_-=0$, then the other one is also trivial leading by Lemma~\ref{Identity_Re_z2} to $\varphi_1=0$ which in turn yields, by the eigenvalue equation for $H_\mu$, the contradiction $z\varphi_2=0$.

    From the symmetry w.r.t.\ $0$ of the spectrum of~$L_0 + \omega\Id$ and the definition of~$t$,
    the spectrum of $P_- L_0$, restricted to $\operatorname{Ran} (P_-)$, is contained in the interval in $(-\infty, -t-\omega]$. Using that $|L_0| = -L_0$ on $\operatorname{Ran} (P_-)$, we obtain
    \begin{equation}\label{eq:bound-neg}
        \pscal{\varphi_-, \abs{L_0} \varphi_-} \geq (t+\omega) \norm{\varphi_-}^2
        \quad \text{ and } \quad
        \pscal{\varphi_-, \abs{L_0}^{-1} \varphi_-} \leq (t+\omega)^{-1} \norm{\varphi_-}^2.
    \end{equation}
    
    We now bound all quantities appearing in~\eqref{eq:bound_real_z^2}.
    For its denominator, combining~\eqref{eq:0-exp-value} with~\eqref{eq:bound-neg} gives
    \[
        \pscal{\varphi_1, \abs{L_0}^{-1} \varphi_1} = 2  \pscal{\varphi_-, \abs{L_0}^{-1} \varphi_-}  \leq 2 ( t + \omega)^{-1} \norm{\varphi_-}^2.
    \]

    The first term in the numerator of~\eqref{eq:bound_real_z^2} has to be treated in two different ways depending on the value of~$\eta$.
    From~\eqref{eq:0-exp-value-2}, we obtain on one hand
    \[
        \pscal{\varphi_+, L_\mu \varphi_+} = - \pscal{\varphi_-, L_\mu \varphi_-} + \mu \pscal{\varphi_1, (Q_{+-} + Q_{-+}) \varphi_1},
    \]
    which will be a useful identity to obtain bounds for small values of~$\eta$. 
    
    On the other hand, for large values of~$\eta$, we use Jensen's inequality for the convex function $0<\lambda \mapsto 1/\lambda$. This inequality reads
    \[
    	E_{\mathsf{m}}[\lambda] \geq E_{\mathsf{m}}[\lambda^{-1}]^{-1}, 
    \]
    where $E_\mathsf{m}$ denotes the expectation value with respect to any probability measure $\mathsf{m}$ on~$(0, +\infty)$.
    We take for $\mathsf{m}$ the probability measure $\norm{\varphi_+}_{L^2}^{-2} \mathsf{m}(L_0)_{\varphi_+}$, where $\mathsf{m}(L_0)_{\varphi_+}$ denotes the spectral measure of $L_0$ in the state $\varphi_+$, i.e., for any interval $I$, 
    \[
    	\mathsf{m}(L_0)_{\varphi_+}(I) = \pscal{\varphi_+, \Id_I(L_0) \varphi_+ }.
    \]
    This gives
    \[
       \frac{  \pscal{\varphi_+, L_0 \varphi_+} }{\norm{\varphi_+}^2}  \geq \left(  \frac{  \pscal{\varphi_+, L_0^{-1} \varphi_+} }{\norm{\varphi_+}^2} \right)^{-1} = \frac{\norm{\varphi_+}^2}{\pscal{\varphi_-, \abs{L_0}^{-1} \varphi_-}}\,,
    \]
    where we have also used~\eqref{eq:0-exp-value}.
    Combining it with~\eqref{eq:bound-neg} and the definition of~$\eta$, we obtain
    \begin{align*}
          \pscal{\varphi_+, L_\mu \varphi_+} 
          &= \pscal{\varphi_+, L_0 \varphi_+} - \mu \pscal{\varphi_+, Q \varphi_+} \\
          & \geq \norm{\varphi_+}^4 (t+\omega) \norm{\varphi_-}^{-2}- \mu \pscal{\varphi_+, Q \varphi_+} = \eta^2 (t+ \omega) \norm{\varphi_-}^{2}- \mu \pscal{\varphi_+, Q \varphi_+}.
    \end{align*}

    Inserting these two statements on~$\pscal{\varphi_+, L_\mu \varphi_+}$ into~\eqref{eq:bound_real_z^2}, we obtain
    \begin{align*}
        2\pscal{\varphi_-, \abs{L_0}^{-1}\varphi_-}\Re z^2 &= \alpha \pscal{\varphi_+, L_\mu \varphi_+} + (1-\alpha)\pscal{\varphi_+, L_\mu \varphi_+}  - \pscal{\varphi_-, L_\mu \varphi_-} \\
        &\geq
        \begin{multlined}[t]
            \alpha  \left(- \pscal{\varphi_-, L_\mu \varphi_-} + \mu \pscal{\varphi_1, (Q_{+-} + Q_{-+}) \varphi_1}\right) \\
            + (1-\alpha)\left( \eta^2 (t+ \omega) \norm{\varphi_-}^{2}- \mu \pscal{\varphi_+, Q \varphi_+}\right) - \pscal{\varphi_-, L_\mu \varphi_-},
        \end{multlined}
    \end{align*}
    for all $\alpha \in [0,1]$. After grouping terms and using again~\eqref{eq:bound-neg}, we are left with
    \[
       2\pscal{\varphi_-, \abs{L_0}^{-1}\varphi_-}\Re z^2 \geq
        \begin{multlined}[t]
            \left(\alpha (1-\eta^2)+ 1 + \eta^2 \right) (t+\omega)\norm{\varphi_-}^{2} \\
            + \mu \pscal{\varphi_1, \left((1+\alpha)Q_{--} + \alpha Q_{+-} + \alpha Q_{-+} - (1-\alpha) Q_{++}\right) \varphi_1}.
        \end{multlined}
    \]
     In the last term, we ``complete the square'' to obtain
     \begin{multline*}
        \pscal{\varphi, \bigl((1+\alpha)Q_{--} + \alpha Q_{+-} + \alpha Q_{-+} - (1-\alpha) Q_{++}\bigr) \varphi} \\
        \begin{aligned}
            &= (1+\alpha)\pscal{\varphi_- + \tfrac{\alpha}{1+\alpha} \varphi_+, Q\left(\varphi_- + \tfrac{\alpha}{1+\alpha} \varphi_+\right)} -\tfrac{\alpha^2}{1+\alpha} \pscal{\varphi_+, Q \varphi_+} - (1-\alpha)\pscal{\varphi_+, Q \varphi_+} \\
            &\geq - \left(\tfrac{\alpha^2}{1+\alpha}  + (1-\alpha)\right) \normt{Q} \norm{\varphi_+}^2 = -\tfrac{\eta}{1+ \alpha} \normt{Q} \norm{\varphi_-}^2 .
        \end{aligned}
     \end{multline*}
     This finally gives
     \[
        2\frac{\pscal{\varphi_-, \abs{L_0}^{-1}\varphi_-}}{\norm{\varphi_-}^2 }\Re z^2 
        \geq \bigl(\alpha (1-\eta^2)+ 1 + \eta^2 \bigr) (t+\omega) - \frac{\mu \eta}{1+\alpha}  \normt{Q}, \quad \forall\, \alpha \in [0,1]\,.
    \]
  
    Now, if~\eqref{Lemma_ineq_Re_z2_restriction} ---in Theorem~\ref{Thm_ineq_Re_z2}--- holds, i.e., nonnegative r.h.s., then the l.h.s.\ satisfies
    \[
        2(t+\omega)^{-1}\Re z^2 \geq 2\frac{\pscal{\varphi_-, \abs{L_0}^{-1}\varphi_-}}{\norm{\varphi_-}^2 }\Re z^2 \geq 0
    \]
    by the second inequality in~\eqref{eq:bound-neg}. Finally, still using that the r.h.s.\ is nonnegative, we conclude the proof of the lemma:
    \[
        \Re z^2 \geq \frac{(t+\omega)^2}{2} \left(\left(1-\alpha\right)\eta^2 + (1+\alpha) - \frac{\eta}{1+\alpha} \mu \frac{\normt{Q}}{t+\omega}\right) \geq 0, \quad \forall\, \alpha \in [0,1]\,. \qedhere
    \]
\end{proof}

An immediate consequence of the last bound is Corollary~\ref{Corollary_no_imaginary_evs}.
\begin{proof}[Proof of Corollary~\ref{Corollary_no_imaginary_evs}]
	Recall that $H_0$ has only real eigenvalues. Consider the eigenvalue branches of $H_\mu$ as $\mu$ increases from $0$ to $2$.
	We pick $\alpha = 1/2$ in the last formula of the previous proof and check that $\Re z^2 \geq 0$ if $\mu \normt{Q}/(t+\omega) \leq 3\sqrt{3}/2$.
	Since $t>\omega$ and $0\leq\mu\leq2$, condition~\eqref{eq:simplified_corollary} guarantees that
	\[
	    \Re z^2\geq 0 \quad \Leftrightarrow \quad |\Im z| \leq |\Re z|\,.
	\]
	Since $H_\mu$ is an analytic family, any branch of eigenvalues that goes to the imaginary axis must pass through zero. But this is not possible since, by Theorem~\ref{thm_VK_intro}, the algebraic multiplicity of zero, as an eigenvalue of~$H_\mu$, is constant (equal to $2$) for $\mu\in (0,2)$. 
\end{proof}

In order to complete the proof of Theorem~\ref{Bound_on_Re_z2_and_Im_z_general_statement}, we optimize in $\alpha$ in Theorem~\ref{Thm_ineq_Re_z2} and obtain the following lemma.
In order to state it, we define $\theta_+:[0,2)\to(0,3\sqrt{3}/8]$ by
\begin{equation}\label{General_nonlinearity_lower_bound_Re_z2_Def_theta_plus}
	\theta_+(\xi) := 2\frac{\left(2 - 3\xi + \sqrt{9(2-\xi)^2+8\xi}\right) \left(6-3\xi + \sqrt{9(2-\xi)^2+8\xi}\right)^{\frac{3}{2}}}{\left(14 - 9 \xi + 3 \sqrt{9(2-\xi)^2+8\xi}\right)^2}\,.
\end{equation}
\begin{lemma}\label{technical_lemma}
    Let $f$ satisfy Assumption~\ref{Assumption_general_nonlinearity}, $t$ be as in~\eqref{Def_t}, and $\theta_+$ as in~\eqref{General_nonlinearity_lower_bound_Re_z2_Def_theta_plus}. Assume that $\omega \in (0,m)$, $E \in [0,2m)$, and $\mu \geq 0$ are such that $E< t+\omega$ and
    \begin{equation}\label{General_nonlinearity_lower_bound_Re_z2_condition_on_Q}
        \frac{\mu \normt{Q}}{4 (t+\omega)} \leq \theta_+\!\left( \frac{2 E^2}{(t+\omega)^2} \right)\,.
    \end{equation}
    
    If $z \in \C \setminus (\R \cup i \R)$ is an eigenvalue of $H_\mu$, then
    \[
        \Re z^2 \geq E^2\,.
    \]
\end{lemma}
This lemma being technical, we defer its proof to Appendix~\ref{Appendix_proof_technical_lemma}.

 \begin{remark}
     The restriction $E< t+\omega$ ---which comes from our method of proof (in particular from the fact we have no bound on $\eta>0$ defined in Theorem~\ref{Thm_ineq_Re_z2}, hence we perform some minimization over it)--- implies that, within our method, the intersections of the hyperbola in Figure~\ref{fig:spectrum_H_2} with the real axis cannot go beyond, nor reach, the outer thresholds $\pm(m+\omega)$.
     Of course, a priori bounds on $\eta$ would improve our method and could lead to reach and go beyond the outer thresholds. Going beyond those thresholds is interesting because they are points from which can emerge eigenvalue branches of type~(b) ---see Section~\ref{Section_Main_results}--- when $\omega$ varies. See e.g.~\cite[last paragraph of~p.2]{BouCom-19} for a summary on points from which non-real eigenvalues can emerge.
 \end{remark}

We now complete the proof of Theorem~\ref{Bound_on_Re_z2_and_Im_z_general_statement}.
\begin{proof}[Proof of Theorem~\ref{Bound_on_Re_z2_and_Im_z_general_statement}]
      If $E=0$, since $ \lim_{\omega \to m}\normt{Q_\omega} = 0$ by Proposition~\ref{prop:basic_groundstate_properties} and $t>\omega$ by definition, there exists  $\omega_{E=0} \in (0,m)$ such that for all $\mu \in (0,2]$ and $\omega \in [\omega_{E=0}, m)$,
    \[
        \mu \frac{\normt{Q_\omega}}{4 (t + \omega)} \leq \theta_+(0) = \frac{3\sqrt{3}}{8}\,.
    \]
    By Lemma~\ref{technical_lemma}, this implies~\eqref{Bound_Re_z2_general_statement} in the case $E=0$.
    
    Now, for a fixed $E \in (0, 2m)$, we restrict our attention to $\omega$'s such that $\omega \in [E/2,m)$. In this case, we choose $\xi\equiv\xi(E,\omega,t)$ as
    \[
        \frac{E^2}{2 m^2} \leq \xi := \frac{2 E^2}{(t + \omega)^2} < \frac{E^2}{2 \omega^2} \leq 2\,,
    \]
    so we obtain $\lim_{\omega \to m}\xi(E, \omega, t) = E^2 /(2m^2) < 2 $.
    Since $ \lim_{\omega \to m}\normt{Q_\omega} = 0$ from Proposition~\ref{prop:basic_groundstate_properties}, there exists $\omega_E \in [E/2,m)$ such that for all $\mu \in (0,2]$ and $\omega \in [\omega_E, m)$,
    \[
        \mu \frac{\normt{Q_\omega}}{4 (t + \omega)} \leq \frac{\normt{Q_\omega}}{2 (t + \omega)} \leq \theta_+(\xi)\,.
    \]
    By Lemma~\ref{technical_lemma}, this implies $\Re z^2 \geq (t + \omega)^2 \xi / 2 = E^2$, concluding the proof.
\end{proof}

These results show that, in order to obtain quantitative bounds, we need to estimate or compute $\normt{Q}$. This is done in the next section for power nonlinearities.

\section{Power nonlinearities: explicit estimates}\label{Section_Power_nonlinearity_generalities}
From now on, we consider the case of power nonlinearities $f(s)=s |s|^{p-1}$, $p>0$. They satisfy Assumption~\ref{Assumption_general_nonlinearity}. In that case, the solitary wave solutions
\[
    \phi_0(x):=\phi_0(p, \omega; x) := \begin{pmatrix} v(p, \omega; x) \\ u (p, \omega; x)\end{pmatrix},
\]
can be found by explicitly integrating the ODE (see, e.g.,~\cite{ChuPel-06, CooKhaMihSax-10,LeeKuoGav-75,MerQuiCooKhaSax-12}), and are given for any~$(p,\omega)\in(0,+\infty)\times(0,m)$ by
\begin{align}
    v(x) := v(p, \omega; x) &:= \frac{1}{\sqrt{1-\nu\tanh^2(p\kappa x)}} \left[(p+1)(m-\omega)\frac{1-\tanh^2(p\kappa x)}{1-\nu\tanh^2(p\kappa x)}\right]^{\frac1{2p}} \label{Def_v_p}
    \intertext{and}
    u(x) := u(p, \omega; x) &:= \sqrt\nu \tanh(p\kappa x) \, v(p, \omega; x) \label{Def_u_p} \,,
\end{align}
where we have introduced the parameters
\begin{equation*}
    \kappa = \sqrt{m^2 - \omega^2} \quad \textrm{ and } \quad \nu = \frac{m-\omega}{m+\omega}\in(0, 1)\,.
\end{equation*}

These explicit formulae allow us to compute exactly the contribution of~$\phi_0$ in~$L_0$:
\begin{align}
     f\!\left(v^2 - u^2\right) = \left(v^2 - u^2\right)^p &= (p+1)(m-\omega)\frac{1-\tanh^2(p\kappa \cdot)}{1-\nu\tanh^2(p\kappa \cdot)} \label{v2minusu2tothep}\\
     &= 2m(p+1)\frac{\nu}{1+\nu}\frac{1-\tanh^2(p\kappa \cdot)}{1-\nu\tanh^2(p\kappa \cdot)}\,,\nonumber
\end{align}
where we used $v^2-u^2>0$ to remove the absolute value in the nonlinearity and we wrote it in two ways as the latter point of view will turn out to be useful in some of our proofs.
We immediately read from it that
\begin{equation}\label{infinity_norm_v2minusu2tothep}
    \norm{f\!\left(v^2 - u^2\right)}_{L^\infty(\R)} = \norm{\left(v^2 - u^2\right)^p}_{L^\infty(\R)} = (p+1)(m-\omega) = 2m(p+1)\frac{\nu}{1+\nu}\,.
\end{equation}

Inserting~\eqref{Def_v_p} and~\eqref{Def_u_p} in the definition~\eqref{Def_of_Q} of~$Q$ leads to the explicit expression
\begin{equation}\label{Expression_of_Q}
    \begin{aligned}[b]
        Q &= 
        p \left(v^2 - u^2\right)^{p-1} \begin{pmatrix} v^2 & -uv \\ -uv & u^2\end{pmatrix}
        \\
        &= p (p+1)(m-\omega)\frac{1-\tanh^2(p\kappa\cdot)}{\left(1-\nu\tanh^2(p\kappa\cdot)\right)^2} \begin{pmatrix} 1 & -\sqrt\nu \tanh(p\kappa\cdot) \\ -\sqrt\nu \tanh(p\kappa\cdot) & \nu \tanh^2(p\kappa\cdot)\end{pmatrix}.
    \end{aligned}
\end{equation}

\subsection{Rescaled operators}\label{Subsection_rescaled_operators}
Given the explicit spatial depends in~$p\kappa x$ of~$\phi_0$ and because it will often be convenient for shortness and clarity, we will use the convention that a tilde means a spatial rescaling by a factor $p\kappa$. That is, for instance,
\[
    \tilde{v}(x) = \tilde{v}(p, \omega; x) = \frac{1}{\sqrt{1-\nu\tanh^2(x)}} \left[(p+1)(m-\omega)\frac{1-\tanh^2(x)}{1-\nu\tanh^2(x)}\right]^{\frac1{2p}}.
\]
Analoguously, we also define the unitary operator $U_{p\kappa}$ through its adjoint
\[
    U_{p\kappa}^*\psi(x):=\sqrt{p\kappa}\psi(p\kappa x)\,,
\]
and, for any operator $O$, the corresponding operator $\tilde{O} := U_{p\kappa} O U^*_{p\kappa}$. For instance, $\tilde{D}_m := U_{p\kappa} D_m U^*_{p\kappa} = i p\kappa \sigma_2 \partial_x + m\sigma_3$,
\[
    \tilde{L}_\mu := U_{p\kappa} L_\mu U^*_{p\kappa} := \tilde{L}_\mu (p, \omega) := i p\kappa \sigma_2 \partial_x + m\sigma_3 - \omega \Id - |\pscalns{\widetilde{\phi_0},\sigma_3\widetilde{\phi_0}}|^p\sigma_3 - \mu \tilde{Q}\,
\]
with
\[
    \tilde{Q} := U_{p\kappa} Q U^*_{p\kappa} = p (p+1)(m-\omega) \frac{1-\tanh^2}{\left(1-\nu\tanh^2\right)^2} \begin{pmatrix} 1 & -\sqrt\nu \tanh \\ -\sqrt\nu \tanh & \nu \tanh^2\end{pmatrix},
\]
and
\[
    \tilde{H}_\mu := U_{p\kappa} H_\mu U^*_{p\kappa} = \begin{pmatrix} 0 & \tilde{L}_0 \\ \tilde{L}_\mu & 0 \end{pmatrix}.
\]
This operator $U_{p\kappa}$ leaves invariant spectra and maps eigenfunctions to eigenfunctions.

With these explicit formulae and notations in place, we can
start computations. We first check the Vakhitov--Kolokolov criterion, then we compute $\normt{Q(p,\omega,\cdot)}$ and plug it into the bounds in Section~\ref{Section_Lower_Bound_Re_z2}, in order to prove Theorem~\ref{Thm_H2_no_purely_imaginary_evs}.

\subsection{Vakhitov--Kolokov condition}\label{Section_VK_condition}
We check the Vakhitov--Kolokolov criterion for power nonlinearities. Since we are unable to compute a closed expression for $\norm{\phi_0(\omega)}_{L^2}$ for any $p>0$, the next lemma gives sufficient conditions on~$(p, \omega)$ to have $\partial_\omega \norm{\phi_0(\omega)}_{L^2}^2$ positive or negative.

\begin{lemma}\label{Check_VK_criterion}
	Let $f(s) = s |s|^{p-1}$.
	\begin{itemize}
		\item If $p\in(0,2]$, then
		\[
		    \forall\, \omega\in(0,m), \quad \partial_\omega\norm{\phi_0(p, \omega; \cdot)}_{L^2}^2 < 0\,.
		\]
		\item If $p>2$, then
		\begin{align*}
		    \forall\, \omega\leq \frac{p}{3p-4} m, \quad &\partial_\omega\norm{\phi_0(p, \omega; \cdot)}_{L^2}^2 < 0
		    \intertext{and there exists $\omega_+\in\left(\frac{p}{3p-4} m,\sqrt{\frac{p+1}{2p-1}}m\right]$ such that}
		    \forall\, \omega>\omega_+, \quad &\partial_\omega\norm{\phi_0(p, \omega; \cdot)}_{L^2}^2 > 0\,.
		\end{align*}
	\end{itemize}
\end{lemma}
\begin{proof}[Proof of Lemma~\ref{Check_VK_criterion}]
	First, since $\partial_\omega \nu = -\frac{2m}{(m+\omega)^2}<0$, we have that
	\[
	    \sgn \partial_\omega\norm{\phi_0(p, \omega; \cdot)}_{L^2}^2 = - \sgn \partial_\nu\norm{\phi_0(p, \omega; \cdot)}_{L^2}^2
	\]
	and we work with the latter in this proof.
	
	Using~\eqref{Def_v_p},~\eqref{Def_u_p}, $(p+1)(m-\omega)=2m(p+1)\frac{\nu}{1+\nu}$ and~$\frac2\kappa=\frac1m\frac{1+\nu}{\sqrt\nu}$, we have
	{\allowdisplaybreaks
	\begin{align*}
	    \norm{\phi_0}_{L^2}^2 &= \norm{v^2+u^2}_{L^1} = \frac1{p\kappa}\norm{\tilde{v}^2+\tilde{u}^2}_{L^1} = \frac2{p\kappa}\int_0^\infty\tilde{v}^2+\tilde{u}^2 \\
	        &= m^{\frac{1}{p}-1} 2^{\frac{1}{p}}\frac{(p+1)^{\frac1{p}}}{p}\frac{1+\nu}{\sqrt\nu} \left(\frac{\nu}{1+\nu}\right)^{\frac1{p}} \int_0^\infty \frac{1+\nu\tanh^2}{1-\nu\tanh^2}\left[\frac{1-\tanh^2}{1-\nu\tanh^2}\right]^{\frac1{p}} \\
            &=:m^{\frac{1}{p}-1} 2^{\frac{1}{p}}\frac{(p+1)^{\frac1{p}}}{p} F(\nu)\,,
	\end{align*}
	}%
	and our goal is to determine the sign of~$\partial_\nu F(\nu)$. We have
	\[
	F(\nu) = \frac{1+\nu}{\sqrt\nu} \left(\frac{\nu}{1+\nu}\right)^{\frac1{p}}\int_0^1 \frac{1+\nu y^2}{1-\nu y^2}\left[\frac{1- y^2}{1-\nu y^2}\right]^{\frac1{p}}\frac{\di y}{1-y^2} =: \int_0^1 h(\nu,y)\di y\,.
	\]
	
	Since $y\mapsto (1+\nu y^2)(1-\nu y^2)^{-(1+p^{-1})}$ is increasing on~$(0,1)$ for any $\nu\in(0,1)$, we have
	\[
	0 < h(\nu,y) \leq \nu^{\frac1{p}-\frac1{2}} \frac{(1+\nu)^{2-\frac1p}}{(1-\nu)^{1+\frac1p}} (1-y^2)^{\frac1p-1} \in L^1((0,1))\,,
	\]
	because $\int_0^1(1-y^2)^{\frac1p-1}\di y=\frac{\sqrt{\pi}}{2}\frac{\Gamma(\frac1p)}{\Gamma(\frac12+\frac1p)}<+\infty$ for any $p>0$. Hence $y\mapsto h(\nu,y)$ is Lebesgue integrable for any $\nu\in(0,1)$. Moreover, for any $y\in(0,1)$, $\partial_\nu h(\nu,y)$ exists for all $\nu\in(0,1)$ and
	\[
	    \partial_\nu h(\nu,y) = \frac{\nu^{\frac1p-\frac32}}{(1+\nu)^{\frac1{p}}} \times \frac{P\left(y^2\right)}{2 (1-\nu y^2)^{\frac1{p}+2} (1- y^2)^{1-\frac1{p}}}\,,
	\]
	where the polynomial
	\[
	    P_\nu(z):= \nu^2\left( 2 p^{-1}\nu + (1-\nu) \right)z^2 + \nu(1+\nu) \left( 2 p^{-1} + 4 \right)z + \left(\nu + 2 p^{-1} - 1\right)
	\]
	is strictly increasing on~$(0,1)$. And, finally, 
	\[
	    |\partial_\nu h(\nu,y)| \leq \frac{\nu^{\frac1p-\frac32} \left|P_\nu(1) \right|}{2 (1-\nu)^{\frac1{p}+2}(1+\nu)^{\frac1{p}}} \times \frac{1}{(1- y^2)^{1-\frac1{p}}}\,,
	\]
	which is in~$L^1((0,1))$ as a function of~$y$, for any $\nu\in(0,1)$. Thus
	\[
	    F'(\nu) = \partial_\nu \int_0^1 h(\nu,y)\di y = \int_0^1 \partial_\nu h(\nu,y)\di y\,.
	\]
	
	Now, a sufficient condition for $\partial_\omega\norm{\phi_0}_{L^2} < 0$ ---i.e., $\partial_\nu\norm{\phi_0}_{L^2} > 0$---, is for the polynomial $P_\nu$ to be positive on~$(0,1)$ or equivalently, since $P_\nu$ is strictly increasing on~$(0,1)$, that $P(0)=\nu + 2 p^{-1} - 1 \geq 0$.
	This proves the claim in the case $p\leq2$ as well as the first claim in the case $p>2$ but only in the subcase $\nu \geq 1- 2/p \Leftrightarrow \omega \leq m/(p-1)$.
	
	Similarly, a sufficient condition for $\partial_\omega\norm{\phi_0}_{L^2} > 0$ is for the polynomial $P_\nu$ to be negative on~$(0,1)$ or, equivalently, that $P(1) \leq 0$. Assuming $p>2$, and since $0<\nu<1$,
	\[
	    P(1)\leq 0 \Leftrightarrow (1+\nu)\left((p-2)\nu^2 - 6p\nu + p-2 \right) \geq 0 \Leftrightarrow \nu \leq \nu_*(p)\,,
	\]
	with
	\[
	    \nu_*(p) := \frac{3p}{p-2} - \frac{2}{p-2}\sqrt{(2p-1)(p+1)}\,.
	\]
	Defining $\omega_* := \frac{1-\nu_*}{1+\nu_*} m=\sqrt{\frac{p+1}{2p-1}}m$, we have proved the second claim in the case $p>2$.
    
    Now, for the first claim in the case $p>2$ for $\omega\in\left(\frac{m}{p-1},\frac{p}{3p-4} m\right] \Leftrightarrow \nu\in\left[\frac{1}{2}\frac{p-2}{p-1},1-\frac{2}{p}\right)$, we proceed as follow. Since the function $y\mapsto \partial_\nu h (\nu,y)$ is strictly increasing on~$(0,1)$ for any $\nu$, we have
    \[
        F'(\nu) = \int_0^z \partial_\nu h(\nu,y)\di y + \int_z^1 \partial_\nu h(\nu,y)\di y > z \partial_\nu h(\nu,0) + (1-z) \partial_\nu h(\nu,z)\,.
    \]
    Inserting $z=\sqrt{\nu}$ in the r.h.s., we obtain
    \begin{multline*}
        \textrm{r.h.s.} = \frac{\nu^{\frac1p-\frac32}}{(1+\nu)^{\frac1{p}}} \left\{\sqrt{\nu} \left(\nu+2 p^{-1}-1\right) \vphantom{+ \frac{ \left(2 p^{-1}\nu+1-\nu\right) \nu^3 + \nu^2(1+\nu)\left(2 p^{-1}+4\right) + \left(\nu+2 p^{-1}-1\right) }{(1+\sqrt{\nu}) (1+\nu)^{\frac{1}{p}} \left(1-\nu^2\right)^2}} \right. \\
        \left.  + \frac{ \left(2 p^{-1}\nu+1-\nu\right) \nu^3 + \nu^2(1+\nu)\left(2 p^{-1}+4\right) + \left(\nu+2 p^{-1}-1\right) }{(1+\sqrt{\nu}) (1+\nu)^{\frac{1}{p}} \left(1-\nu^2\right)^2} \right\}
    \end{multline*}
    and a sufficient condition to ensure $F'(\nu)>0$ ---i.e., $\partial_\omega\norm{\phi_0}_{L^2} < 0$--- is 
    \begin{multline*}
        \left(2 p^{-1}-1\right) \nu^4 + \left(2 p^{-1}+5\right) \nu^3 + \left(2 p^{-1}+4\right)\nu^2 \\
        + \left\{ 1 + \sqrt{\nu}(1+\sqrt{\nu}) (1+\nu)^{\frac{1}{p}} \left(1-\nu^2\right)^2 \right\}\left(\nu+2 p^{-1}-1\right) > 0\,.
    \end{multline*}
    Now, since we are in the case $\nu+2 p^{-1}-1<0$,
    we have
    \[
        \sqrt{\nu}(1+\sqrt{\nu})^2 \left(1-\nu^2\right)^2 >  \sqrt{\nu}(1+\sqrt{\nu}) (1+\nu)^{\frac{1}{p}} \left(1-\nu^2\right)^2 > 0, \qquad \forall\,\nu\in(0,1)
    \]
    and, studying the polynomial in $\sqrt\nu$, one can check that there exists $\nu_0\in\left(\left(\frac35\right)^2,\left(\frac23\right)^2\right)$ such that the l.h.s.\ attains its maximum at $\nu_0$ and is strictly monotonic on~$(0,\nu_0)$ and on~$(\nu_0,1)$. Therefore,
    \[
        \textrm{l.h.s.}(\nu)\leq \max\limits_{\left(\frac{9}{25},\frac49\right)} \textrm{l.h.s.} \leq \frac23\left(1+\frac23\right)^2 \left(1-\left(\frac35\right)^4\right)^2 <2\,,
    \]
    and we obtain the following sufficient condition to ensure $F'(\nu)>0$:
    \[
         \left(2 p^{-1}-1\right) \nu^4 + \left(2 p^{-1}+5\right) \nu^3 + \left(2 p^{-1}+4\right)\nu^2 + 3\nu + 3\left(2 p^{-1}-1\right) > 0\,.
    \]
    Finally, this polynomial in~$\nu$ being strictly increasing on~$(0,1)$ and positive at $\nu=\frac{1}{2}\frac{p-2}{p-1}$ concludes the proof that $\partial_\omega\norm{\phi_0}_{L^2} < 0$ for $p>2$ and~$\omega\in\left(\frac{m}{p-1},\frac{p}{3p-4} m\right]$.
\end{proof}

\subsection{Estimates on \texorpdfstring{$\beta(p)$}{beta(p)}}\label{Subsection_Power_nonlinearity_Lower_Bound_Re_z2}
We derive here lower bounds on~$\omega$ for condition~\eqref{General_nonlinearity_lower_bound_Re_z2_condition_on_Q} to hold. We do that by means of our general bound on~$\Re z^2$ combined with the explicit formula of the norm~$\normt{Q_\omega}$ of~$Q_\omega$.

\begin{lemma}\label{Lemma_Q}
The operator $Q_\omega$, $\omega\in(0,m)$, which acts as point-wise multiplication by the two by two matrix given in~\eqref{Expression_of_Q}, is positive semi-definite and satisfies
\begin{equation}\label{Norm_Q}
    \normt{Q_\omega} =
    \left\{
        \begin{aligned}
            &2p(p+1)m\frac{\nu}{1+\nu}=p(p+1)(m-\omega), \quad & \textrm{if }\, \nu < \frac13 \Leftrightarrow \omega > \frac{m}2\,,\\
            & p\frac{p+1}{2}\frac{m}2\frac{1+\nu}{1-\nu}=p\frac{p+1}2\frac{m^2}{2\omega}, \quad & \textrm{if }\, \nu \geq \frac13 \Leftrightarrow \omega \leq \frac{m}2\,.
        \end{aligned}
    \right.
\end{equation}
\end{lemma}
\begin{proof}
For all $x \in \R$, the matrix $Q_\omega(x)$ is positive semi-definite
with eigenvalues $0$ and~$p (v^2 - u^2)^{p-1} (v^2 + u^2)$, hence
\begin{equation}\label{Norm_Q_formula_norm_infinity}
    \normt{Q_\omega} = p \norm{\left(v^2 - u^2\right)^{p-1} \left(v^2 + u^2\right)}_\infty = p \norm{\left(\tilde{v}^2 - \tilde{u}^2\right)^{p-1} \left(\tilde{v}^2 + \tilde{u}^2\right)}_\infty.
\end{equation}
We have
\[
    \left(\tilde{v}^2 - \tilde{u}^2\right)^{p-1} \left(\tilde{v}^2 + \tilde{u}^2\right) = \left(\tilde{v}^2 - \tilde{u}^2\right)^p \frac{\tilde{v}^2 + \tilde{u}^2}{\tilde{v}^2 - \tilde{u}^2} = 2m(p+1)\frac{\nu}{1+\nu} F,
\]
where $F$ is the even function
\[
    F:=\frac{1-\tanh^2}{1-\nu\tanh^2} \times \frac{1 + \nu\tanh^2}{1 - \nu\tanh^2}\,.
\]
Its derivative satisfies
\[
	\left(1 - \nu\tanh^2\right)^3 F' = 2\left(3\nu - 1 - \nu(3-\nu)\tanh^2\right) \left(1-\tanh^2\right) \tanh.
\]
Therefore, on~$(0,+\infty)$,
\begin{align*}
    F'(t)>0
    &\Leftrightarrow
    3\nu - 1 > \nu(3-\nu)\tanh^2(t)\\
    &\Leftrightarrow
    \left\{
        \begin{aligned}
            &t < t_\nu := \arctanh\sqrt{\frac{3\nu - 1}{\nu(3-\nu)}}, & \textrm{if } \quad \nu \geq \frac13 \Leftrightarrow \omega \leq \frac{m}2\,,\\
            &t\in\emptyset, & \textrm{if } \quad \nu < \frac13 \Leftrightarrow \omega > \frac{m}2\,.
        \end{aligned}
    \right.
\end{align*}
Consequently, $\norm{F}_\infty=F(0)=1$ and~$\norm{(v^2 - u^2)^{p-1} (v^2 + u^2)}_\infty=2m(p+1)\frac{\nu}{1+\nu}$ if $\nu < \frac13$, otherwise $\norm{F}_\infty=F(t_\nu)=\frac18\frac{1+\nu}{\nu}\frac{1+\nu}{1-\nu}$ and~$\norm{(v^2 - u^2)^{p-1} (v^2 + u^2)}_\infty=\frac{m}2 \frac{p+1}2\frac{1+\nu}{1-\nu}$. We therefore have proved~\eqref{Norm_Q}.
\end{proof}

We end this section by explaining how to combine the expression for $\normt{Q}$ with the bounds of the previous section in order to estimate the range of $(p,\omega)$ for which eigenvalues~$z \in \C \setminus (\R \cup i \R)$ associated to~$H_2$ satisfy~$\Re z^2 \geq E^2$. For the sake of completeness, we give formulae in the general case $t>\omega$, but also in the case $t = m$ (i.e., assuming $L_0$ has no positive eigenvalues). From these formulae, we will then derive Theorem~\ref{Thm_H2_no_purely_imaginary_evs} at the end of this section and Lemma~\ref{restrictions_due_Re_z_above_threshold} at the end of the next section.

By Lemma~\ref{Lemma_Q} and defining $\theta\equiv\theta(\omega, t, \mu, p):=\frac{\mu\normt{Q_\omega}}{4(t+\omega)}$, the condition~\eqref{General_nonlinearity_lower_bound_Re_z2_condition_on_Q} on~$\omega$ becomes
\[
    4 \theta_+\!\left( \frac{2 E^2}{(t+\omega)^2} \right) \geq 4 \theta = \mu\frac{\normt{Q_\omega}}{t+\omega}
    =\left\{
        \begin{aligned}
            &\mu p(p+1)\frac{m-\omega}{t+\omega}, \quad & \textrm{if }\, \omega > \frac{m}2\,,\\
            & \mu p\frac{p+1}{4}\frac{m^2}{\omega(t+\omega)}, \quad & \textrm{if }\, \omega \leq \frac{m}2\,.
        \end{aligned}
    \right.
\]
We recall that $\theta_+$ is defined in~\eqref{General_nonlinearity_lower_bound_Re_z2_Def_theta_plus}.

Since $\omega\mapsto\theta_+(2 E^2 / (t+\omega)^2 )$ is non-decreasing on $(\max\{0,E-t\},m)$ ---constant if~$E=0$ and strictly inscreasing if~$E>0$---, because $\theta_+$ is strictly decreasing, and since~$\omega \mapsto \theta(\omega, t)$ is strictly decreasing (for $\mu, p>0$) from $+\infty$ to $0$, determining the $\omega$'s for which~\eqref{General_nonlinearity_lower_bound_Re_z2_condition_on_Q} holds is equivalent to finding the unique~$\tilde\omega\in(0,m)$ for which equality in~\eqref{General_nonlinearity_lower_bound_Re_z2_condition_on_Q} holds:
\begin{equation}\label{Section_5_equation_defining_optimal_omega}
	\theta(\tilde\omega, t,\mu, p) = \theta_+\!\left( \frac{2 E^2}{(t+\tilde\omega)^2} \right).
\end{equation}
Indeed, our $\omega$'s are then those verifying $\omega\geq\tilde\omega$.

Moreover, $t\in[\omega,m]$, $t\mapsto\theta_+(2 E^2 / (t+\tilde\omega)^2 )$ is nondecreasing, and $t \mapsto \theta(\tilde\omega, t)$ is decreasing. Thus, taking $t=\tilde\omega$ gives us an upper bound on this $\tilde\omega$, that is, a sufficient condition on $\omega$ for~\eqref{General_nonlinearity_lower_bound_Re_z2_condition_on_Q} to hold. Meanwhile taking $t=m$ gives us a lower bound on this $\tilde\omega$, that is, \emph{necessary condition}, i.e., the largest possible range of $\omega$'s that one can obtain with the bounds that we have.

Unfortunately, if $E>0$, finding a closed formula for these conditions, let alone for $\tilde\omega$ itself, seems out of reach. However, for fixed explicit values for $\mu>0$ and $p>0$, and either $t=m$ or $t=\omega$, a dichotomy easily approximates $\tilde\omega$. We obtained that way the corresponding curves in Figure~\ref{fig:spectrum_H_2}.

Nevertheless, for the choice $E=0$, for which the condition $\omega\geq E/2$ is hence trivially verified, the computations are explicit since the r.h.s.\ of~\eqref{Section_5_equation_defining_optimal_omega} is the constant $\theta_+(0)=3\sqrt{3}/8$ and we can prove the absence of non-zero eigenvalues of~$H_2$ on the imaginary axis for couples $(p, \omega)$ as plotted in Figure~\ref{Power_nonlin_admissible_ranges_Re_z2_and_VK}.
In order to state the result we define
\begin{equation}\label{Power_nonlin_threshold_omega_for_Re_z2}
	\left\{
	\begin{aligned}
		\beta(p) &= \frac{2p(p+1)}{2p(p+1) + 3\sqrt{3}}, &\qquad\textrm{if } p \geq \frac{\sqrt{1 + 6\sqrt{3}} - 1}{2}\,,\\
		\beta(p) &= \sqrt{\frac{p(p+1)}{6\sqrt{3}}}, &\qquad\textrm{if } p \leq \frac{\sqrt{1 + 6\sqrt{3}} - 1}{2}\,.
	\end{aligned}
	\right.
\end{equation}
\begin{theorem}\label{Thm_H2_no_purely_imaginary_evs}
    Let $f(s) = s |s|^{p-1}$, $p>0$, and $\beta(p)$ defined in~\eqref{Power_nonlin_threshold_omega_for_Re_z2}. Then, $H_2$ has no non-zero eigenvalues on the imaginary axis for
    \begin{itemize}
        \item $p\leq2$ and~$\omega \in [\beta(p) m, m)$;
        \item $p>2$ and~$\omega \in \left[\beta(p) m, \frac{p}{3p-4} m\right]$.
    \end{itemize}
\end{theorem}
\begin{figure}[ht]
	\includegraphics[width=0.8\columnwidth]{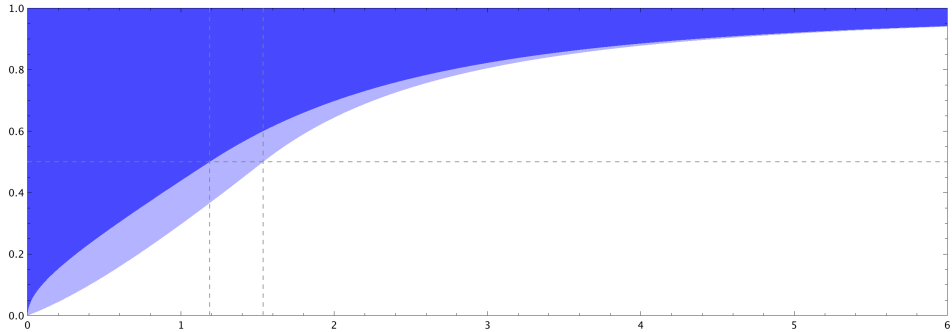}
	\captionsetup{width=.9\textwidth}
	\caption{Dark blue: range $\omega/m > \beta(p)$, a sufficient condition to have the \emph{cone condition} $\Re z^2 \geq 0$ for eigenvalues $z$ of $H_2$ outside the imaginary axis, independently of the spectrum of $L_0$. 
	Light blue: range of $\omega/m$'s for which the same condition holds under the additional assumption that $L_0$ has no positive eigenvalues. 
	Vertical dashed lines: $p_\circ(2)\in(1.18,1.19)$ and $p_*(2)\in(1.53,1.54)$.
	Horizontal dashed line:~$\omega/m=1/2$.}
	\label{Fig_Range_lower_bound_on_Re_z2}
\end{figure}
\begin{proof}[Proof of Theorem~\ref{Thm_H2_no_purely_imaginary_evs} and of the extra claim of Figure~\ref{Fig_Range_lower_bound_on_Re_z2}]
	The upper bounds in the theorem are due to Lemma~\ref{Check_VK_criterion}.
	
	The lower bound on $\omega$ in the theorem (common to both cases) is the sufficient condition: taking $t=\omega$ in~\eqref{Section_5_equation_defining_optimal_omega}. For the sake of generality, we compute it for general $\mu>0$. The sufficient condition becomes
	\begin{equation*}
	    \left\{
	    \begin{aligned}
	        &\frac{3\sqrt{3}}{2} \geq \mu p(p+1)\frac{m-\omega}{\omega} &&\Leftrightarrow \quad \omega \geq \frac{\mu p(p+1)}{\mu p(p+1) + 3\sqrt{3}} m, &\qquad\textrm{if } \omega > \frac{m}{2}\,,\\
	        &\frac{3\sqrt{3}}{2} \geq \mu \frac{p(p+1)}{4}\frac{m^2}{2\omega^2} &&\Leftrightarrow \quad \omega \geq \sqrt{\frac{\mu p(p+1)}{3\sqrt{3}}} \frac{m}{2},\, &\qquad\textrm{if } \omega \leq \frac{m}{2}\,,
	    \end{aligned}
	    \right.
	\end{equation*}
	and denoting the unique $p>0$ at which both lower bounds are equal to $m/2$ by
	\[
	    p_\circ \equiv p_\circ(\mu) := \left(\sqrt{1 + 12\mu^{-1}\sqrt{3}} - 1\right) \bigg/ 2
	\]
	---that is, $p_\circ>0$ s.t.\ $\mu p_\circ(p_\circ+1) = 3\sqrt{3}$---, we have equivalently
	\[
	    \omega \geq \omega_\circ(\mu) :=
	    \left\{
	    \begin{aligned}
	        &\frac{\mu p(p+1)}{\mu p(p+1) + 3\sqrt{3}} m, \quad &\textrm{if } p > p_\circ\,,\\
	        &\sqrt{\frac{\mu p(p+1)}{12\sqrt{3}}} m, &\textrm{if } p \leq p_\circ\,.
	    \end{aligned}
	    \right.
	\]
	Evaluating it at $\mu=2$ gives the lower bound of the theorem and $p_\circ=(\sqrt{1 + 6\sqrt{3}} - 1)/2$.
	
	The lower bound on the light blue area in Figure~\ref{Fig_Range_lower_bound_on_Re_z2} is the necessary condition: taking $t=m$ in~\eqref{Section_5_equation_defining_optimal_omega}. A similar computation on the necessary condition yields
	\[
	    \omega \geq \omega_* :=
	    \left\{
	    \begin{aligned}
	        &\frac{2\mu p(p+1) - 3\sqrt{3}}{2\mu p(p+1) + 3\sqrt{3}} m, &\textrm{if } p > p_*\,,\\
	        &\left(\sqrt{1+\frac{2\mu p(p+1)}{3\sqrt{3}}} - 1\right) \frac{m}{2}, \quad &\textrm{if } p \leq p_*\,,
	    \end{aligned}
	    \right.
	\]
	where $p_*$ is the unique $p>0$ at which both lower bounds are equal to $m/2$:
	\[
	    p_* \equiv p_*(\mu) := \left(\sqrt{1 + 18\mu^{-1}\sqrt{3}} - 1\right) \bigg/2\,
	\]
	Taking $\mu=2$ gives the light blue area and $p_*=(\sqrt{1 + 9\sqrt{3}} - 1)/2$\,.
\end{proof}

\section{The massive Gross--Neveu model: \texorpdfstring{$f(s)=s$}{f(s)=s}} \label{Section:Gross_Neveu}

In the case $f(s)=s$ (i.e., $p=1$), we can actually prove that $L_0$ has no other eigenvalues than $-2\omega$ and~$0$. To do so, we need the following result on resonances.
We recall that a general one-dimensional Dirac operator of the form $D_m - \omega - V$, where $V$ is some decaying, matrix valued potential, has a resonance at  $\lambda \in \R$ if the ordinary differential equation
\[
	(D_m -\omega -V) \psi = \lambda \psi
\]
has solutions $\psi$ in~$L^\infty(\R)$ which are not in $L^2(\R)$.

In the case at hand, it is known that $L_0$ has resonances with an explicit solution for the \emph{generalized eigenvalue} $\psi$. The formulae for these resonances that appear in~\cite[Lemma 5.5]{BerCom-12} seem to contain a typo in the expression of the second component $S$ and since we were not able to locate a derivation of the expression in the literature, we include the details.

\begin{lemma}\label{Resonances_L0_p_equal_1}
	Let $f(s)=s$ and $\omega\in(0,m)$. Then the values $m-\omega$ and~$-m-\omega$ are resonances of~$L_0$ with respective generalized eigenfunctions~$\left(R,S\right)\transp$ and~$\left(S,R\right)\transp$, where
	\[
	    R := \frac{uv}{v^2-u^2}
	    \quad \textrm{ and } \quad
	    S := -\frac{\nu}{1-\nu} \frac{v^2-\nu^{-1}u^2}{v^2-u^2} \,.
	\]
\end{lemma}
\begin{proof}
    We equivalently prove the result for the rescaled problem defined in Section~\ref{Subsection_rescaled_operators}. Namely, that
    \[
        \tilde{L}_0\begin{pmatrix}\tilde{R}\\\tilde{S}\end{pmatrix} = (m-\omega)\begin{pmatrix}\tilde{R}\\\tilde{S}\end{pmatrix}.
    \]
    The resonance at $-m-\omega$ is obtained by symmetry of the spectrum of~$L_0$ w.r.t.\ $-\omega$ (see Proposition~\ref{Prop:eigenvalues_symmetry}), with corresponding (generalized) eigenfunction obtained by exchanging the spinors.
    
    Rewriting
    \[
        \tilde{R} = \frac{\sqrt{\nu}\tanh}{1-\nu\tanh^2} \quad \textrm{ and } \quad \tilde{S} = -\frac{\nu}{1-\nu} \frac{1-\tanh^2}{1-\nu\tanh^2} = \frac{1}{1-\nu\tanh^2} - \frac{1}{1-\nu}\,,
    \]
    we compute
    \begin{align*}
        \tilde{M} &= m - f(\tilde{v}^2 - \tilde{u}^2) = m - 2(m-\omega)\frac{1-\tanh^2}{1-\nu\tanh^2}= m \left( 1 - 4\frac{\nu}{1+\nu}\frac{1-\tanh^2}{1-\nu\tanh^2} \right),\\
        \tilde{R}' &= \sqrt{\nu} \frac{1-\tanh^2}{\left(1-\nu\tanh^2\right)^2} \left(1+\nu\tanh^2\right)
        = -2 \frac{m}{\kappa}\frac{1-\nu}{1+\nu} \frac{1+\nu\tanh^2}{1-\nu\tanh^2} \tilde{S}\\
        \intertext{and}
        \tilde{S}' &= \left(\frac{1}{1-\nu\tanh^2}\right)' = 2\nu \frac{1-\tanh^2}{\left(1-\nu\tanh^2\right)^2} \tanh = 4\frac{m}{\kappa}\frac{\nu}{1+\nu} \frac{1-\tanh^2}{1-\nu\tanh^2} \tilde{R} \,,
    \end{align*}
    and conclude the proof since
    \begin{align*}
       \frac1m\left( \tilde{L}_0 + \omega \right)\begin{pmatrix}\tilde{R}\\\tilde{S}\end{pmatrix} &= \frac{\kappa}{m}\begin{pmatrix}\tilde{S}'\\-\tilde{R}'\end{pmatrix} + \frac{\tilde{M}}{m} \begin{pmatrix}\tilde{R}\\-\tilde{S}\end{pmatrix} \\
            &= \begin{pmatrix}4\frac{\nu}{1+\nu} \frac{1-\tanh^2}{1-\nu\tanh^2} \tilde{R}\\2\frac{1-\nu}{1+\nu} \frac{1+\nu\tanh^2}{1-\nu\tanh^2} \tilde{S}\end{pmatrix} + \left( 1 - 4\frac{\nu}{1+\nu}\frac{1-\tanh^2}{1-\nu\tanh^2} \right) \begin{pmatrix}\tilde{R}\\-\tilde{S}\end{pmatrix}= \begin{pmatrix} \tilde{R}\\ \tilde{S}\end{pmatrix},
    \end{align*}
    where the last equality for the lower spinor is due to
    \[
        2(1-\nu)(1+\nu\tanh^2) + 4\nu(1-\tanh^2) = 2(1+\nu)(1-\nu\tanh^2) \,. \qedhere
    \]
\end{proof}

We can now prove that, in the massive Gross--Neveu model, $L_0$ has no other eigenvalues than $-2\omega$ and~$0$.
\begin{lemma}[Spectrum of~$L_0$ for $f(s)=s$]\label{L0_no_other_evs_p_equal_1}
	Let $f(s)=s$ and $\omega\in(0,m)$. Then
	\[
	    \sigma\left(L_0\right) \cap (-m-\omega, m-\omega) = \{-2\omega, 0\}\,.
	\]
\end{lemma}
\begin{proof}
    Thanks to Theorem~\ref{groundstate}, we are left with proving that $L_0$ has no eigenvalues in~$(-m-\omega,m-\omega)\setminus[-2\omega,0]$. By symmetry of the spectrum, it is sufficient to prove that there are no eigenvalues in~$(0,m-\omega)$. To proceed, we will use the resonances given in Lemma~\ref{Resonances_L0_p_equal_1}.
    
	As in the proof of Theorem~\ref{groundstate}, we define~$A=L_0 + \omega\Id$, which admits the resonance $+m$ with the same generalized eigenfunctions $(R,S)\transp$ as $L_0$, and for which we equivalently have to prove that it has no eigenvalues in~$(\omega,m)$. After the same coordinate transformation $U$ as in the proof of this theorem, the resonance corresponds to bounded solutions
	\[
	    R \mp S = \frac{uv}{v^2-u^2} \pm \frac{\nu}{1-\nu} \times \frac{v^2-\nu^{-1}u^2}{v^2-u^2} = \pm \frac{\nu}{1-\nu} \times \frac{1  \pm\frac{1-\nu}{\sqrt\nu}\tanh(\kappa\cdot) - \tanh^2(\kappa\cdot)}{1-\nu\tanh^2(\kappa\cdot)}
	\]
	respectively to the equations
	\[
	    (- \partial_x^2 + M^2 \pm \prim{M} - m^2) f = 0\,.
	\]
	If these solutions have a single zero then, by Sturm's oscillation theorem and since we know by Theorem~\ref{groundstate} that the groundstate energy of~$- \partial_x^2 + M^2 \pm \prim{M}$ is $\omega^2$, it shows that there are no eigenvalues of~$A^2$ in the interval $(\omega^2, m^2)$ and we are done.
	
	To show that these solutions have a single zero in $\R$, we check that the polynomials $h_\pm(y) := 1 \pm\frac{1-\nu}{\sqrt\nu}y - y^2$, appearing in  the numerators, have a single zero in $(-1,1)$. Since $h_-(x) = h_+(-x)$, we study~$h_+$. The roots of $h_+$ are $- \sqrt{\nu}$ and $\sqrt{\nu}^{-1}$. Since $\nu \in (0,1)$, $h_+$ has a single zero on~$(-1,1)$.
\end{proof}

This means, in particular, that $t$ defined in~\eqref{Def_t} equals $m$, the bottom of the essential spectrum of $L_0 + \omega\Id$, and we can obtain a larger range of~$\omega$'s in Theorem~\ref{Thm_H2_no_purely_imaginary_evs} than the one given for the general case $p>0$. Indeed, $t=m$ implies that our necessary condition developed in Section~\ref{Subsection_Power_nonlinearity_Lower_Bound_Re_z2} is actually a sufficient condition too. We therefore obtain, in the massive Gross--Neveu model, that
\begin{equation*}
    \beta(1) = \frac{\sqrt{1+\frac{8}{3\sqrt{3}}} - 1}{2} \gtrapprox 0.2968\,.
\end{equation*}

Finally, evaluating at $\omega/m=7/20$, for the choice $E=m-\omega$ ---for which the condition $\omega \geq E/2$ becomes $\omega/m\geq 1/3$---, the values of $\theta(\omega, t=m,\mu=2, p=1)$ and $\theta_+(\omega,t=m, E=m-\omega)$, we find that~\eqref{General_nonlinearity_lower_bound_Re_z2_condition_on_Q} holds and have therefore obtained the following which compares $\Re z$ (for eigenvalues not lying on the axes) to the inner thresholds of the essential spectrum of~$H_2$.
\begin{lemma}\label{restrictions_due_Re_z_above_threshold}
    Let $f(s)=s$, $\frac{7}{20}\leq\frac{\omega}{m}<1$, and~$z \in \C \setminus(\R \cup i \R)$ be an eigenvalue of~$H_2$. Then
    \[
        \left|\Re z\right|^2 > \Re z^2 > (m-\omega)^2.
    \]
\end{lemma}
\begin{remark}
    Of course, this exact factor $\frac{7}{20}$ is not optimal within our framework, but it is close since a dichotomy finds it being larger than $0.3448$.
\end{remark}

\section{Power nonlinearities: \texorpdfstring{$L_2$}{Lmu} has only one eigenvalue in~\texorpdfstring{$(-2\omega,0)$}{(-2w,0)}} \label{Section_minmax_principle_and_Lmu}

For this final section, we return to the general case of power nonlinearities $f(s) = s |s|^{p-1}$.
The goal of this section is to determine the number of eigenvalues of~$L_2$ between the eigenvalues $-2\omega$ and~$0$. 
\begin{theorem}\label{Thm_L2_only_one_ev}
    Let $f(s) = s |s|^{p-1}$, with $p>0$, and $\omega\in(0,m)$.
    Then, $L_2$ has exactly one eigenvalue in~$(-2\omega, 0)$.
\end{theorem}
We expect this result to hold for groundstates with general $f$, at least if~$f(s)$ behaves as a power for small $s$.

By simplicity and continuity with $\omega$ of eigenvalues, this number is independent of~$\omega \in (0,m)$, and therefore, it is sufficient to compute it for $\omega$ close to $m$.
Therefore, it is sufficient to find lower bounds on the eigenvalues of~$L_\mu$, $\mu\geq0$, as $\omega$ approaches $m$.
We do this by expanding the eigenvalues of $L_\mu$ in the parameter $\kappa = \sqrt{m^2 -\omega^2}$.
In this limit, the operators converge to the free Dirac operator, with the first-order correction given by a Schr{\"o}\-dinger operator with P\"{o}schl--Teller potential, whose eigenvalues and eigenfunctions are known explicitly.
We start with a summary of the results, followed by a reminder about the minmax principle for operators with a gap, and the proof of the expansions.

\subsection{Eigenvalues of \texorpdfstring{$L_\mu$}{L-mu} in the non-relativistic limit}

We define, for any $p>0$ and~$\mu\geq0$, the real number
\begin{equation}\label{Def_s}
    \mathfrak{s}:=\frac{\sqrt{1+4\frac{p+1}{p^2}(1+p\mu)} - 1}{2} \,,
\end{equation}
which is equivalently defined as the positive real number such that
\[
    \mathfrak{s}(\mathfrak{s}+1) = \frac{p+1}{p^2}(1+p\mu)\,.
\]
\begin{remark}
    We recall the notations $\lceil x\rceil$ and~$\lfloor x\rfloor$: for $x\in\R$, $\lceil x\rceil$ is the smallest integer larger or equal to $x$ and~$\lfloor x\rfloor$ the largest integer smaller or equal to $x$. That is, $\lceil x\rceil, \lfloor x\rfloor \in\Z$ such that $\lceil x\rceil-1 < x \leq \lceil x\rceil$ and~$\lfloor x\rfloor \leq x < \lfloor x\rfloor +1$.
\end{remark}
\begin{theorem}[Eigenvalues in the non-relativistic limit]\label{negative_spectrum_of_Lmu_nonrelativistic_limit}
    Let $f(s) = s |s|^{p-1}$, $p>0$, $\mu\geq0$, and $\mathfrak{s}$ be as in~\eqref{Def_s}.
   For sufficiently small $\kappa=\sqrt{m^2-\omega^2}\to0^+$, the operator $L_\mu(\omega)$ has at least $\lceil \mathfrak{s} \rceil$ eigenvalues $\lambda_1 < \dotsc < \lambda_{\lceil \mathfrak{s} \rceil}$ in~$(-2\omega,m-\omega)$.
These eigenvalues have an expansion, as $\kappa \to 0^+$, of the form
    \begin{equation}\label{expansion_eigenvalues_Lmu_nonrelativistic_limit}
        \lambda_k = \frac{1 -p^2(\mathfrak{s}+1-k)^2}{2m} \kappa^2 + O\!\left(\kappa^3\right), \qquad k=1,\dotsc,\lceil \mathfrak{s} \rceil\,.
    \end{equation}
    Moreover, if there exists an $(\lceil \mathfrak{s} \rceil+1)$-th eigenvalue $\lambda_{\lceil \mathfrak{s} \rceil + 1}$ in~$(-2\omega,m-\omega)$, then it is positive and admits the lower bound
    \begin{equation}\label{Lower_bound_expansion_first_positive_ev_Lmu_nonrelativistic_limit}
        \lambda_{\lceil \mathfrak{s} \rceil + 1} \geq m - \omega - O\!\left(\kappa^3\right).
    \end{equation}
\end{theorem}

Notice that the bound on~$\lambda_1$ obtained in~\eqref{expansion_eigenvalues_Lmu_nonrelativistic_limit} is better in the non-relativistic limit than the bound $\lambda_1 \geq -\mu\normt{Q} = -\mu p(p+1)(m-\omega)$ obtained in Lemma~\ref{A_priori_bounds_on_first_ev_of_Lmu}.
    
Also, since $p\mapsto \mathfrak{s}(p,\mu)$ is strictly decreasing on~$(0,+\infty)$, for any $\mu\geq0$, we have
\[
    \mathfrak{s}(p) > \lim\limits_{q\to+\infty} \mathfrak{s}(q) = \frac{\sqrt{1+4\mu} - 1}{2} \geq 0\,,
\]
and $\lceil \mathfrak{s} \rceil\geq1$ for any $p>0$ and~$\mu\geq0$. In the particular case $\mu=2$, it gives $\mathfrak{s}>1$, hence $\lceil \mathfrak{s} \rceil\geq2$, which means that $L_2$ has at least $3$ eigenvalues for any $p>0$. Furthermore, $\mathfrak{s}(p,\mu) = \frac{1}{p}+\frac{\mu}{2} + O(p)$ when $p\to0^+$, hence the number of eigenvalues diverges when $p\to0^+$, independently of~$\mu$.

\begin{corollary}\label{negative_evs_of_L2}
    Let $f(s) = s |s|^{p-1}$ with~$p>0$. In the non-relativistic limit, the operator $L_2(\omega)$ has exactly three eigenvalues $-2\omega = \lambda_0 < \lambda_1 < \lambda_2=0$ in~$[-2\omega,0]$, and the second eigenvalue $\lambda_1$ satisfies
    \[
       \lambda_1 = -\frac{p(p+2)}{2m} \kappa^2 + O\!\left(\kappa^3\right).
    \] 
    Moreover,
    \begin{itemize}
        \item if $p<1$, then $\lambda_3$ exists, is positive and satisfies
        \[
            \lambda_3=  \frac{p(2-p)}{2m} \kappa^2 +  O\!\left(\kappa^3\right);
        \]
        \item if $p\geq1$ and if $\lambda_3$ exists, then it is positive and admits the lower bound
        \[
            \lambda_3 \geq m - \omega - O\!\left(\kappa^3\right).
        \]
    \end{itemize}
\end{corollary}
This shows, in particular, that the first positive eigenvalue, $\lambda_3$, is asymptotically close to the essential spectrum for $p\geq1$.
In the case $p<1$, $\lambda_3$ is positive but away from the essential spectrum and, actually, the smaller the power $p$, the more  eigenvalues lie in~$(0,m-\omega)$. 
\begin{proof}[Proof of Corollary~\ref{negative_evs_of_L2}]
    Remarking that, at $\mu=2$, $\mathfrak{s}=\mathfrak{s}(\mu=2,p)=\frac{p+1}{p}$ hence $\lceil \mathfrak{s} \rceil\geq2$ with $\lceil \mathfrak{s} \rceil\geq 3 \Leftrightarrow p<1$, the bounds are a direct transcription of those in Theorem~\ref{negative_spectrum_of_Lmu_nonrelativistic_limit}.
    
    Now, we have $\lambda_0(\omega) = -2\omega$, 
    \[
        \left|\lambda_1(\omega) + \frac{p(p+2)}{2m} \kappa^2\right| \leq O\!\left(\kappa^3\right) \qquad \textrm{ and } \qquad \left|\lambda_2(\omega)\right| \leq O\!\left(\kappa^3\right),
    \]
    and the next eigenvalue $\lambda_3(\omega)$, if it exists, is positive. Therefore, the eigenvalue $0$ can only be $\lambda_2(\omega)$ and as a consequence, $-2\omega=\lambda_0(\omega) < \lambda_1(\omega)<\lambda_2(\omega)=0$ are the only eigenvalues in~$[-2\omega,0]$, in the non-relativistic limit.
\end{proof}

With that result, we can now prove Theorem~\ref{Thm_L2_only_one_ev}, which relies on the continuity of the eigenvalues of~$L_2$ with respect to $\omega$.
\begin{proof}[Proof of Theorem~\ref{Thm_L2_only_one_ev}]
The eigenvalues of~$L_2(p, \omega)$ are simple, and continuous with respect to $\omega$. We know that $-2\omega$ and~$0$ are eigenvalues of~$L_2(p, \omega)$ for any $\omega \in (0,m)$. Therefore, the number of eigenvalues in $(-2\omega, 0)$ is independent of $\omega \in (0,m)$. The previous corollary shows that this number equals one for large $\omega$.
\end{proof}

\subsection{The minmax principle for operators with a gap.}
The key tool for the proof of Theorem~\ref{negative_spectrum_of_Lmu_nonrelativistic_limit} is a minmax principle for eigenvalues inside a gap in the essential spectrum of an operator, as shown first in~\cite{DolEstSer-00,GriSei-99} (see~\cite{DolEstSer-06,EstLewSer-19,SchSolTok-19} for related results). This theorem gives a variational characterization of eigenvalues of self-adjoint operators inside a gap in the essential spectrum. We use the following formulation of the principle.
\begin{theorem}[{\cite[Theorem 1]{DolEstSer-06}}]\label{minmax_with_gap}
    Let $A$ be a self-adjoint operator wit domain~$D(A)$, in a Hilbert space $\cH$. Suppose that $\Lambda_{\pm}$ are orthogonal projections on~$\cH$ with $\Lambda_+ + \Lambda_- = \Id_\cH$ and such that 
    \[
        F_\pm:=\Lambda_\pm D(A) \subset D(A)\,.
    \]
    Define $\gamma_0$, the lower limit of the gap, as
      \begin{equation}\label{def_0_minmax_level}
        \gamma_0:=\sup\limits_{x_-\in F_-\setminus\{0\}} \frac{\pscal{x_-,A x_-}_\cH}{\norm{x_-}_\cH^2}\,,
    \end{equation}
    and~$\gamma_\infty$, its upper limit, as
    \[
        \gamma_\infty:=\inf(\sigma_{\textrm{ess}}(A)\cap(\gamma_0,+\infty))\in[\gamma_0,+\infty]\,.
    \]
    Finally, for $k \in \N\setminus\{0\}$, the minmax levels are defined as 
     \begin{equation}\label{def_minmax_levels}
		    \gamma_k := \inf\limits_{\substack{V\subset F_+\\\dim V = k}}\sup\limits_{x\in (V\oplus F_-)\setminus\{0\}}\frac{\pscal{x,Ax}_\cH}{\norm{x}_\cH^2}\,.
	\end{equation}
	If $\gamma_0 <+ \infty $ and the \emph{gap condition}
	\begin{equation}\label{gap_condition_in_thm}
		\gamma_0 < \gamma_1
	\end{equation}
	is satisfied, then for any $k\geq1$ either $\gamma_k$ is the $k$-th eigenvalue of~$A$ in~$(\gamma_0, \gamma_\infty)$, counted with multiplicity, or $\gamma_k=\gamma_\infty$. In particular, $\gamma_\infty \geq \sup_{k\geq1}\gamma_k \geq \gamma_1$.
\end{theorem}

We will apply this theorem to $A=L_\mu$, so $D(A)=H^1(\R,\C^2)$.
If the projectors $\Lambda_{\pm}$ are chosen as the spectral projectors associated to $(- \infty, -m-\omega]$ and its complement, we obtain~$\lambda_0 = -m-\omega$, $\lambda_\infty= m-\omega$. All hypotheses are automatically satisfied, so the minmax formula gives exactly the eigenvalues in the gap in the essential spectrum. However, the characterization is not very useful since these projectors are not known explicitely. The strength of the above theorem is that it still gives useful information for well-chosen explicit projections, adapted to the non-relativistic limit.

While it would be possible to prove Theorem~\ref{negative_spectrum_of_Lmu_nonrelativistic_limit} by using the typical projectors on upper and lower spinor components, the projections $\Lambda_\pm$ allow for a more streamlined presentation. Indeed, careful estimates show that the contribution from $F_-$ to the eigenvalues is of order $\kappa^6$ only.

We apply the minmax principle of Theorem~\ref{minmax_with_gap} with $A= L_\mu$, $D(A)= H^1(\R, \C^2)$ and, with $\alpha =(2 m )^{-1}$, define the subspaces
\begin{equation}\label{eq:def_H_pm}
    \cH_+:=\left\{ \left(h,-\alpha h'\right)\transp \, \left| \, h\in H^1(\R) \right. \right\} \quad \textrm{ and } \quad \cH_-:=\left\{ \left(-\alpha h', h\right)\transp \, \left| \, h\in H^1(\R) \right. \right\}.
\end{equation}
The orthogonal projectors on these subspaces are given by the pseudo-differential operators
\[
    \Lambda_{+} = \cF^* \frac{1}{1+\alpha^2 \xi^2}\begin{pmatrix} 1 & -i\alpha \xi \\
    i\alpha \xi & \alpha^2 \xi^2\end{pmatrix} \cF \quad \textrm{ and } \quad  \Lambda_- = \sigma_1 \Lambda_+ \sigma_1\,,
\]
where $\cF$ denotes the Fourier transform and~$\xi$ the variable in Fourier space.
It is clear from this expression that $\Lambda_+ + \Lambda_- = \Id_{\cH}$ and that $F_{\pm} := \Lambda_{\pm}H^1(\R,\C^2) \subset H^1(\R,\C^2)$. 
These projections can be obtained by considering the spectral projectors on the positive and negative spectrum for the free Dirac operator $D_m$ and keeping the first terms in a formal expansion for small $\xi/m$.

As a final definition that will allow for a concise notation, we define for $h \in H^1(\R,\C)$,
\begin{equation}\label{Def_l+_l-}
    l_+(h) := \left(h,-\alpha h'\right)\transp \quad \textrm{ and } \quad l_-(h) := \left(-\alpha h', h \right)\transp.
\end{equation}
With this definition,
\[
    F_{\pm} = \left\{l_\pm (h) \, \left| \, h \in H^2(\R,\C) \right. \right\}.
\]

\begin{proof}[Proof of Theorem~\ref{negative_spectrum_of_Lmu_nonrelativistic_limit}]
    We compute the quantities appearing in the Rayleigh quotient.
    First of all, recalling that $\alpha=(2m)^{-1}$, we have
    \[
        \pscal{l_+(h), D_m l_+(h)} = \pscal{\begin{pmatrix}h\\-\alpha h'\end{pmatrix}, \begin{pmatrix}m h-\alpha h''\\ - h' + \alpha m h'\end{pmatrix}} = m \norm{l_+(h)}^2 + \alpha \norm{h'}^2,
    \]
    where we have used integration by parts to simplify the expressions. This identity motivates the definition of~$\Lambda_\pm$.
    Introducing the terms of $L_\mu$ involving $v$ and~$u$, we are left with
    \begin{align}
        &\pscal{l_+(h), L_\mu l_+(h)} \label{eq:expansion_+_quotient}\\
        &\begin{multlined}[t]
            =(m-\omega) \norm{l_+(h)}^2 + \alpha \norm{h'}^2   -\int{(v^2 -u^2)^p}\left(\abs{h}^2 - \alpha^2\abs{h'}^2\right) \nonumber\\
            - p \mu \left( \pscal{h, (v^2 -u^2)^{p-1} v^2 h} + 2 \Re \pscal{\alpha h' , (v^2 -u^2)^{p-1} uv h} - \alpha^2 \pscal{h', (v^2 -u^2)^{p-1} u^2 h'} \right)
      \end{multlined} \nonumber\\
      &
        \begin{multlined}[t]
            =(m-\omega) \norm{l_+(h)}^2 + \alpha \norm{h'}^2   -\int(v^2 -u^2)^p \left\{(1+p \mu )\abs{h}^2 - \alpha^2\abs{h'}^2\right\} \\
            - p \mu \left( \pscal{h, (v^2 -u^2)^{p-1} u^2 h} + 2 \Re \pscal{\alpha h' , (v^2 -u^2)^{p-1} uv h} + \alpha^2 \pscal{h', (v^2 -u^2)^{p-1} u^2 h'} \right)
      \end{multlined} \nonumber
    \end{align}
    Here, we rearranged the terms because $u \ll v$ in the non-relativistic limit, and the terms on the second line will contribute only to the error term.
    
    Similarly, by using $l_-(g) = \sigma_1 l_+(g)$, we find
    \begin{multline*}
        \pscal{l_-(g), L_\mu l_-(g)} = -(m + \omega)\norm{l_-(g)}^2 - \alpha \norm{g'}^2 \\
        + \int (v^2 -u^2)^p \left\{\norm{g}^2 - \alpha^2\norm{g'}^2\right\} -\mu \pscal{l_-(g), Ql_-(g)},
    \end{multline*}
    where we will not need a detailed expansion of the last term.
       
    Finally, for the cross terms, we compute
    \[
        \pscal{l_-(g), D_m  l_+(h)} = \pscal{\begin{pmatrix}-\alpha g' \\ g \end{pmatrix}, \begin{pmatrix}m h-\alpha h''\\ - h' + h'/2 \end{pmatrix}} = \alpha^2 \pscal{g', h''},
    \]
    and obtain
    \[
        \abs{  \pscal{l_-(g), L_\mu  l_+(h)}  - \alpha^2 \pscal{g', h''}} \leq \left(\norm{(v^2 - u^2)^p}_\infty + \mu \normt{Q} \right)\norm{l_+(h)}\norm{l_-(g)}.
    \]

    We first estimate $\gamma_0$. Since $Q$ is nonnegative,
    \[
        \gamma_0 = \sup_{g \in H^2} \frac{\pscal{l_-(g), L_\mu  l_-(g)} }{\norm{l_-(g)}^2 } \leq \sup_{g \in H^2} \frac{\pscal{l_-(g), L_0  l_-(g)} }{\norm{l_-(g)}^2 } \leq - (m+\omega ) + \norm{\left(v^2 -u^2\right)^p}_\infty.
    \]
    By taking a sequence of trial functions $g_R = R^{-1/2}g(\frac{\cdot + 2 R}{ R})$ for a smooth and normalized~$g$, it can be shown that
    \[
        \gamma_0 \geq -(m+\omega) - \alpha \inf_{g \in H^2} \frac{\norm{g'}^2}{\norm{l_- g}^2} = -(m+\omega)\,,
    \]
    Combining with the bound in~\eqref{infinity_norm_v2minusu2tothep}, we have
    \begin{equation}\label{NonRelat_Bounds_gamma0}
       -(m+\omega) \leq \gamma_0 \leq   -(m+\omega) + (p+1)(m-\omega)\,.
    \end{equation}
    Therefore, from its definition, we have $\gamma_\infty = m-\omega$ for $\omega > \frac{p-1}{p+1} m$\,.

\medskip
\textbf{Lower bound.}
    Next, we obtain a lower bound for the minmax levels
    \[
        \gamma_k = \inf_{\substack{E \subset H^2 \\ \dim E = k}}\sup_{\substack{(g,h) \in H^2\times E\\ (g,h)\neq(0,0)}} \!\! \frac{\pscal{l_+(h) + l_-(g), L_\mu (l_+(h) + l_-(g))}}{\norm{l_+(h)}^2 + \norm{l_-(g)}^2} \geq \inf_{\substack{E \subset H^2 \\ \dim E = k}} \sup_{\substack{h \in E\\h\neq0}} \frac{\pscal{l_+(h), L_\mu l_+(h)}}{\norm{l_+(h)}^2 } \,.
    \]
    From~\eqref{eq:expansion_+_quotient}, we find
    \begin{multline*}
        \pscal{l_+(h), L_\mu l_+(h)} \geq (m-\omega) \norm{l_+(h)}^2 + \alpha \norm{h'}^2 - (1+p\mu)\int{(v^2 -u^2)^p}{\abs{h}^2 } \\
        - p \mu \left(\norm{(v^2 -u^2)^{p-1}u^2}_\infty + \norm{(v^2 -u^2)^{p-1} u v}_\infty\right) \norm{l_+(h)}^2.
    \end{multline*}

    We finally use the definition of~$\mathfrak{s}$ in~\eqref{Def_s} to obtain
    \[
        \pscal{l_+(h), L_\mu l_+(h)} \geq (m-\omega) \norm{l_+(h)}^2 + \frac{1}{2m} \norm{h'}^2  - \frac{p^2\kappa^2}{m}  \pscal{h, \frac{\mathfrak{s}(\mathfrak{s}+1)}{2 \cosh( p\kappa \cdot )} h} - R_1 \norm{l_+(h)}^2,
    \]
    where we used $\alpha=1/(2m)$ and have defined
    \begin{multline*}
        R_1 := p\mu\left(\norm{(\tilde{v}^2-\tilde{u}^2)^{p-1}\tilde{u}\tilde{v}}_\infty + \norm{(\tilde{v}^2-\tilde{u}^2)^{p-1}\tilde{u}^2}_\infty\right) \\
        + (1+p\mu)\norm{(\tilde{v}^2-\tilde{u}^2)^p - \frac{(p+1)\kappa^2}{2m\cosh^2}}_\infty = O\!\left(\kappa^3\right)
    \end{multline*}
    in view of the bounds~\eqref{Nonrelativistic_expansion_kappa3_term}--\eqref{Nonrelativistic_expansion_reminder_type_1} in Appendix~\ref{Appendix_nonrelativistic_expansions}.
    
    In order to obtain a precise bound here, we write everything in terms of the rescaled variable $\tilde x := p \kappa x$ and obtain
    \begin{equation} \label{eq:lower_lambda_k-1} 
        \gamma_k  \geq m -\omega + \frac{p^2 \kappa^2}{m} \inf_{\substack{E \subset H^2 \\ \dim E = k}} \sup_{\substack{h \in E\\h\neq0}} \frac{\frac{1}{2} \norm{h'}^2 - \frac{\mathfrak{s}(\mathfrak{s}+1)}{2 } \pscal{h, \cosh^{-2} h} }{\norm{h}^2 + \alpha^2 p^2 \kappa^2 \norm{h'}^2 } - R_1\,.
    \end{equation}

    At this point, we recognize the usual\footnote{Up to the additional term in the denominator.} minmax formula for the eigenvalues of the Schr{\"o}\-dinger operator with the P\"oschl--Teller potential
    \[
        V_{\textrm{PT}}(\mathfrak{s},x) = -\frac{\mathfrak{s}(\mathfrak{s}+1)}{2\cosh^2x}\,.
    \]
    It is known that the spectrum of~$-\frac12 f'' + V_{\textrm{PT}}(\mathfrak{s},x) f$ is the union of its essential spectrum $[0,+\infty)$ and exactly $\lceil \mathfrak{s} \rceil$ eigenvalues, which are
    \begin{equation}\label{NRJ_PoschlTeller}
        E_j := -\frac{(\mathfrak{s}+1-j)^2}2, \qquad j=1,\dotsc,\lceil \mathfrak{s} \rceil\,,
    \end{equation}
    with the Legendre functions\footnote{Also called \emph{Ferrers functions}, since $-1<\tanh<1$.} $\left\{P_{\mathfrak{s}}^{\mathfrak{s}+1-j}(\tanh)\right\}_{j=1,\dotsc,\lceil \mathfrak{s} \rceil}$ as corresponding eigenfunctions.
    
    We define, for any integer $k\in \llbracket 1,\dotsc,\lceil \mathfrak{s} \rceil \rrbracket$, the $k$-dimensional space
    \begin{equation}\label{Def_V_k}
        V_k := \vect\left\{P_{\mathfrak{s}}^{\mathfrak{s}+1-j}(\tanh)\right\}_{1\leq j \leq k}.
    \end{equation}
    That is, the eigenspace corresponding to the first $k$ eigenvalues of this operator.
    
    Returning to~\eqref{eq:lower_lambda_k-1}, we find that for $k \leq \lceil \mathfrak{s} \rceil$, the infimum is negative (use $V_k$ as trial subspace), and therefore we may bound
    \begin{multline*}
        0 > \inf_{\substack{E \subset H^2 \\ \dim E = k}} \sup_{\substack{h \in E\\h\neq0}} \frac{\pscal{h,-\frac{1}{2}h'' - \frac{\mathfrak{s}(\mathfrak{s}+1)}{2\cosh^2} h}}{\norm{h}_2^2 + \frac{p^2\kappa^2}{4m^2}\norm{h'}_2^2} \geq \inf_{\substack{E \subset H^2 \\ \dim E = k}} \sup_{\substack{h \in E\\h\neq0}} \frac{ \pscal{h, -\frac12 h'' - \frac{\mathfrak{s}(\mathfrak{s}+1)}{2\cosh^2}h}}{\norm{h}^2} \\
        = \sup\limits_{h\in V_k \setminus\{0\}} \frac{ \pscal{h, -\frac12 h'' - \frac{\mathfrak{s}(\mathfrak{s}+1)}{2\cosh^2}h}}{\norm{h}^2} = E_k\,.
    \end{multline*}
    Inserting~\eqref{NRJ_PoschlTeller}, we obtain
    \begin{equation}\label{eq:lower_gamma_k}
        \gamma_k  \geq  m -\omega - \frac{p^2 \kappa^2}{2 m}(\mathfrak{s}+1-k)^2 - R_1\,.
    \end{equation}
    
    In the case $k = 1$, we obtain 
    \begin{equation} \label{eq:lower_gamma_1}
        \gamma_1 \geq m - \omega -\frac{p^2 \mathfrak{s}^2 \kappa^2}{2 m} - R_1
    \end{equation}
    and therefore, for sufficiently large values of $\omega$, the gap condition $\gamma_1 > \gamma_0$ is satisfied in view of~\eqref{NonRelat_Bounds_gamma0}.
    We conclude from Theorem~\ref{minmax_with_gap} that the minmax levels $\gamma_k$ are the eigenvalues of $L_\mu$ in $(\gamma_0, m -\omega)$.
    
    We now show that indeed, $\gamma_k = \lambda_k$, i.e., the minmax levels $\gamma_k$ are \emph{all} the eigenvalues of~$L_\mu$ in~$(-2\omega,m-\omega)$. First, for sufficiently large $\omega$, the lower bound~\eqref{eq:lower_gamma_1} gives $\gamma_1 > -2\omega$, so we conclude that the eigenvalue $-2\omega$ lies below~$\gamma_0$.
    Second, by Lemmas~\ref{A_priori_bounds_on_first_ev_of_Lmu} and~\ref{Lemma_Q} (for $\omega > m/2$), there are no eigenvalues in
    \[
        (-2\omega, -\mu \normt{Q} ) =(-2\omega, -\mu p(p+1)(m-\omega))\,.
    \]
    Since $\gamma_0 < -\mu p(p+1)(m-\omega)$ for sufficiently large $\omega$, $\{\gamma_k\}_{k\geq1}$ are \emph{all} the eigenvalues of~$L_\mu$ in~$(-2\omega,m-\omega)$.
    
   Returning to the lower bound~\eqref{eq:lower_gamma_k}, in combination with the expansion
    \begin{equation}\label{Nonrelativistic_expansion_m_minus_omega}
        m-\omega = \frac{\kappa^2}{2m} + O\!\left(\kappa^4\right) ,
    \end{equation}
    we obtain
    the required lower bound for $k \leq \lceil \mathfrak{s} \rceil $.
    
    For $k > \lceil \mathfrak{s} \rceil $, any subspace of dimension~$k$ contains a function in the positive eigenspace of the P\"{o}schl--Teller Schr{\"o}\-dinger operator, hence
    \[
        \gamma_{\lceil \mathfrak{s} \rceil + 1}  \geq m -\omega - R_1\,. 
    \]

\medskip
\textbf{Upper bound for $k \leq \lceil \mathfrak{s} \rceil $.}
    For an upper bound, we can not get rid of the supremum over $g$.
    However, we can choose the subspace $E_k$ in the infimum in order to match the upper bound.
    So, we restrict to $F_k := \{h(p\kappa x) | h \in V_k\}$, where $V_k$, defined in~$\eqref{Def_V_k}$ is the span of the first $k$ eigenfunctions of the P\"{o}schl--Teller Schr{\"o}\-dinger operator. 
    We will not use their explicit expression (yet), but
    use bounds on derivatives of $h \in {F_k}$
    \[
        \norm{h'} \leq p\kappa \sqrt{\mathfrak{s}(\mathfrak{s}+1)} \norm{h} \quad \textrm{ and } \quad \norm{h''} \leq p^2\kappa^2 \mathfrak{s}(\mathfrak{s}+1) \norm{h},
    \]
    see Appendix~\ref{Appendix_Bound_derivatives_Legendre} for details.
    For $h \in V_k$, we bound
    \begin{align*}
        \abs{\pscal{l_-(g), L_\mu  l_+(h)} } &= \abs{   \alpha^2 \pscal{g', h''}+ \pscal{l_-(g), (v^2 - u^2)^p \sigma_3 l_+(g)} -\mu  \pscal{l_-(g), Q  l_+(g)}  } \\
        &\leq p^2\kappa^2 \mathfrak{s} (\mathfrak{s}+1) \frac{3m+\omega}{2 m(m+\omega)} \norm{l_-(g)}\norm{l_+(h)}:= R_2 \norm{l_-(g)}\norm{l_+(h)},
    \end{align*}
    where we have used~\eqref{infinity_norm_v2minusu2tothep} and Lemma~\ref{Lemma_Q} (for $\omega > m/2$), so $R_2 = O(\kappa^2)$.
    Therefore,
    {\allowdisplaybreaks
    \begin{align*}
        \gamma_k &= \inf_{\substack{E \subset H^2 \\ \dim E = k}} \sup_{(g,h) \in \left(H^2\times E\right)\setminus\{0\}} \frac{\pscal{l_+(h) + l_-(g), L_\mu (l_+(h) + l_-(g))}}{\norm{l_+(h)}^2 + \norm{l_-(g)}^2} \\
        &\leq \sup_{(g,h) \in \left(H^2\times F_k\right)\setminus\{0\}} \frac{\pscal{l_+(h), L_\mu l_+(h)} + \pscal{l_-(g), L_\mu l_-(g))} + 2\Re   \pscal{l_-(g), L_\mu  l_+(h)} }{\norm{l_+(h)}^2 + \norm{l_-(g)}^2} \\
        & \leq \sup_{(g,h) \in \left(H^2\times F_k\right)\setminus\{0\}} \frac{\pscal{l_+(h), L_\mu l_+(h)} + \gamma_0 \norm{l_-(g)}^2 + 2 R_2 \norm{l_-(g)}\norm{l_+(h)} }{\norm{l_+(h)}^2 + \norm{l_-(g)}^2} \\
        &\leq \sup_{(g,h) \in \left(H^2\times F_k\right)\setminus\{0\}} \frac{\pscal{l_+(h), L_\mu l_+(h)} + \kappa^3 \norm{l_+(h)}^2 + (\gamma_0 + \kappa^{-3}R_2^2 )\norm{l_-(g)}^2  }{\norm{l_+(h)}^2 + \norm{l_-(g)}^2}\,,
    \end{align*}
    }%
    where we used the definition of $\gamma_0$.
    For the maximization over $g$, we note that, with $q := \norm{l_-(g)}^2 / \norm{l_+(h)}^2$, the quotient is of the form 
    \[
        \frac{A + B q}{1 + q}\,.
    \]
    
    If $B < A$, the maximum is attained at $q=0$.
    Now, we note that 
    \[
        A:= \frac{\pscal{l_+(h), L_\mu l_+(h)}}{\norm{l_+(h)}^2} +\kappa^3 \geq  m - \omega -\frac{p^2 \mathfrak{s}^2 \kappa^2}{2 m} - R_1  
    \]
    by the estimates leading to~\eqref{eq:lower_gamma_1}.
    Moreover, $\kappa^{-3}R_2^2 \sim \kappa$ for sufficiently small $\kappa$, hence
    \[
        B:= \gamma_0 + \kappa^{-3}R_2^2  <  A \,.
    \]
    
    We finally insert~\eqref{eq:expansion_+_quotient} again, neglect the nonpositive terms, bound (out of the supremum) the terms of order $\kappa^3$ or higher, insert the P\"{o}schl--Teller potential
    \[
        (1 + \mu p)\pscal{h,\frac{(p+1)\kappa^2}{2m\cosh^2}h} = \frac{p^2\kappa^2}{m}\pscal{h,\frac{\mathfrak{s}(\mathfrak{s}+1)}{2\cosh^2}h}
    \]
    in the numerator, neglect ---see \eqref{Nonrelativistic_expansion_reminder_type_1} in Appendix~\ref{Appendix_nonrelativistic_expansions}--- the negative term
    \[
        -\left((\tilde{v}^2-\tilde{u}^2)^p - \frac{(p+1)\kappa^2}{2m\cosh^2}\right)<0\,,
    \]
    then finally rescale variables in the remaining supremum, and  are left with
    {\allowdisplaybreaks
    \begin{align*}
          \gamma_k 
        &\leq \sup_{ h \in F_k} \frac{\pscal{l_+(h), L_\mu l_+(h)} + \kappa^3 \norm{l_+(h)}^2 }{\norm{l_+(h)}^2} \\
        &\leq m -\omega  + \sup_{ h \in F_k} \frac{\alpha \norm{h'}^2 - (1 + \mu p)\int (v^2- u^2)^p \abs{h}^2    }{\left(\norm{h}^2 + \alpha^2 \norm{h'}^2\right)} + R_3 \\
         &\leq m -\omega  + \frac{p^2\kappa^2}{m}\sup_{ h \in V_k} \frac{ \pscal{h, -\frac12 h'' - \frac{\mathfrak{s}(\mathfrak{s}+1)}{2\cosh^2}h}}{\norm{h}^2 + \alpha^2 {p^2 \kappa^2} \norm{h'}^2} + R_3\,,
    \end{align*}
    }%
    where we have defined 
    \[
        R_3 :=\kappa^3 +  \norm{\left(v^2 -u^2\right)^p}_\infty \kappa^2 \mathfrak{s}(\mathfrak{s}+1) + \mu p \norm{(v^2- u^2)^{p-1} u v}_\infty =  O\!\left(\kappa^3\right).
    \]
    
    Since the numerator is negative for $h \in {V_k}$, we also use
    \[
        \frac{\norm{h}^2 + \alpha^2 p^2 \kappa^2 \norm{h'}^2}{ \norm{h}^2} \leq 1 + \alpha^2 p^2 \kappa^2 \mathfrak{s}(\mathfrak{s}+1 )
    \]
    to obtain finally
    \[
          \gamma_k \leq m -\omega - \frac{p^2\kappa^2}{2m}  \frac{(\mathfrak{s}+1-k)^2}{1 + \alpha^2 p^2 \kappa^2 \mathfrak{s}(\mathfrak{s}+1) } + R_3, \qquad \forall\, k\in \{ 1,\dotsc,\lceil \mathfrak{s} \rceil \}\,. 
    \]
    This concludes the proof of Theorem~\ref{negative_spectrum_of_Lmu_nonrelativistic_limit}.
\end{proof}

\appendix
    \phantomsection 
    \addcontentsline{toc}{section}{Appendices}
    \addtocontents{toc}{\protect\setcounter{tocdepth}{0}} 
        
\section{ODE arguments}\label{Appendix_ODE}
In this appendix, we give the proofs of Propositions~\ref{prop:basic_groundstate_properties} (basic properties of the groundstates) and~\ref{prop_differentiability} because some intermediate steps are useful in both proofs. However, we stress that the proof of Proposition~\ref{prop_differentiability} requires spectral properties of $L_2$ that have been established in Proposition~\ref{Prop:eigenvalues_symmetry} and Lemma~\ref{L_mu_eigenvalues_simple_and_analytic_in_mu}.
\begin{proof}[Proof of Proposition~\ref{prop:basic_groundstate_properties}]
Equation~\eqref{eq:phi_0} satisfied by $\phi_0(\omega) = (v_\omega, u_\omega)\transp$ can be written as
\begin{equation} \label{eq:groundstate-system}
	\begin{cases}
		v_\omega'= - (M_\omega+\omega) u_\omega \\
		u_\omega'= -(M_\omega-\omega) v_\omega\,,
	\end{cases}
\end{equation}
with $M_\omega = m - f\!\left(v_\omega^2 - u_\omega^2\right)$ as defined in~\eqref{Def_M}.
Therefore, we assume $v_\omega(0)\geq0$ since $-\phi_0(\omega)$ is a solution if and only if $\phi_0(\omega)$ is a solution.
We define $F(s) := \int_0^s f(t) \di t $ on~$(0,+\infty)$ and $\tilde{F}\in\mathcal{C}^0([0,+\infty))\cap \mathcal{C}^1((0,+\infty))$ by $\tilde{F}(0)=0$ and $\tilde{F}(s) := F(s)/s$ on~$(0,+\infty)$. The continuity at the origin derives from $f'>0$ on~$(0,+\infty)$ by Assumption~\ref{Assumption_general_nonlinearity}, as it gives
\begin{equation}\label{nonlinearity_compared_to_its_primitive}
	0<F(s) = \int_0^s f(t) \di t < s f(s) \quad \textrm{ on } (0,+\infty)\,,
\end{equation}
where we used the assumption $f(0)=0$ for the positivity of $F$.

\medskip
\noindent \emph{Existence, positivity and regularity.}
The existence is established in~\cite[Lemma 3.2]{BerCom-12}. In its notation, we have $g(s) = m - f(s)$ and $G(s) = m s - F(s)$, and have to check that
\[
	\exists\, s^*_\omega>0\,, \, \forall\, s \in (0, s^*_\omega)\,, \, \frac{G(s)}{s} > \omega = \frac{G(s^*_\omega)}{s^*_\omega} \neq g(s^*_\omega)\,.
\]

First, $G(s)/s>g(s)$ on~$(0,+\infty)$ by~\eqref{nonlinearity_compared_to_its_primitive} ---hence checking already, for any $s^*_\omega>0$, the non-equality---, which implies $\lim_{s\searrow0} G(s)/s \geq g(0) = m > \omega$.
Second, denoting $f^{-1}$ the inverse of $f:[0,+\infty)\to[0,\lim_{+\infty} f)$, $\lim_\infty f \geq m$ from Assumption~\ref{Assumption_general_nonlinearity} gives
\[
	\forall\, \epsilon>0\,, \, \forall\, s > f^{-1}(m-\epsilon)\,, \quad \tilde{F}(s) \geq \frac{1}{s} \int_{f^{-1}(m-\epsilon)}^s f(t) \di t > \frac{s-f^{-1}(m-\epsilon)}{s} (m-\epsilon)\,.
\]
Thus $\lim_{+\infty} \tilde{F} \geq m$, and $\lim_{s\to+\infty} G(s)/s \leq 0 < \omega$.
Third, $s^2 \tilde{F}'(s)= s f(s) - F(s)>0$ on~$(0,+\infty)$ by~\eqref{nonlinearity_compared_to_its_primitive}. Summarizing, $G(s)/s = m - \tilde{F}(s)$ is strictly decreasing on~$(0,+\infty)$ with $\omega$ is in its image set, hence there is a unique $s^*_\omega>0$ satisfying the properties. The existence of a non-zero solution $(v_\omega, u_\omega) \in H^1(\R,\R^2)$ to~\eqref{eq:phi_0} with $v_\omega$ is even and $u_\omega$ odd is therefore proved by~\cite[Lemma 3.2]{BerCom-12}. Moreover, since $u_\omega(0)=0$ and $v_\omega(0)\geq0$, we can assume that $v_\omega(0)>0$ as, otherwise, $\phi_0(\omega)$ is the trivial solution by the Cauchy theory.

To prove~\emph{\ref{it:positivity}}, notice that~\eqref{eq:groundstate-system} is the Hamiltonian system $h_\omega(u_\omega,v_\omega)=0$ associated to 
\[
    h_\omega(u,v) := \frac{1}{2} \left( \omega \left(v^2 + u^2\right) - m \left(v^2 - u^2\right) + F\circ \left(v^2 - u^2\right) \right).
\]
On one hand, since $f>\tilde{F}$ on~$(0,+\infty)$ by~\eqref{nonlinearity_compared_to_its_primitive} and $v_\omega(0)>0$,~\eqref{eq:groundstate-system} implies
\[
	u_\omega'(0) = (\omega - m + f\circ v_\omega^2(0) ) v_\omega(0) > \left(\omega - m + \tilde{F}\circ v_\omega^2(0)\right)v_\omega(0) = (m - \omega)v_\omega(0) > 0\,.
\]
On another hand, it also gives $u_\omega(x) \neq 0$ if $x \neq 0$ since, otherwise, there would be $x >0$ (by oddity of $u$) such that $(v_\omega(x), u_\omega(x))= (v_\omega(0), u_\omega(0))$ and the solution would be periodic (in $x$), contradicting $\phi_0(\omega) \in L^2(\R)$. Hence, we conclude that $u_\omega>0$ on~$(0,+\infty)$, since $u_\omega$ is continuous, non-zero on~$(0,+\infty)$ with $u_\omega'(0)>0$ and $u_\omega(0)=0$.

Moreover, $d_\omega:=v_\omega^2 - u_\omega^2>0$ on~$\R$. Indeed, $d_\omega(0)=v_\omega(0)^2>0$ and, if $d_\omega(x)=0$ for some $x\neq0$, then $0=h_\omega(u_\omega,v_\omega)(x) = \frac{\omega}{2}  (v_\omega(x)^2 + u_\omega(x)^2)$ hence $v_\omega(x)= u_\omega(x) = 0$ contradicting $u_\omega\neq0$ on~$\R\setminus\{0\}$ just obtained. Finally, $v_\omega > u_\omega>0$ on~$(0,+\infty)$ because $|v_\omega| > u_\omega>0$ on~$(0,+\infty)$, since~$d_\omega>0$ on~$\R$, with $v_\omega$ continuous and $v_\omega(0)>0$.

Since $(v_\omega, u_\omega)\in H^1(\R)\subset\mathcal{C}^0(\R)$ and $f\in \mathcal{C}^0(\R)$ by Assumption~\ref{Assumption_general_nonlinearity},~\eqref{eq:groundstate-system} gives $(v_\omega, u_\omega)\in \mathcal{C}^1(\R)$. Moreover, $f\circ d_\omega\in \mathcal{C}^1(\R)$ since $d_\omega>0$ on $\R$ and $f\in \mathcal{C}^1(\R\setminus\{0\})$ by Assumption~\ref{Assumption_general_nonlinearity}, hence~\eqref{eq:groundstate-system} actually gives $(v_\omega, u_\omega)\in \mathcal{C}^2(\R)$.

\medskip
\noindent \emph{Uniqueness and continuity in $\omega$ of the initial condition.}
We have
\[
	h_\omega(v_\omega(0),0)=0 \quad \Leftrightarrow \quad \tilde{F}(v_\omega^2(0)) = m - \omega\,.
\]
Moreover, we already prove that $\tilde{F}\in\mathcal{C}^0([0,+\infty))\cap \mathcal{C}^1((0,+\infty))$ is strictly increasing on~$[0,+\infty)$ with $\lim_{+\infty} \tilde{F} \geq m$. Thus, it is bijective from~$[0,+\infty)$ to~$[0,\lim_{+\infty} \tilde{F})$ and has a continuous inverse with $v_\omega(0)^2 = \tilde{F}^{-1}(m-\omega)>0$. By the Cauchy theory, this implies the uniqueness of $\phi_0(\omega)$.

\medskip
\noindent \emph{Pointwise decay of $\phi_0(\omega,\cdot)$ as $\omega \to m$.} 
We start by noticing that $h_\omega(u_\omega,v_\omega) = 0$ gives
\begin{equation} \label{eq:comparable-components-apendix}
	0< \omega (v_\omega^2 + u_\omega^2) < m  (v_\omega^2 - u_\omega^2) = m d_\omega \quad \textrm{ on } \R\,,
\end{equation}
due to $F >0$ on~$(0,+\infty)$, $v_\omega^2(0)> 0 = u_\omega^2(0)$, and the parities of $v_\omega$ and $u_\omega$.
Moreover, $\norm{d_\omega}_{L^\infty} = d_\omega(0)$, because $(d_\omega)' = - 4 \omega u_\omega v_\omega$ by~\eqref{eq:groundstate-system} and $v_\omega >u_\omega >0$ on~$(0,+\infty)$, hence
\[
	\frac{\omega}{m} \norm{\phi_0(\omega, \cdot)}^2_{L^\infty} = \frac{\omega}{m} \norm{v_\omega}^2_{L^\infty} \leq \frac{\omega}{m} \norm{v_\omega^2+u_\omega^2}_{L^\infty} \leq \norm{d_\omega}_{L^\infty} = d_\omega(0) = \tilde{F}^{-1}(m-\omega)\,,
\]
for $\omega\in(0,m)$. Since $\tilde{F}^{-1}(0) = 0$, this establishes half of~\emph{\ref{it:decay_in_Nrel_limit}}.

\medskip
\noindent \emph{Uniform exponential decay of $\phi_0(\omega,\cdot)$.} 
In order to prove differentiability, it is important to obtain a pointwise upper bound uniform in $\omega$, for $\omega$ bounded away from $m$ and $0$. Thus, for fixed $\epsilon > 0$,
we assume that $m^2-\omega^2 \geq \epsilon^2 > 0$. In view of~\eqref{eq:comparable-components-apendix}, it is sufficient to bound $d_\omega$, and by symmetry we assume that $x \geq 0$.
We have
\[
  (d_\omega)' = - 4 \omega u_\omega v_\omega = -2\omega \left((v_\omega^2 + u_\omega^2)^2 - d_\omega^2  \right)^{1/2} = - 2\left(\left( m d_\omega - F(d_\omega) \right)^ 2  - \omega^2 d_\omega^2  \right)^{1/2},
\]
using $h(v_\omega, u_\omega) = 0$ for the last equality.
Since $\lim_{s\searrow0} (m - \tilde{F}(s)) \geq m > \omega$ (see the proof of existence and remember that $m - \tilde{F}(s)=G(s)/s$), we define $s_\epsilon$ such that
\[
	\left(  \left( m s - F(s) \right)^ 2  - \omega^2 s^2  \right)^{1/2} = s \left(  \left( m  - \tilde{F}(s) \right)^ 2  - \omega^2   \right)^{1/2} \geq \frac{\epsilon}{2} s\,, \quad \forall\, s \in [0,s_\epsilon] \,.
\]
Since $d_\omega$ is strictly decreasing on $[0,+\infty)$ hence bijective on it, we define $x^*(\omega,\epsilon)$ by $x^*(\omega,\epsilon) = 0$ if $d_\omega(0) \leq s_\epsilon$ and $x^*(\omega,\epsilon) = d_\omega^{-1}(s_\epsilon)$ otherwise, and we have
\[
	(d_\omega)' (x)\leq  - \epsilon d_\omega(x)\,, \quad \forall\, x \geq x^*(\omega,\epsilon)\,.
\]
Integrating it, yields~\emph{\ref{it:exponential_decay}} since it gives
\[
    d_\omega(x) \leq d_\omega(x^*(\omega,\epsilon)) \exp{\left( - \epsilon (x - x^*(\omega,\epsilon)) \right)} \leq s_\epsilon \exp{\left( - \epsilon (x - x^*(\omega,\epsilon)) \right)}\,.
\]

 \medskip
\noindent \emph{Decays of $Q$.}
By the definition of $Q$, the positivity of $d_\omega$, the one of $f'$ on $(0,+\infty)$ from Assumption~\ref{Assumption_general_nonlinearity}, and~\eqref{eq:comparable-components-apendix}, we bound
\[
    \norm{Q_\omega(x)}_{\C^2 \mapsto \C^2} = (v_\omega^2(x) + u_\omega^2(x)) f'(d_\omega(x)) \leq  \frac{m}{\omega} d_\omega(x) f'(d_\omega(x))\,.
\]
It yields~\emph{\ref{it:decay_of_Q}}, the exponential decay in $x$, due to the exponential decay of $v_\omega$ and $u_\omega$, and to $\lim_{s\to0^+} s f'(s) = 0$ from Assumption~\ref{Assumption_general_nonlinearity}. 
It also gives the statement on $Q$ in~\emph{\ref{it:decay_in_Nrel_limit}}:
\[
   \lim_{\omega \to m} \sup_{x \in \R}\norm{Q_\omega(x)}_{\C^2 \mapsto \C^2} \leq  \lim_{\omega \to m} \frac{m}{\omega} \sup_{s \in (0, v_\omega^2(0))} s f'(s) = 0\,. \qedhere
\]
\end{proof}
\begin{proof}[Proof of Proposition~\ref{prop_differentiability}]
We freely use the notation introduced in the previous proof. We remind the reader that the following proof of differentiability of $\omega\mapsto\phi_0(\omega)$ is needed because we do not assume continuity of $f'$ at $0$.

\medskip
\noindent \emph{Uniform continuity in $\omega$.} We show that $\omega \mapsto \phi_0(\omega)$ is continuous with values in $\mathcal{C}^0(\R)$.  
    We rewrite the nonlinear equation~\eqref{eq:phi_0} as an initial value problem on~$[0, +\infty)$ of the form
    \[
        \partial_x \phi_0(\omega) = B(\omega, \phi_0(\omega)) \phi_0(\omega), \quad \phi_0(\omega, 0) = \left(\tilde{F}^{-1}(m-\omega) , 0 \right)\transp,
    \]
    with
    \[
        B(\omega, \phi) = i \omega \sigma_2  - (m - f(\pscal{\phi,\sigma_3 \phi }_{\C^2}))\sigma_1 \,.
    \]
    Here, $B(\omega, \phi)$ is of class $\mathcal{C}^1$ in $\omega$ and $\phi$ (as long as $\phi_1^2 - \phi_2^2$ is bounded away from zero, which is the case on any bounded interval), and thus $\omega \mapsto \phi_0(\omega,x)$ is continuous in $\omega$, uniformly for $x$ in bounded intervals.
    
    We now use the exponential decay from the previous proof. The pointwise continuity implies that $x^*(\omega, \epsilon)$ is bounded for $\omega \in [\epsilon, \sqrt{m^2-\epsilon^2}]$, since otherwise there would be $\omega^*$ such that $d_{\omega^*}(x) > s_\epsilon$ for all $x \geq 0$.  Thus, there exist $r_\epsilon$ and $A_\epsilon$ such that 
\[
	\norm{ \phi_0(\omega, x)}_{\C^2} \leq A_\epsilon \exp(-\epsilon |x|)\, \quad \text{ for all } \quad  \omega \in [\epsilon, \sqrt{m^2-\epsilon^2}] \textrm{ and } x \in \R.
\]
The restriction that $\omega \geq \epsilon$ comes from~\eqref{eq:comparable-components-apendix}.
   Combined with the uniform continuity on bounded intervals, this shows that $\omega \mapsto \phi_0(\omega)$ is continuous with values in $\mathcal{C}^0(\R)$.

\medskip
\noindent \emph{Differentiability.}   
We fix $\epsilon > 0$ and $\omega \in [\epsilon, \sqrt{m^2-\epsilon^2}]$. Assuming $\alpha \in [\epsilon, \sqrt{m^2-\epsilon^2}]$, substracting the groundstate equations for $\omega$ and $\alpha$ gives
    \begin{multline*}
       0 = \left( D_m - \frac{\omega + \alpha}{2}\Id \right)(\phi_{0}(\omega) - \phi_{0}(\alpha) ) -(\omega-\alpha)\frac{\phi_{0}(\omega) + \phi_{0}(\alpha) }{2} \\
       - (f\circ d_\omega) \sigma_3 \phi_{0}(\omega) + (f\circ d_\alpha) \sigma_3 \phi_{0}(\alpha)\,.
    \end{multline*}   
        For fixed $x \in \R$, we apply the mean value theorem to $f$ in order to find $\tilde d(x)$, between $d_\alpha(x)$ and $d_\omega(x)$, such that
    \begin{multline*}
    (f\circ d_\omega) \sigma_3 \phi_{0}(\omega) - (f\circ d_\alpha) \sigma_3 \phi_{0}(\alpha)\\
        \begin{aligned}
            &= f'(\tilde d)(d_\omega - d_\alpha )\sigma_3 \frac{\phi_{0}(\omega) + \phi_{0}(\alpha) }{2} + \frac{f\circ d_\omega + f\circ d_\alpha}{2}\sigma_3(\phi_{0}(\omega) - \phi_{0}(\alpha)) \\
            &= 2 \widetilde{Q}_\alpha(\phi_{0}(\omega) - \phi_{0}(\alpha)) + \frac{f\circ d_\omega + f\circ d_\alpha}{2}\sigma_3(\phi_{0}(\omega) - \phi_{0}(\alpha))\,,
        \end{aligned}
    \end{multline*}
    where we have defined
    \[
        \widetilde{Q}_\alpha := f'(\tilde d) \left(\sigma_3 \frac{\phi_{0}(\omega) + \phi_{0}(\alpha) }{2} \right) \left(\sigma_3 \frac{\phi_{0}(\omega) + \phi_{0}(\alpha) }{2} \right)\transp.
    \]
    Inserting this in the previous identity, we obtain 
    \begin{align*}
       (\omega-\alpha)\frac{\phi_{0}(\omega) + \phi_{0}(\alpha) }{2} &= \left( D_m - \frac{\omega + \alpha}{2}\Id - \frac{f\circ d_\omega + f\circ d_\alpha}{2} \sigma_3 - 2 \widetilde{Q}_\alpha\right)(\phi_{0}(\omega) - \phi_{0}(\alpha))\\
        &:= \tilde{L}_{2, \alpha} (\phi_{0}(\omega) - \phi_{0}(\alpha))\,.
    \end{align*}
    We claim that $\tilde{L}_{2, \alpha}$ converges to $L_2(\omega)$, in $\mathcal{B}(H^1(\R), L^2(\R))$, as~$\alpha$ tends to~$\omega$. 
        For the first two terms, this is clear. For the third term, we use 
    \[
    	\norm{(f\circ d_\omega - f\circ d_\alpha) \Psi}_{L^2} \leq  \norm{f\circ d_\omega - f\circ d_\alpha}_{L^\infty} \norm{\Psi}_{L^2}
    \]
    and the uniform continuity from the previous proof.
        For the last term, we use the uniform convergence on bounded intervals (where $d_\omega$ and $d_\alpha$ are bounded away from zero).
    In order to treat large values of $x$, we factorize out $\tilde{d}$ in the definition of $\widetilde{Q}_\alpha$:
    \[
        \widetilde{Q}_\alpha = \tilde{d} f'(\tilde{d}) \left(\sigma_3 \frac{\phi_{0}(\omega) + \phi_{0}(\alpha) }{2 \sqrt{\tilde{d}}} \right) \left(\sigma_3 \frac{\phi_{0}(\omega) + \phi_{0}(\alpha) }{2 \sqrt{\tilde{d}}} \right)\transp,
    \]
   where the factors in parenthesis are bounded in view of~\eqref{eq:comparable-components-apendix} and their pre-factor $\tilde{d} f'(\tilde{d})$ vanishes at infinity since $\lim_{s\to0^+} s f'(s) = 0$ by Assumption~\ref{Assumption_general_nonlinearity}.

  The convergence in $\mathcal{B}(H^1(\R), L^2(\R))$ implies norm resolvent convergence and, in particular, convergence of the spectrum.
  Recall that for each fixed $\omega$, $L_2(\omega)$ has zero as an isolated simple eigenvalue with an odd eigenfunction.
    Therefore, for sufficiently small~$\abs{\alpha -\omega}$, 
    \[
        \tilde{L}_{2, \alpha} : H^{1,{\rm even}}(\R) \mapsto L^{2,{\rm even}}(\R)
    \]
    is invertible and we conclude that 
    \[
        \frac{\phi_0(\omega)- \phi_0(\alpha)}{\omega-\alpha} = (\tilde{L}_{2, \alpha})^{-1} \frac{\phi_{0}(\omega) + \phi_{0}(\alpha) }{2}\,.
    \]
    This shows the differentiability of $\omega \mapsto \phi_0(\omega)$ as a function with values in $H^{1}(\R,\C^2)$, with $\partial_\omega\phi_0(\omega) = \left(L_2( \omega)\right)^{-1} \phi_0(\omega)$. This function is continuous in $\omega$ because $\left(L_2( \omega)\right)^{-1}$ is bounded on $L^{2,{\rm even}}(\R)$.
     Taking the inner product with $\phi_0(\omega)$ concludes the proof.
\end{proof}

\section{Proof of Lemma~\ref{technical_lemma}}\label{Appendix_proof_technical_lemma}
We start from~\eqref{Lemma_ineq_Re_z2_inequality}, where we introduce the parameter $\theta:=\frac{\mu\normt{Q}}{4(t+\omega)}$ for shortness, and study the function
\[
    g_\theta(\alpha,\eta) := (1-\alpha)\eta^2 + (1+\alpha) - \frac{4\eta \theta}{1+\alpha}\,.
\]
Since we do not know the value of~$\eta>0$,
we need to study the question: Given $\xi\geq0$, for which $\theta>0$ the inequality
\[
     \inf_{\eta>0}\max_{\alpha\in[0,1]} g_\theta(\alpha,\eta) \geq \xi
\]
is verified?
Notice that 
\[
    \partial_\alpha g_\theta(\alpha,\eta) = -\eta^2 + \frac{4 \theta}{(1+\alpha)^2}\eta + 1
\]
is a decreasing function of~$\alpha>0$ for any $\eta>0$. Thus,
\[
    \partial_\alpha g_\theta(\alpha,\eta) \geq \partial_\alpha g_\theta(1,\eta) = -\eta^2 + \theta\eta + 1\,.
\]
Defining $\eta_\star$ as the positive number such that $-\eta_\star^2 +\theta\eta_\star +1=0$, i.e.,
\[
    \eta_\star = \frac{ \theta + \sqrt{\theta^2 + 4}}{2} > 1\,,
\]
we have for all $\eta\in(0,\eta_\star]$ that
\[
    \partial_\alpha g_\theta(\alpha,\eta) \geq  -\eta^2 + \theta\eta + 1 \geq  -\eta_\star^2 + \theta\eta_\star + 1 = 0\,.
\]
Thus, for $\eta\in(0,\eta_\star]$, $\alpha\mapsto g_\theta(\alpha,\eta)$ is an increasing function on~$[0,1]$ and its maximum is $g_\theta(1,\eta) = 2(1 - \eta \theta)$.

Similarly, for any $\alpha>0$,
\[
    \partial_\alpha g_\theta(\alpha,\eta) \leq \partial_\alpha g_\theta(0,\eta) = -\eta^2 + 4 \theta\eta + 1\,.
\]
Defining $\eta_\circ$ as the positive number such that $-\eta_\circ^2 +4\theta\eta_\circ +1=0$, i.e., $\eta_\circ:=2\theta + \sqrt{4\theta^2+1}$, we have for all $\eta\geq\eta_\circ$ that
\[
    \partial_\alpha g_\theta(\alpha,\eta) \leq  -\eta^2 + 4\theta\eta + 1 \leq  -\eta_\circ^2 + 4\theta\eta_\circ + 1 = 0\,.
\]
Thus, for $\eta\geq\eta_\circ$, $\alpha\mapsto g_\theta(\alpha,\eta)$ is a decreasing function on~$[0,1]$ and its maximum is $g_\theta(0,\eta) = \eta^2 - 4\theta\eta + 1$.

Finally, for $\eta\in(\eta_\star,\eta_\circ)$, $\alpha\mapsto g_\theta(\alpha,\eta)$ is increasing then decreasing on~$[0,1]$, and reaches its maximum  at $\alpha_+:=\sqrt{\frac{4\eta \theta}{\eta^2-1}} - 1$ with value
\[
    h(\eta) := g_\theta(\alpha_+,\eta) = 2\left(\eta^2-2\sqrt{\theta}\sqrt{\eta(\eta^2-1)}\right)\,.
\]

Now, since $\eta\mapsto g_\theta(0,\eta) = \eta^2 - 4\theta\eta + 1\geq2$ on~$(\eta_\circ,+\infty)$, $\eta\mapsto g_\theta(1,\eta) = 2(1 - \eta \theta)$ is decreasing ($\theta>0$), $h(\eta_\star) = 2(1 - \eta_\star \theta)$ by construction, and (using the identity $\eta_\star^2= \theta \eta_\star + 1$)
\[
    h'(\eta_\star) = 2\left(2\eta_\star - \sqrt{\theta}\frac{3\eta_\star^2-1}{\sqrt{\eta_\star(\eta_\star^2-1)}}\right) = -2\theta<0\,,
\]
we conclude that
\[
    \inf_{\eta>0}\max_{\alpha\in[0,1]} g_\theta(\alpha,\eta) = \inf_{\eta_\star< \eta < \eta_\circ } h(\eta) = h(\eta_\theta) = 2\eta_\theta^2\frac{3-\eta_\theta^2}{3\eta_\theta^2 -1},
\]
where $\eta_\theta$ is defined as the unique real number in~$(1,+\infty)$ such that
\begin{equation}\label{Def_eta_theta}
    4\frac{\eta_\theta^3(\eta_\theta^2-1)}{(3\eta_\theta^2-1)^2} = \theta.
\end{equation}
Note that the l.h.s.\ of~\eqref{Def_eta_theta} being a strictly increasing function on~$(1,+\infty)$ ---hence one-to-one from $(1,+\infty)$ to $(0,+\infty)\ni \theta$---, it gives that~\eqref{Def_eta_theta} has a unique solution $\eta_\theta$ in~$(1,+\infty)$ and, on another hand that $\eta_\star<\eta_\theta<\eta_\circ$. Indeed, recalling that $1<\eta_\star<\eta_\circ$ and the equations they respectively solve, we have
\begin{align*}
        \eta_\star<\eta_\theta &\Leftrightarrow \theta > 4\frac{\eta_\star^3(\eta_\star^2-1)}{(3\eta_\star^2-1)^2} = \frac{4\eta_\star^4}{(3\eta_\star^2-1)^2}\theta \Leftrightarrow \eta_\star^2>1
    \intertext{and}
    \eta_\theta<\eta_\circ &\Leftrightarrow \theta < 4\frac{\eta_\circ^3(\eta_\circ^2-1)}{(3\eta_\circ^2-1)^2} = \frac{16\eta_\circ^4}{(3\eta_\circ^2-1)^2}\theta \Leftrightarrow 3\eta_\circ^2-1 < 4\eta_\circ^2\,.
\end{align*}
Moreover, as a by-product, the problem has a solution only for $\xi<2$, since
\[
    \inf_{\eta>0}\max_{\alpha\in[0,1]} g_\theta(\alpha,\eta) = \inf_{\eta_\star< \eta < \eta_\circ } h(\eta) = h(\eta_\star) = 2(1 - \eta_\star \theta) < 2\,.
\]

Summarizing, we have obtained $\eta_\theta>1$ defined by~\eqref{Def_eta_theta} for which
\[
    \inf_{\eta>0} \max_{\alpha\in[0,1]} g_\theta(\alpha,\eta) \geq \xi \quad \Leftrightarrow \quad h(\eta_\theta) \geq \xi.
\]

Now, keeping in mind that $\xi\in[0,2)$, that $\eta\mapsto 2\eta^2\frac{3-\eta^2}{3\eta^2 -1}$ is strictly decreasing and that~$\eta\mapsto 4\frac{\eta_\theta^3(\eta_\theta^2-1)}{(3\eta_\theta^2-1)^2}$ is stritly increasing, we have
\begin{align*}
    \inf_{\eta>0} \max_{\alpha\in[0,1]} g_\theta(\alpha,\eta) \geq \xi \Leftrightarrow \quad h(\eta_\theta) \geq \xi \quad &\Leftrightarrow \quad \eta_\theta \leq \frac{\sqrt{6-3\xi + \sqrt{9(2-\xi)^2+8\xi}}}{2}\\
    &\Leftrightarrow \quad \frac{\mu\normt{Q}}{4(t+\omega)}=:\theta \leq  \theta_+(\xi)\,,
\end{align*}
where $\theta_+$ is defined in~\eqref{General_nonlinearity_lower_bound_Re_z2_Def_theta_plus}, that is, $\theta_+:[0,2)\to(0,3\sqrt{3}/8]$ with
\[
    \theta_+(\xi):= 2\frac{\left(2 - 3\xi + \sqrt{9(2-\xi)^2+8\xi}\right) \left(6-3\xi + \sqrt{9(2-\xi)^2+8\xi}\right)^{\frac{3}{2}}}{\left(14 - 9 \xi + 3 \sqrt{9(2-\xi)^2+8\xi}\right)^2}\,. \tag*{\qed} 
\]

\section{Details on some computations.}

\subsection{\texorpdfstring{$L^\infty$}{L-infinity}-norms of terms involving \texorpdfstring{$\tilde{u}$}{u} and \texorpdfstring{$\tilde{v}$}{v}.}\label{Appendix_nonrelativistic_expansions}
We give here the details on computing several $L^\infty$-norms that we need in Section~\ref{Section_minmax_principle_and_Lmu}.

\textbf{Expansion of~$\norm{(\tilde{v}^2-\tilde{u}^2)^p}_\infty$}. We have
{\allowdisplaybreaks
\begin{align}\label{Nonrelativistic_expansion_v2_minus_u2_to_p}
    0 < (\tilde{v}^2-\tilde{u}^2)^p &= (p+1)(m-\omega)\frac{1-\tanh^2}{1-\nu\tanh^2} = (p+1)\frac{m^2-\omega^2}{m-\omega+2\omega\cosh^2} \nonumber\\
        &\leq (p+1)(m-\omega)
        \begin{aligned}[t]
            &=\norm{(\tilde{v}^2-\tilde{u}^2)^p}_\infty = (\tilde{v}^2-\tilde{u}^2)^p(0) \\
            &= \frac{p+1}{2m}\kappa^2 + \frac{p+1}{8m^3}\kappa^4 + O\!\left(\kappa^6\right).
        \end{aligned}
\end{align}
}%

\textbf{Expansion of~$\norm{(\tilde{v}^2-\tilde{u}^2)^{p-1}\tilde{u}\tilde{v}}_\infty$}. Defining $h_3(y):=\frac{y\left(1-y^2\right)}{\left(1-\nu y^2\right)^2}$, we have
\[
    (\tilde{v}^2-\tilde{u}^2)^{p-1}\tilde{u}\tilde{v} = (p+1) \sqrt{\nu} (m-\omega) h_3(\tanh)\,.
\]
Since $h_3$, defined on~$(-1,1)$, attains at $\pm\sqrt{y_3}$, with $y_3:=\frac{ - 3(1-\nu) + \sqrt{9(1-\nu)^2+4\nu}}{2\nu}$, its extrema $\pm\frac{\sqrt{y_3}(1-y_3)}{(1-\nu y_3)^2}$, we have the non-relativistic expansion
\begin{align}\label{Nonrelativistic_expansion_kappa3_term}
    \norm{(\tilde{v}^2-\tilde{u}^2)^{p-1}\tilde{u}\tilde{v}}_\infty &= (p+1) \sqrt{\nu} (m-\omega) \norm{h_3}_\infty \nonumber\\
        &=(p+1)(m-\omega)\frac{\sqrt{y_3}(1-y_3)}{(1-\nu y_3)^2} = \frac{p+1}{6\sqrt{3}m^2}\kappa^3\!\left(1+O\!\left(\kappa^2\right)\right).
\end{align}

\textbf{Expansion of~$\norm{(\tilde{v}^2-\tilde{u}^2)^{p-1}\tilde{u}^2}_\infty$}. Since on~$[0,1)$, $h_1(y):=\frac{1-y}{1-\nu y}\frac{\nu y}{1-\nu y}$ is nonnegative with a maximum $\frac{\nu}{4(1-\nu)} = \frac{m-\omega}{8\omega}$ at $\frac{1}{2-\nu}$, we have
\begin{multline}\label{Nonrelativistic_expansion_v2_minus_u2_to_p_minus_1_u2}
	0 \leq (v^2-u^2)^{p-1}u^2 = (p+1)(m-\omega)h_1\left(\tanh^2\right) \\
	\leq (p+1)\frac{(m-\omega)^2}{8\omega} = \norm{(\tilde{v}^2-\tilde{u}^2)^{p-1}\tilde{u}^2}_\infty = \frac{p+1}{32 m^3}\kappa^4\!\left(1+O\!\left(\kappa^2\right)\right).
\end{multline}

\textbf{Expansion of~$\norm{(\tilde{v}^2-\tilde{u}^2)^p - \frac{(p+1)\kappa^2}{2m\cosh^2}}_\infty$}. Since $h_2(y):= \frac{1}{m-\omega+2\omega y} - \frac{1}{2m y}$ is positive on~$(0,+\infty)$ with a maximum $\frac{1}{m}\frac{\sqrt{m}-\sqrt{\omega}}{\sqrt{m}+\sqrt{\omega}}$ at $\frac{1}{2}\left(1+\sqrt{\frac{m}{\omega}}\right)$, we have
{\allowdisplaybreaks
\begin{multline}\label{Nonrelativistic_expansion_reminder_type_1}
    0<(\tilde{v}^2-\tilde{u}^2)^p - \frac{(p+1)\kappa^2}{2m\cosh^2} = (p+1)\kappa^2h_2\left(\cosh^2\right) \\
    \leq (p+1) \frac{\kappa^2}{m}\frac{\sqrt{m}-\sqrt{\omega}}{\sqrt{m}+\sqrt{\omega}} 
        \begin{aligned}[t]
            &= \norm{(\tilde{v}^2-\tilde{u}^2)^p - \frac{(p+1)\kappa^2}{2m\cosh^2}}_\infty \\
            &= (p+1)\frac{\kappa^4}{8 m^3}\!\left(1+O\!\left(\kappa^2\right)\right).
        \end{aligned}
\end{multline}
}%

\subsection{Bounds on derivatives of $h \in F_k$}\label{Appendix_Bound_derivatives_Legendre}
For fixed $\mathfrak{s}$ defined in~\eqref{Def_s}, we denote by
\[
    p_j(x) = P_{\mathfrak{s}}^{\mathfrak{s}+1-j}(\tanh(p\kappa x)), \quad j = 1, \ldots, \lceil \mathfrak{s} \rceil\,,
\]
the eigenfunctions of the Schr{\"o}\-dinger operator 
\[
    -\frac{1}{2}\partial_x^2 + p^2 \kappa^2 V_{\textrm{PT}}(\mathfrak{s}, p\kappa x) = -\frac{1}{2}\partial_x^2 - p^2 \kappa^2 \frac{\mathfrak{s}(\mathfrak{s}+1)}{2\cosh^2(p\kappa x)}\,.
\]
They satisfy the eigenvalue equation
\[
    - p_j'' (x) = p^2 \kappa^2 \left( \frac{\mathfrak{s}(\mathfrak{s}+1)}{\cosh^2(p\kappa x)} - (\mathfrak{s}+1-j)^2 \right) p_j(x)\,.
\]
We want to bound the first and second derivatives of $h \in \vect\{p_1, \ldots, p_k\}$.
First,
\[
    \sup_{h \in \vect\{p_1, \ldots, p_k\} \setminus\{0\}} \frac{\norm{h'}}{\norm{h}} =\max_{j \in \{1, \cdots, k\}}  \frac{\norm{p_j'}}{\norm{p_j}} \quad \textrm{ and }  \quad \sup_{h \in \vect\{p_1, \ldots, p_k\} \setminus\{0\}} \frac{\norm{h''}}{\norm{h}} =\max_{j \in \{1, \cdots, k\}}  \frac{\norm{p_j''}}{\norm{p_j}}\,.
\]

For the first bound, we multiply the eigenvalue equation by $p_j$, integrate by parts in the first term and bound $\cosh^{-2}\leq1$ in order to obtain
\[
    \norm{p_j'}^2 = p^2 \kappa^2 \left( \pscal{p_j, \frac{\mathfrak{s}(\mathfrak{s}+1)}{\cosh^2(p\kappa \cdot)}p_j} - (\mathfrak{s}+1-j)^2\norm{p_j}^2 \right) \leq p^2 \kappa^2 \mathfrak{s}(\mathfrak{s}+1)\norm{p_j}^2.
\]

For the second bound, we take the norm on both sides of the eigenvalue equation and, since $0<\mathfrak{s}+1-j\leq \mathfrak{s}$ and $0<(\mathfrak{s}+1-j)^2\leq \mathfrak{s}^2<\mathfrak{s}(\mathfrak{s}+1)$ for $j\in\llbracket 1, \lceil \mathfrak{s} \rceil \rrbracket$, we obtain for all $j\in\llbracket 1, \lceil \mathfrak{s} \rceil \rrbracket$ that
\begin{multline*}
	\frac{\norm{p_j''}}{\norm{p_j}} \leq p^2 \kappa^2 \norm{\frac{\mathfrak{s}(\mathfrak{s}+1)}{\cosh^2(p\kappa \cdot)} - (\mathfrak{s}+1-j)^2}_\infty \\
    	= p^2 \kappa^2 \max\left\{\mathfrak{s}(\mathfrak{s}+1) - (\mathfrak{s}+1-j)^2, (\mathfrak{s}+1-j)^2 \right\} \leq p^2 \kappa^2 \mathfrak{s}(\mathfrak{s}+1)\,.
\end{multline*}    

    \addtocontents{toc}{\protect\setcounter{tocdepth}{1}}

\end{document}